\documentclass[11 pt, journal, onecolumn]{IEEEtran}
\usepackage{cite}
\usepackage{amsmath,amssymb,amsthm}
\usepackage{enumitem}
\usepackage{array}
\usepackage{algorithm}
\usepackage[compatible]{algpseudocode}
\newcolumntype{P}[1]{>{\centering\arraybackslash}p{#1}}
\newcolumntype{M}[1]{>{\centering\arraybackslash}m{#1}}

\usepackage{url}
\usepackage[ bookmarks=false,colorlinks=true,urlcolor=blue,citecolor=blue,linkcolor=blue]{hyperref}

\usepackage{cleveref}

\newtheorem{lem}{Lemma}
\newtheorem{thm}{Theorem}
\newtheorem{defn}{Definition}
\newtheorem{coro}{Corollary}

\newtheorem{rem}{Remark}

\crefname{equation}{}{}
\Crefname{equation}{}{}
\crefname{thm}{theorem}{theorems}
\Crefname{thm}{Theorem}{Theorems}
\crefname{clm}{claim}{claims}
\Crefname{clm}{Claim}{Claims}
\Crefname{coro}{Corollary}{Corollaries}
\Crefname{lem}{Lemma}{Lemmas}
\Crefname{sec}{Section}{Sections}
\crefname{app}{appendix}{appendices}
\Crefname{app}{Appendix}{Appendices}
\Crefname{part}{Part}{Parts}
\crefname{prop}{proposition}{propositions}
\Crefname{prop}{Proposition}{Propositions}
\Crefname{propty}{Property}{Properties}
\crefname{figure}{fig.}{fig.}
\Crefname{figure}{Fig.}{Fig.}
\crefname{defn}{definition}{definitions}
\Crefname{defn}{Definition}{Definitions}
\crefname{fact}{fact}{facts}
\Crefname{fact}{Fact}{Facts}
\crefname{app}{appendix}{appendices}
\Crefname{app}{Appendix}{Appendices}
\crefname{algo}{algorithm}{algorithms}
\Crefname{algo}{Algorithm}{Algorithms}
\crefname{algorithm}{algorithm}{algorithms}
\Crefname{algorithm}{Algorithm}{Algorithms}
\crefname{conj}{conjecture}{conjectures}
\Crefname{conj}{Conjecture}{Conjectures}
\crefname{obs}{observation}{observations}
\Crefname{obs}{Observation}{Observations}
\crefname{rem}{remark}{remarks}
\Crefname{rem}{Remark}{Remarks}

\ifCLASSINFOpdf
\else
\fi
\usepackage{amsmath}
\usepackage{bm}
\usepackage{amsfonts}
\usepackage{mathtools}
\usepackage{array}
\usepackage{mathtools}
\usepackage[normalem]{ulem}

\ifCLASSOPTIONcompsoc
 \usepackage[caption=false,font=normalsize,labelfont=sf,textfont=sf]{subfig}
\else
 \usepackage[caption=false,font=footnotesize]{subfig}
\fi
\begin{document}
\bstctlcite{IEEEexample:BSTcontrol}

\title{A Unified Coded Deep Neural Network Training Strategy Based on Generalized PolyDot Codes for Matrix Multiplication}
%
\IEEEoverridecommandlockouts
\author{Sanghamitra~Dutta$^*$, Ziqian~Bai$^*$,  Haewon Jeong, Tze Meng Low and Pulkit Grover\\ \small{$^*$ Equal Contribution}
\thanks{ S. Dutta, H. Jeong, T. M. Low and P. Grover are with Department of ECE at Carnegie Mellon University, PA, United States. Z. Bai is with Simon Fraser University, BC, Canada. [Corresponding Author Contact: sanghamd@andrew.cmu.edu]}
\thanks{This work was presented in part in \cite{dutta2018DNN2} at the IEEE International Symposium on Information Theory (ISIT) 2018.}
}

\maketitle
\begin{abstract}
This paper has two contributions. First, we propose a novel coded matrix multiplication technique called \emph{Generalized PolyDot codes} that advances on existing methods for coded matrix multiplication under storage and communication constraints. This technique uses ``garbage alignment,'' \textit{i.e.}, aligning computations in coded computing that are not a part of the desired output. Generalized PolyDot codes bridge between the recent Polynomial codes and MatDot codes, trading off between recovery threshold and communication costs. Second, we demonstrate that Generalized PolyDot coding can be used for training large Deep Neural Networks (DNNs) on unreliable nodes that are prone to soft-errors, e.g., bit flips during computation that produce erroneous outputs. This requires us to address three additional challenges: (i) prohibitively large overhead of coding the weight matrices in each layer of the DNN at each iteration; (ii) nonlinear operations during training, which are incompatible with linear coding; and (iii) not assuming presence of an error-free master node, requiring us to architect a fully decentralized implementation. Because our strategy is completely decentralized, \textit{i.e.}, no assumptions on the presence of a single, error-free master node are made, we avoid any ``single point of failure.'' We also allow all primary DNN training steps, namely, matrix multiplication, nonlinear activation, Hadamard product, and update steps as well as the encoding and decoding to be error-prone. We consider the case of mini-batch size $B=1$, as well as $B>1$; the first leverages coded matrix-vector products, and the second coded matrix-matrix products, respectively. The problem of DNN training under soft-errors also motivates an interesting, probabilistic error model under which a real number $(P,Q)$ MDS code is shown to correct $P-Q-1$ errors with probability $1$ as compared to $\lfloor \frac{P-Q}{2} \rfloor$ for the more conventional, adversarial error model. We also demonstrate that our proposed coded DNN strategy can provide unbounded gains in error tolerance over a competing replication strategy and a preliminary MDS-code-based strategy \cite{dutta2018DNN1} for both these error models. Lastly, as an example, we demonstrate an extension of our technique for a specific neural network architecture, namely, sparse autoencoders.

\end{abstract}




%
\IEEEpeerreviewmaketitle


\section{Introduction}
\label{sec:introduction}

Ever-increasing data and computing requirements increasingly require massively distributed and parallel processing. However, as the number of parallel processing units are scaling, the expected number of faults, errors or delays in computing are also scaling\cite{dean2013tail,ferreira2011evaluating,bergman2008exascale,geist2016supercomputing, ziegler1996terrestrial,soft-error}. Thus, one of the major challenges of large-scale computing today is ensuring ``reliability at scale''.  \emph{Coded computing} has emerged as a promising solution to the various problems arising from the unreliability of processing nodes in parallel and distributed computing, such as straggling delays\cite{dean2013tail} or ``soft-errors''\cite{geist2016supercomputing, ziegler1996terrestrial,soft-error}. It is, in fact, a significant step in a long line of work on noisy computing started by von Neumann\cite{von1956probabilistic} in 1956, that has been followed upon by \textit{Algorithm-Based Fault Tolerance} (ABFT)\cite{ABFT1984,ABFT1988FFT,ABFT1998generalized,ABFT2015matrix}, the predecessor of coded computing. Recent advances in coded computing~\cite{NewsletterPaper,gauri2014delay,gauri2015straggler,gauri2014efficient,lee2016speeding,lee2018speeding,dutta2016short,dutta2017coded,lee2017matrix,yu2017polynomial,GC1,GC2,allerton17,yu2018entangled,li2015coded,GC3,azian2017consensus,yang2016logistic,yang2017encoded,yang2016convolution,yang2016encoded,Salman1,Salman2,Salman3,Salman4,Emina1,Emina2,Virtualization,heterogeneousclusters,GC4,Suhas1,Suhas2,yang2017NIPS,Ramtin1,lee2017multicore,jeongFFT,baharav2018straggler,suh2017matrix,mallick2018rateless, wang2018coded, wang2018fundamental,severinson2018block,ye2018communication,haddadpour2018codes,haddadpour2018straggler,yang2018ISIT,ferdinand2016anytime,ferdinand2018hierarchical,song2017pliable,kosaian2018learning,seth2018bigdata,jeong2018locally} have generated significant interest both within and outside information theory. 

The problem of distributed matrix-matrix multiplication $\bm{S}=\bm{W}\bm{X}$ under storage constraints, \textit{i.e.}, when each node is allowed to store a fixed $\frac{1}{K}$ fraction of each of the matrices $\bm{W}$ and $\bm{X}$, has been of considerable interest in the coded computing community. In \cite{lee2017matrix}, the authors introduced Product codes that use two different Maximum Distance Separable (MDS)~\cite{ryan2009channel} codes to encode the two matrices being multiplied. This work was followed upon by \cite{yu2017polynomial}, where the authors proposed Polynomial codes, that use a polynomial-based encoding scheme for the storage-constrained matrix multiplication problem to achieve a lower \textit{recovery threshold}, \textit{i.e.}, the number of processing nodes to wait for out of a total of $P$ nodes. In our prior work~\cite{allerton17}, we first demonstrated that the recovery threshold for the storage-constrained matrix multiplication problem can be reduced further (beyond Polynomial codes~\cite{yu2017polynomial}) in scaling sense by using novel constructions called \emph{MatDot codes}. When a fixed $\frac{1}{K}$ fraction of each matrix can be stored at each node, MatDot achieves a recovery threshold of $2K-1$ as compared to Polynomial codes which achieve a threshold of $K^2$, albeit at a higher communication cost. In fact, MatDot codes can be proved to be optimal for storage-constrained matrix-matrix multiplication using fundamental limits in \cite{yu2018entangled}. In the same work \cite{allerton17}, we also proposed the PolyDot codes for matrix-matrix multiplication that interpolate between Polynomial codes (for low communication costs) and MatDot codes (for lowest recovery threshold), trading off recovery threshold and communication costs, as illustrated in \Cref{fig:RecThres}.

\begin{figure}
\centering
\includegraphics[height=5cm]{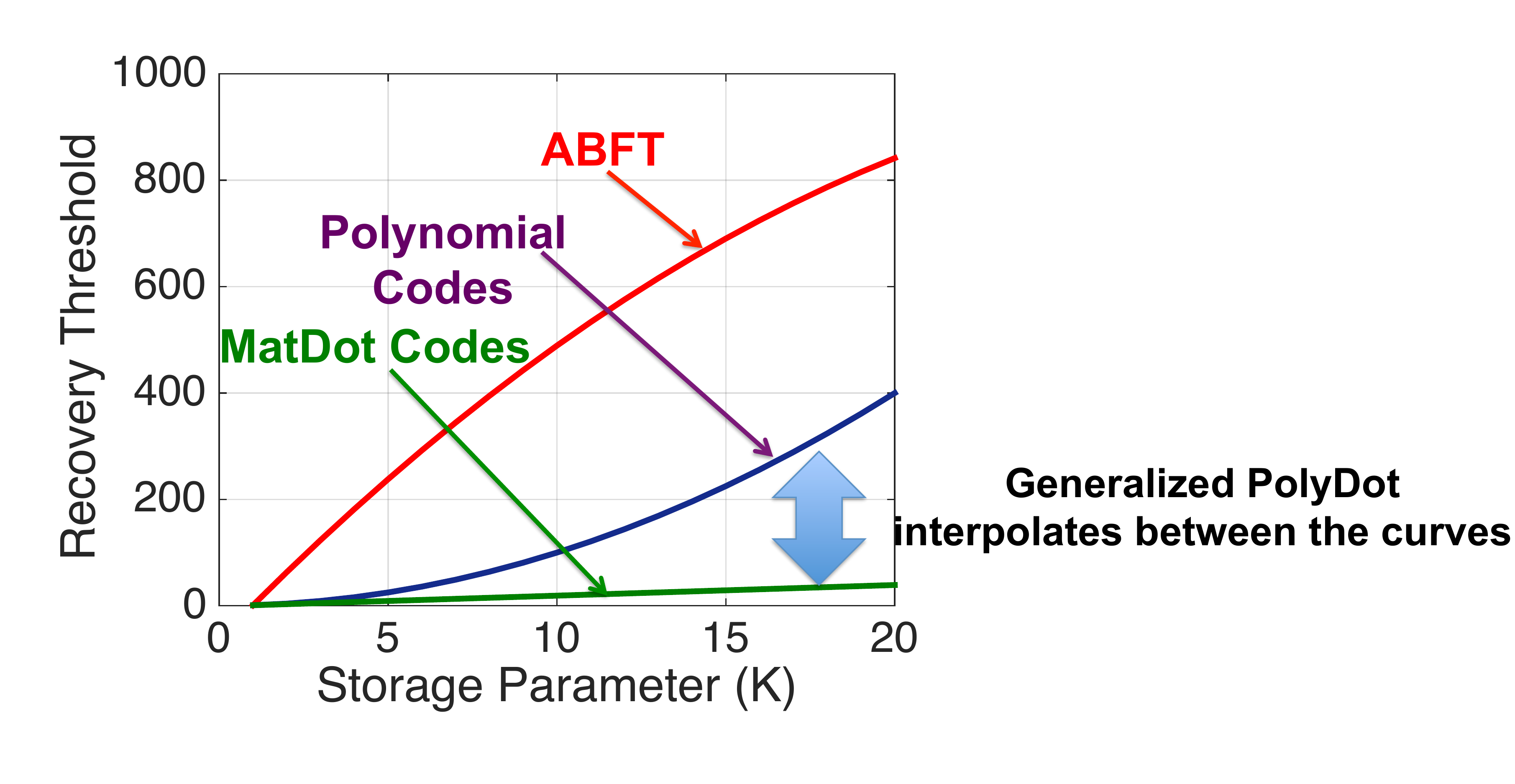}
\caption{An illustration showing how the recovery threshold scales with storage parameter $K$, \textit{i.e.}, when each node can store a fraction $1/K$ of each of the matrices being multiplied. Total number of nodes is $P=1000$. Generalized PolyDot codes interpolates between MatDot codes and Polynomial codes, trading off recovery threshold and communication costs. \vspace{-0.5cm}}
\label{fig:RecThres}
\end{figure}
In this work, we have two main contributions as follows:\\
 \textbf{1. Generalized PolyDot Codes:}
First we build upon our prior work on PolyDot codes~\cite{allerton17} and propose a new class of codes for distributed matrix-vector and matrix-matrix multiplication called Generalized PolyDot codes. The proposed Generalized PolyDot codes interpolate better between Polynomial codes and MatDot codes, improving the recovery threshold (see \Cref{thm:gen_poly_MV,thm:gen_poly_MM}) by introducing the new idea of ``garbage alignment'' for this problem. Garbage alignment is essentially a clever substitution in the multivariate polynomial introduced in our prior work on PolyDot framework of matrix multiplication~\cite{allerton17} that aligns some of the unwanted coefficients and reduces the number of unknowns during polynomial interpolation, as we elaborate in \Cref{sec:motivating_examples}. The recovery threshold of PolyDot codes are improved by a factor of $2$ using Generalized PolyDot codes by replacing the bijection-based substitution with garbage alignment. We note that, a concurrent work~\cite{yu2018entangled}, that in fact appeared at the same venue as the original publication of this work~\cite{dutta2018DNN2}, also achieve the same recovery threshold as Generalized PolyDot codes through slightly different routes.


\noindent \textbf{2. Coded DNN Training Strategy:} Our next contribution is that we develop a unified coded computing strategy, by appropriately utilizing Generalized PolyDot codes, for the training of \textit{model-parallel}\footnote{\textit{Data parallel} and \textit{model parallel} are two different architectures for DNN training. In data parallelism~\cite{dutta2018slow}, different nodes store and train a different replica of the entire DNN on different pieces of data, and a central parameter server combines inputs from all the nodes to train a central replica of the DNN. In model parallelism, different parts of a single DNN are parallelized across multiple nodes. Coding for data parallel training is examined in \cite{GC1,GC2,GC3,GC4}.} Deep Neural Networks in presence of \emph{soft-errors}. Soft-errors~\cite{geist2016supercomputing, ziegler1996terrestrial} refer to undetected errors, e.g., bit-flips or gate errors in computation, that can corrupt the end result\footnote{Ignoring soft-errors entirely during training of DNNs can severely degrade the accuracy of training, as we experimentally observe in~\cite{dutta2018DNN1}.}.  \textcolor{black}{For this problem, we consider two kinds of error models: an adversarial model and a probabilistic model. Interestingly, as we show in \Cref{thm:main}, under the probabilistic model a real number $(P,Q)$ MDS code can theoretically correct $P-Q-1$ errors with probability $1$ as compared to $\lfloor \frac{P-Q}{2} \rfloor$ for the more conventional, adversarial case. While ideas of correcting more errors than $\frac{P-Q}{2}$ have been prevalent in finite fields (see \cite{guruswami2007algorithmic}), to the best of our knowledge this seems to be the first result of this nature for real number error correction using MDS codes (see \cite{yang2016logistic} for similar results on LDPC codes). 
}


The problem of coded DNN training under soft-errors also has several additional challenges (see~\cite{dutta2018DNN1}), that are all addressed by our proposed unified strategy, as follows: 


\begin{itemize}
\item \textbf{Prohibitive overhead of encoding matrix $\bm{W}$ at each iteration:}  Existing works on coded matrix multiplication (e.g.~for computing $\bm{W}_{N\times N} \bm{X}_{N \times B}$ where $N \gg B$) require encoding of both the matrices $\bm{W}$ and $\bm{X}$. Because $N \gg B$, encoding $\bm{W}$ is computationally expensive (complexity $\Theta(N^2P)$) and in fact can be as expensive as the matrix multiplication itself ($\Theta(N^2B)$). Thus, for this problem, the coding techniques in~\cite{dutta2016short,lee2018speeding, lee2016speeding,GC1,GC2,lee2017matrix,allerton17,yu2017polynomial} would be helpful if $\bm{W}$ is known in advance and is fixed over a large number of computations so that the encoding cost is amortized. However, when training DNNs, because the parameter matrices \textit{update} at every iteration, a naive extension of existing techniques would require encoding of parameter matrices at every iteration and thus introduce an undesirable additional overhead of $\Theta(N^2P)$ at every iteration that can no longer be amortized. To address this, we carefully weave Generalized PolyDot codes into the operations of DNN training so that an initial encoding of the weight matrices is maintained \emph{across the updates at each iteration}. To do so, at each iteration each node locally encodes much \textit{smaller matrices consisting of $NB$ elements instead of the large matrix $\bm{W}$ of $N^2$ elements}, adding negligible overhead. In particular, for the case of $B=1$, this simply reduces to encoding \textit{vectors instead of matrices} which is much cheaper in terms of computational complexity.

\item \textbf{Master node acting as a single point of failure:} Because of our focus on soft-errors in this work, if we allow the architecture to use a master node, this node can often become a ``single point of failure'' (considered undesirable in parallel computing literature, e.g.,~\cite{nelson1990fault}). Thus,  we consider a completely decentralized setting, with no master node. In that spirit, our  strategy allows encoding/decoding to be error-prone~\cite{Tay_Bel_68} as well, along with all the other primary steps, namely, matrix multiplication, nonlinear activation, Hadamard product and update. We only introduce two verification steps (that check for decoding errors by exchanging some values among all nodes and comparing) that are extremely low complexity\footnote{The longer a computation, the more is the probability of soft-errors~\cite{li2007memory}. In fact, the number of soft-errors that occur within a time interval is often modelled as a Poisson random variable with mean proportional to the length of the time interval.} and hence, may be assumed to be error-free.


\item \textbf{Nonlinear activation between layers:} The nonlinear activation (e.g.~sigmoid, ReLU) between layers also poses a difficulty for coded training because most coding techniques are linear. 
To circumvent this issue, we code the linear operations (matrix multiplication of complexity $\Theta(N^2B)$) at each layer separately. Matrix multiplication and update are \textit{the} most critical and complexity-intensive steps in the training of DNNs as compared to other operations such as nonlinear activation or Hadamard product which are of complexity $\Theta(NB)$, and hence are also more likely to have errors. Moreover, as our implementation is decentralized, every node acts as a low-complexity functional replica of the master node, performing encoding/decoding/nonlinear activation/Hadamard product and helping us detect (and if possible correct) errors in all the primary steps, including the nonlinear activation step.
\end{itemize}

\textit{Overview of results in coded DNN Training:} We show (in \Cref{thm:error_tolerance_MV,thm:error_tolerance_MM}) that under both the adversarial and probabilistic error models, the coded DNN strategy using Generalized PolyDot codes improves the error tolerance in scaling sense over competing replication strategy and a preliminary MDS-code-based DNN training strategy\cite{dutta2018DNN1}. Moreover, to demonstrate the utility of the DNN training strategy, we also show in \Cref{thm:complexity_MV,thm:complexity_MM} that the additional overhead due to coding per iteration is negligible as compared to the computational complexity of the local matrix operations at each node as long as the number of processors $P^4=o(N)$. Our coding technique also extends to DNN training with regularization as is done more commonly in practice. 

\subsection{Why are Generalized PolyDot codes a natural choice for coded DNN training?}
For a distributed matrix-matrix multiplication problem $\bm{S}=\bm{W}\bm{X}$, Polynomial codes~\cite{yu2017polynomial} use a horizontal splitting of the first matrix $\bm{W}$ and vertical splitting of the second matrix $\bm{X}$ into $K$ blocks each, thus satisfying the storage constraint (see~\Cref{fig:Poly_Mat_tradeoff}). Alternately, MatDot codes~\cite{allerton17} use a vertical splitting of the first matrix $\bm{W}$ and a horizontal splitting of the second matrix $\bm{X}$ into $K$ blocks each. PolyDot~\cite{allerton17} and Generalized PolyDot codes both incorporate simultaneous vertical and horizontal splitting of both matrices $\bm{W}$ and $\bm{X}$ to interpolate between Polynomial codes and MatDot codes, while satisfying storage constraints. \textbf{Interestingly, as it turns out, the problem of coded computing by splitting the matrix $\bm{W}$ both vertically and horizontally arises naturally in DNN training.}

In the problem of DNN training, at each iteration in any particular layer, the same matrix $\bm{W}$ is required to be multiplied once with another matrix $\bm{X}$ (or vector $\bm{x}$) from the right side in the feedforward stage, \textit{i.e.}, $\bm{W}\bm{X}$, and once with another matrix $\bm{\Delta}^T$ (or vector $\bm{\delta}^T$) from the left side in the backpropagation stage, \textit{i.e.}, $\bm{\Delta}^T\bm{W}$. One would like to use the same encoding on $\bm{W}$ (or sub-matrices of $\bm{W}$) for both the matrix-matrix multiplication because the available storage is limited and storing two encoded sub-matrices of $\bm{W}$ for the two matrix-matrix multiplications is expensive. Now, suppose that we choose to use MatDot codes for $\bm{W}\bm{X}$ and hence stick with vertical partitioning of the first matrix $\bm{W}$. Then, we would have to use Polynomial codes for $\bm{\Delta}^T\bm{W}$ as $\bm{W}$ is the second matrix for this multiplication. This is also pictorially illustrated in \Cref{fig:Poly_Mat_tradeoff}. Thus, we would like to allow for both horizontal and vertical partitioning of the matrix $\bm{W}$ into a grid of $m \times n$ sub-matrices with $mn=K$ for the storage constraint. This would allow us to be able to interpolate between MatDot and Polynomial codes for the two matrix-matrix multiplications $\bm{W}\bm{X}$ and $\bm{\Delta}^T\bm{W}$, so as to achieve a good recovery threshold (and hence error tolerance) for both forward and backward matrix-matrix multiplications.

\begin{figure}
\centering
\subfloat[For the matrix multiplication $\bm{S}=\bm{W}\bm{X}$, Polynomial codes and MatDot codes split the matrix $\bm{W}$ in different ways. We would like to interpolate between the two strategies by allowing for both horizontal and vertical splitting, while satisfying the storage constraint.]{
\includegraphics[height=3.8cm]{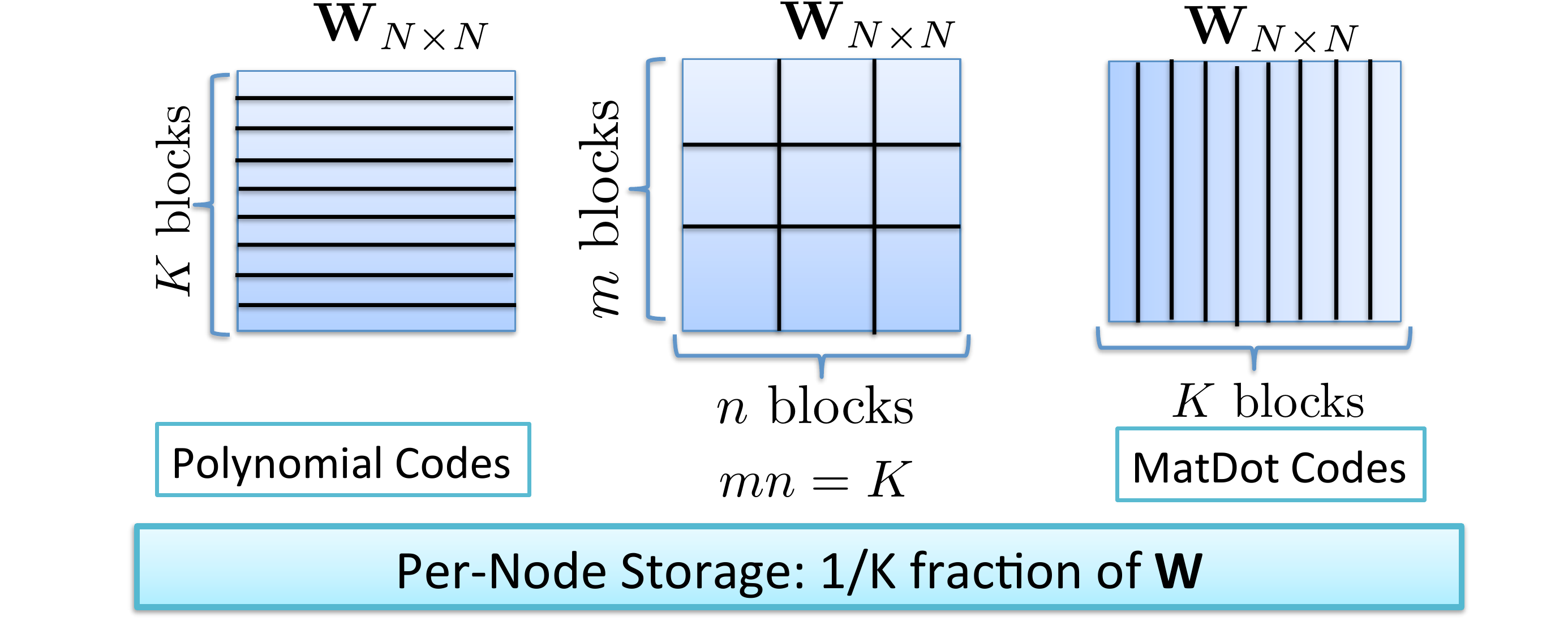} }
\hspace{0.5cm}
\subfloat[The same matrix $\bm{W}$ is multiplied with another matrix once from the right side in feedforward stage ($\bm{W}\bm{X}$) and once from the left side in the backpropagation stage ($\bm{\Delta}^T\bm{W}$) in DNN training. ]{ \includegraphics[height=3.5cm]{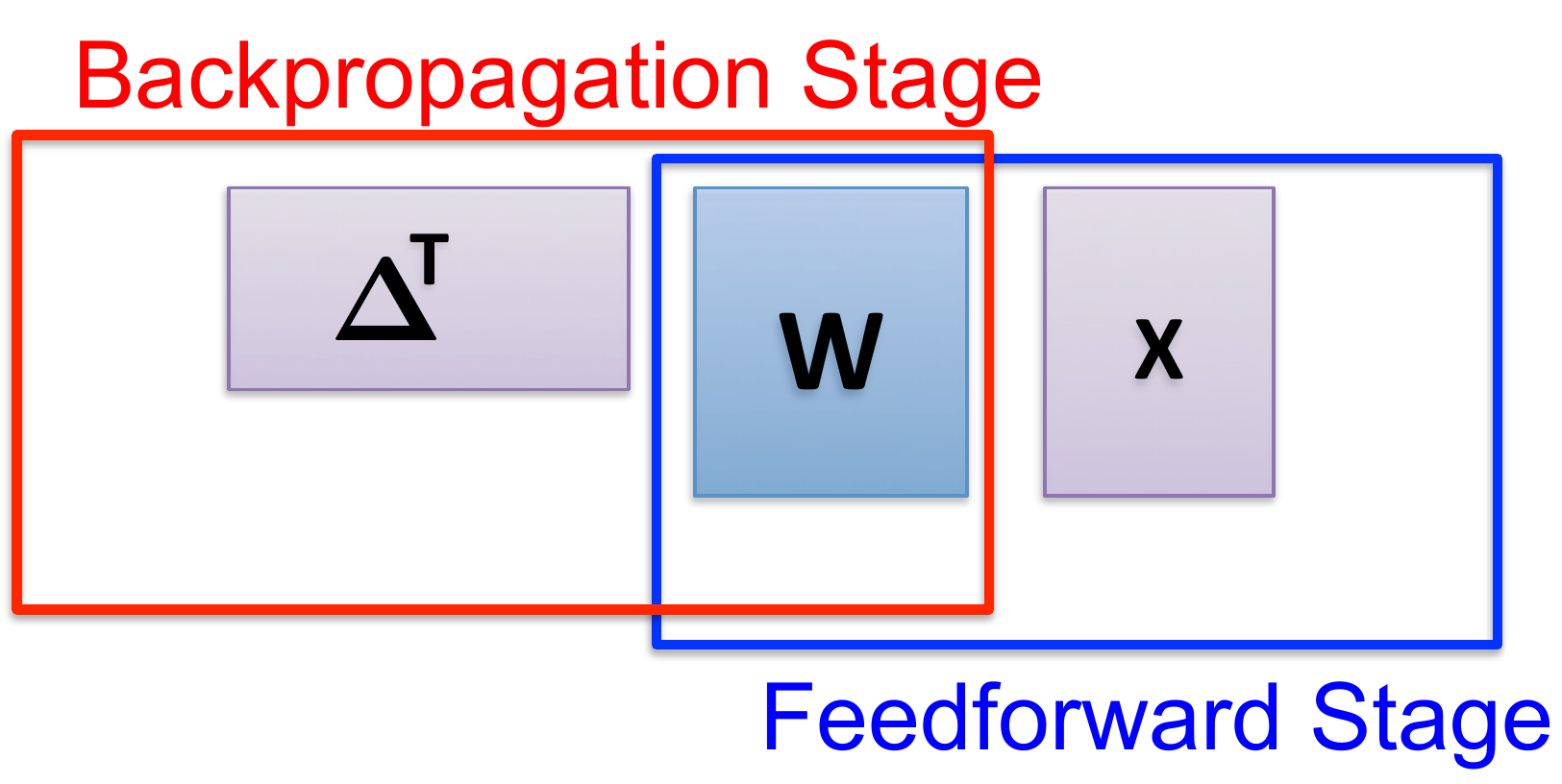} }
\caption{Interpolation between Polynomial codes and MatDot codes:  the former considers only horizontal splitting of the first matrix $\bm{W}$ while the latter considers only the vertical splitting of $\bm{W}$. To meet the storage constraint, we are required to use the same encoding of $\bm{W}$ so that we are able to perform both $\bm{W}\bm{X}$ and $\bm{\Delta}^T\bm{W}$. Suppose, we decide to stick to only vertical partitioning of $\bm{W}$. This would result in MatDot codes for $\bm{W}\bm{X}$ but Polynomial codes for $\bm{\Delta}^T\bm{W}$. Thus, we would like a tradeoff between the two strategies, allowing for both horizontal and vertical partitioning into $m \times n$ blocks, such that $mn=K$ so that the storage constraint is satisfied. }
\label{fig:Poly_Mat_tradeoff}
\end{figure}

We note that a concurrent work~\cite{yu2018entangled}, that in fact appeared at the same conference as the original publication of this work~\cite{dutta2018DNN2}, proposes the coding scheme ``Entangled Polynomial codes'' that achieve the same recovery threshold as Generalized PolyDot codes through slightly different routes for the problem of distributed matrix multiplication under storage constraints allowing for both vertical and horizontal splitting. In this expanded version of \cite{dutta2018DNN2} we show that the Generalized PolyDot codes, that were originally proposed for matrix-vector products in \cite{dutta2018DNN2}, naturally extend to matrix-matrix products as well. This is because Generalized PolyDot codes are simply a clever substitution in a multivariate polynomial in our prior PolyDot framework for matrix-matrix multiplication~\cite{allerton17} (this work~\cite{allerton17} precedes both \cite{yu2018entangled} and \cite{dutta2018DNN2}), so that some unwanted coefficients of the polynomial align with each other, reducing the degree of the polynomial and hence number of unknowns in polynomial interpolation, thereby improving the recovery threshold by a factor of $2$. More importantly, \textcolor{black}{our work also introduces a novel result in real-number error correction} and is also the first line of work that considers the problem of coded DNN training.

\subsection{Organization.}
\textcolor{black}{The rest of the paper is organized as follows. We introduce our two problem formulations, namely, (1) coded matrix multiplication; and (2) coded DNN, in \Cref{sec:problem_formulation}. First, we address Problem $1$, \textit{i.e.}, the coded matrix multiplication problem. For this problem, we introduce some motivating examples in \Cref{sec:motivating_examples} and then describe the Generalized PolyDot codes in detail in \Cref{sec:gen_poly_codes}. Then moving on to Problem $2$, we first elaborate upon some modelling assumptions, e.g., the two error models for the coded DNN problem in \Cref{sec:background}, which leads to a novel result on real number error correction. This is followed by possible solutions for the coded DNN problem for $B=1$ (existing strategies in \Cref{sec:existing_strategies}, our proposed strategy in \Cref{sec:coded_DNN_MV}). We formally analyze the error tolerance and the computation and communication costs of our strategy in \Cref{sec:comparison_DNN_MV}. Next, we extend the proposed strategy for the case of $B>1$ in \Cref{sec:extension}. Finally, in \Cref{sec:coded_autoencoder}, we discuss an application of our strategy to more recent but closely related neural network architectures, namely, sparse autoencoders.}

\section{System Models and Problem Formulations for the two problems}
\label{sec:problem_formulation}


\subsection{Problem $1$ (Coded Matrix Multiplication).}

\textbf{System Model:} We assume that there is a centralized, reliable master node and $P$ memory-constrained worker nodes that may be unreliable. The master node allocates computational tasks to the worker nodes. The worker nodes perform their computations in parallel and send their outputs back to the master node. Outputs of some worker nodes may be modeled as erasures (e.g., because of straggling or faults). The master node gathers the outputs of the worker nodes (possibly a subset), and uses them to compute the final result. It is desirable that the computational overhead of the master node as well as the communication costs should be smaller than the local computational complexity of each worker node after parallelization. 

\textbf{Problem Formulation:} Compute distributed matrix multiplication $\bm{S}=\bm{W}\bm{X}$ using $P$ worker nodes prone to erasures such that each node can store only a fixed fraction $\frac{1}{K}$ of matrix $\bm{W}$ and $\frac{1}{K'}$ of matrix $\bm{X}$. For this problem, our goal is to minimize the \emph{erasure recovery threshold}, \textit{i.e.}, the number of nodes that the decoder has to wait for out of the total $P$ nodes, to be able to compute the entire final result. We are also interested in studying tradeoffs between communication costs and recovery threshold.
We assume that both $K$ and $K'$ are less than $P$.

\begin{rem}
When $K=K'$, our prior work on MatDot codes~\cite{allerton17} achieves the optimal recovery threshold of $2K-1$ under the storage constraint, albeit at a high communication cost. Here, we propose a more general encoding strategy for the problem of coded matrix multiplication under fixed storage constraints that also allows us to tradeoff between recovery threshold and communication costs, with MatDot codes and Polynomial codes being two special cases.
\end{rem}

\begin{rem}
Note that no specific assumptions are made on the relative dimensions of $\bm{W}$ and $\bm{X}$ for this problem formulation. 
The matrices $\bm{W}$ and $\bm{X}$ do not necessarily have to be square matrices either.
\end{rem}

\subsection{Problem $2$ (Coded DNN).}

\textbf{Background}: 
Before introducing our system model and problem formulation, we first introduce the main computational steps at each iteration of classical (error-free) DNN training using Stochastic Gradient Descent (SGD)\footnote{As a first step in this direction of coded neural networks, we assume that the training is performed using vanilla SGD. As a future work, we plan to extend these coding ideas to other training algorithms~\cite{ruder2016overview} such as momentum SGD, Adam etc.}. For a more detailed introduction, we refer to the seminal work \cite{rumelhart1988learning} or \Cref{appendix:DNN_background}. 

A DNN consists of $L$ weight matrices (also called parameter matrices), $\bm{W}^l$, of dimensions $N_l \times N_{l-1}$ where $l$ denotes the layer-index and $N_l$ denotes the number of neurons at layer $l$, for $l=1,2,\ldots,L$. These weight matrices 
are updated at each iteration of training based on a ``mini-batch'' of $B$ data-points and their labels. Because we are primarily interested in large models that require parallelization across multiple nodes, we assume $N_l\gg B$ for all $l$. For simplicity of presentation, we also assume $N_l=N$ for all $l$, and thus $N \gg B$. DNN training has $3$ stages in each iteration, (i) the feedforward stage, (ii) the backpropagation stage, and (iii) the update stage. As the operations are similar and repeat across all the layers (see \Cref{appendix:DNN_background}), we limit our discussion to layer $l$. The operations during feedforward stage (see \Cref{fig:DNN_step1}) can be summarized as:

\noindent From layer $l = 1$ to $L$, 
\begin{itemize}[noitemsep,topsep=0pt,leftmargin=0.2cm]
\item $[$Step $O1]$ Compute matrix-matrix product $\bm{S}^l=\bm{W}^l\bm{X}^l$ where $\bm{X}^l$ is a matrix of dimension $N \times B$. 
\item $[$Step $C1]$ Compute $\bm{X}^{(l+1)}=f(\bm{S}^l)$ where $f(\cdot)$ is a nonlinear activation function applied element-wise.
\end{itemize}
At the last layer ($l=L$), the backpropagated error matrix is generated by accessing the true label matrix from memory and the estimated label matrix as output of last layer (see \Cref{fig:DNN_step2}). Then, the backpropagated error propagates from layer $L$ to $1$ (see \Cref{fig:DNN_step3}), also updating the weight matrices at every layer alongside (see \Cref{fig:DNN_step4}). The operations for the backpropagation stage can be summarized as:\\
\noindent From layer $l = L$ to $1$,
\begin{itemize}[noitemsep,topsep=0pt,leftmargin=0.2cm]
\item $[$Step $O2]$ Compute matrix-matrix product $[\bm{C}^l]^T=[\bm{\Delta}^l]^T\bm{W}^l$ where $[\bm{\Delta}^l]^T$ is a matrix of dimensions $B \times N$. 
\item $[$Step $C2]$ Compute Hadamard product $[\bm{\Delta}^{l-1}]^T = [\bm{C}^l]^T \circ g([\bm{X}^l]^T)$ where $g(\cdot)$ is another function applied element-wise (more specifically $\frac{df(u)}{du}=g(f(u))$ for the chosen nonlinear activation function $f(u)$) and the Hadamard product ``$\circ$'' between two matrices of the same dimensions is another matrix of those dimensions, such that its elements are element-wise products of the corresponding elements of the operands.

\end{itemize}
Finally, the step in the Update stage is as follows: \\
For all layers $l$,
\begin{itemize}[noitemsep,topsep=0pt,leftmargin=0.2cm]
\item  $[$Step $O3]$ Update matrix $\bm{W}^l$ as follows: $\bm{W}^l \leftarrow \bm{W}^l + \eta \bm{\Delta}^l[\bm{X}^l]^T$ where $\eta$ is the learning rate. Sometimes a regularization term is added with the loss function in DNN training (elaborated in \Cref{subsec:app_reg}). For L2 regularization, the update rule is modified as: $\bm{W}^l \leftarrow (1-\eta \lambda)\bm{W}^l + \eta \bm{\delta}^l[\bm{x}^l]^T$ where $\eta$ is the learning rate and $\frac{\lambda}{2}$ is the regularization constant.
\end{itemize}

\begin{figure*}[!ht]
\centering
\subfloat[Feedforward stage]{\fbox{\includegraphics[height=4.5cm]{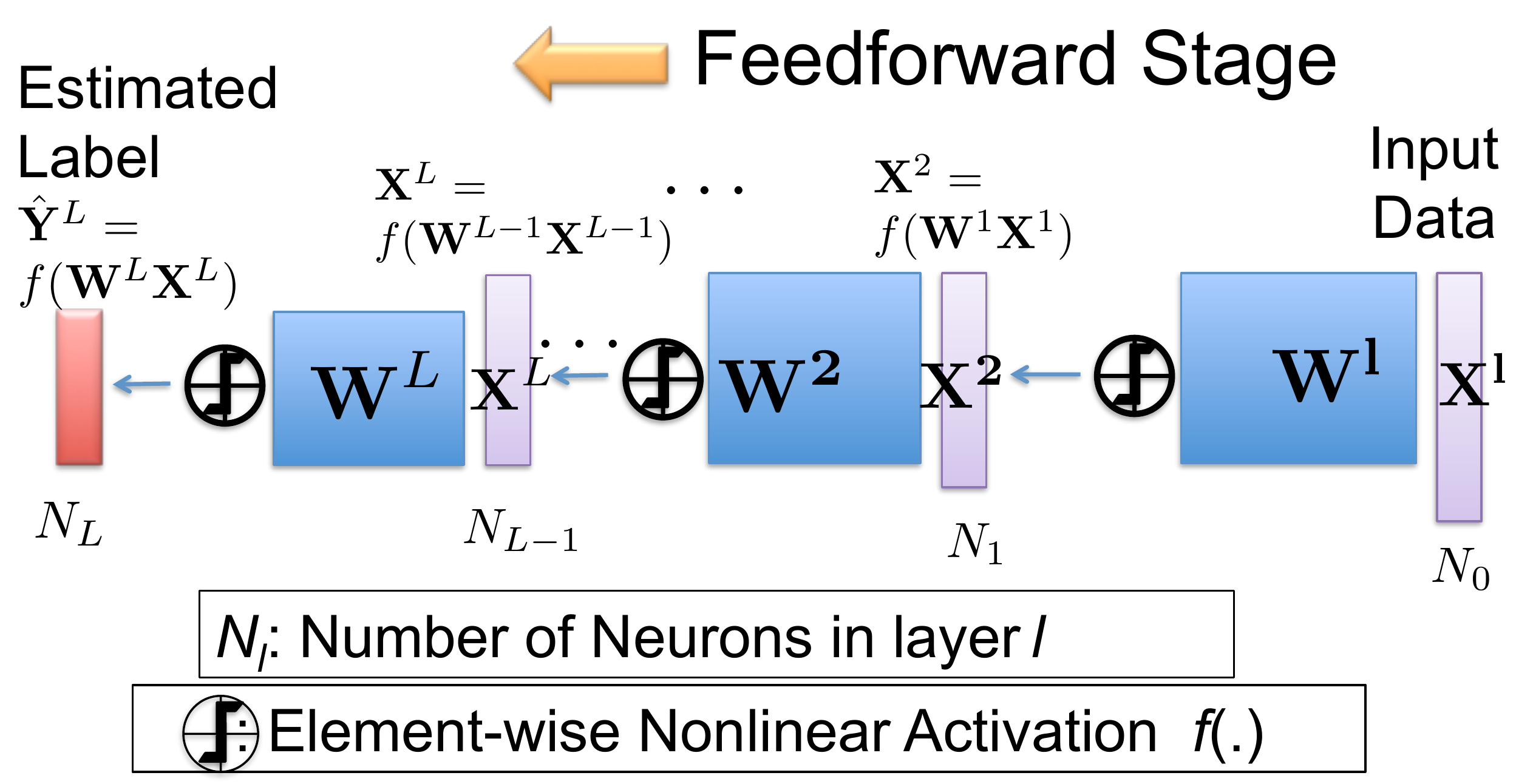}
\label{fig:DNN_step1}}} \quad
\subfloat[Transition at last layer]{\fbox{\includegraphics[height=4.5cm]{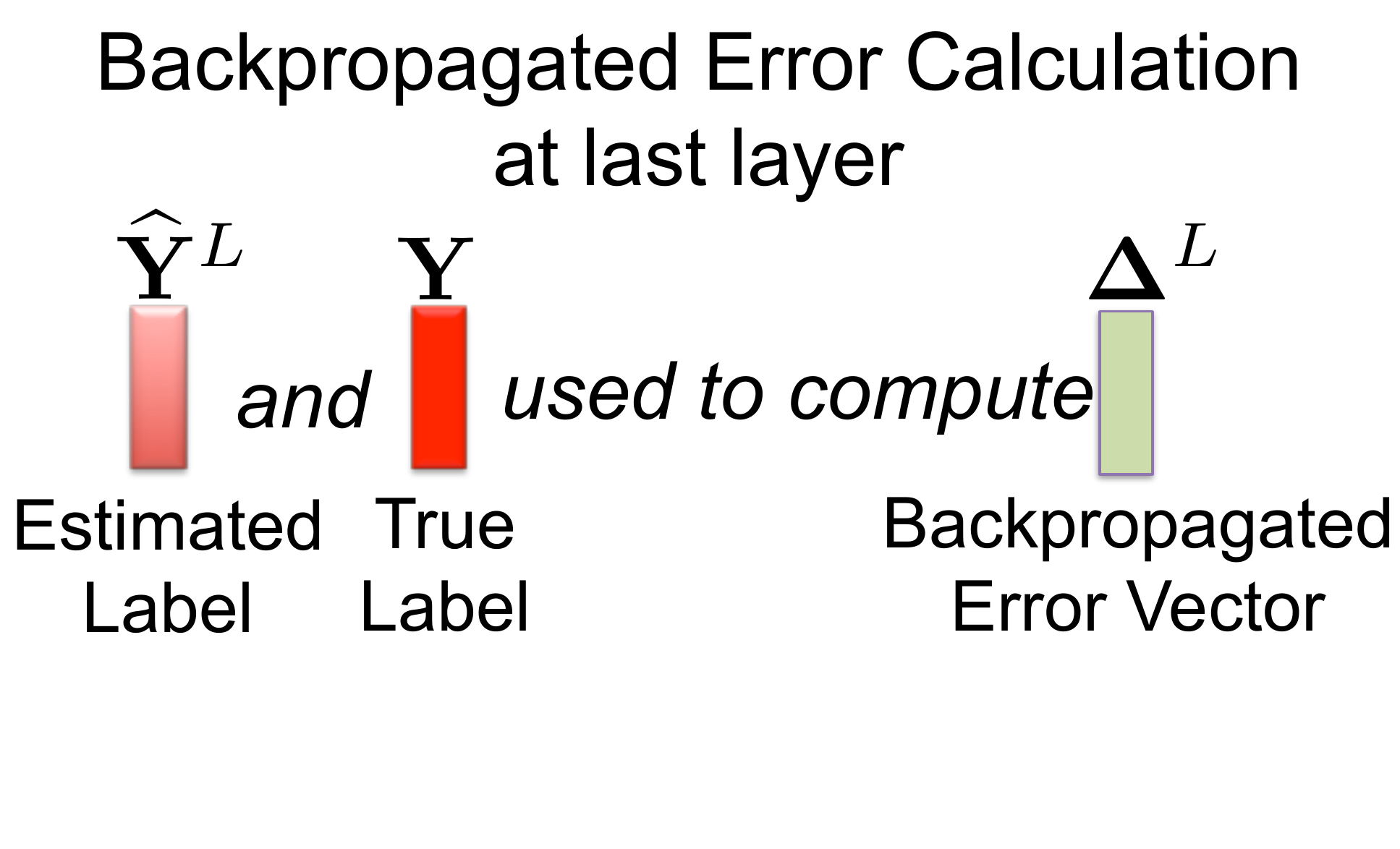}
\label{fig:DNN_step2}}} \\
\subfloat[Backpropagation stage]{\fbox{\includegraphics[height=3.8cm]{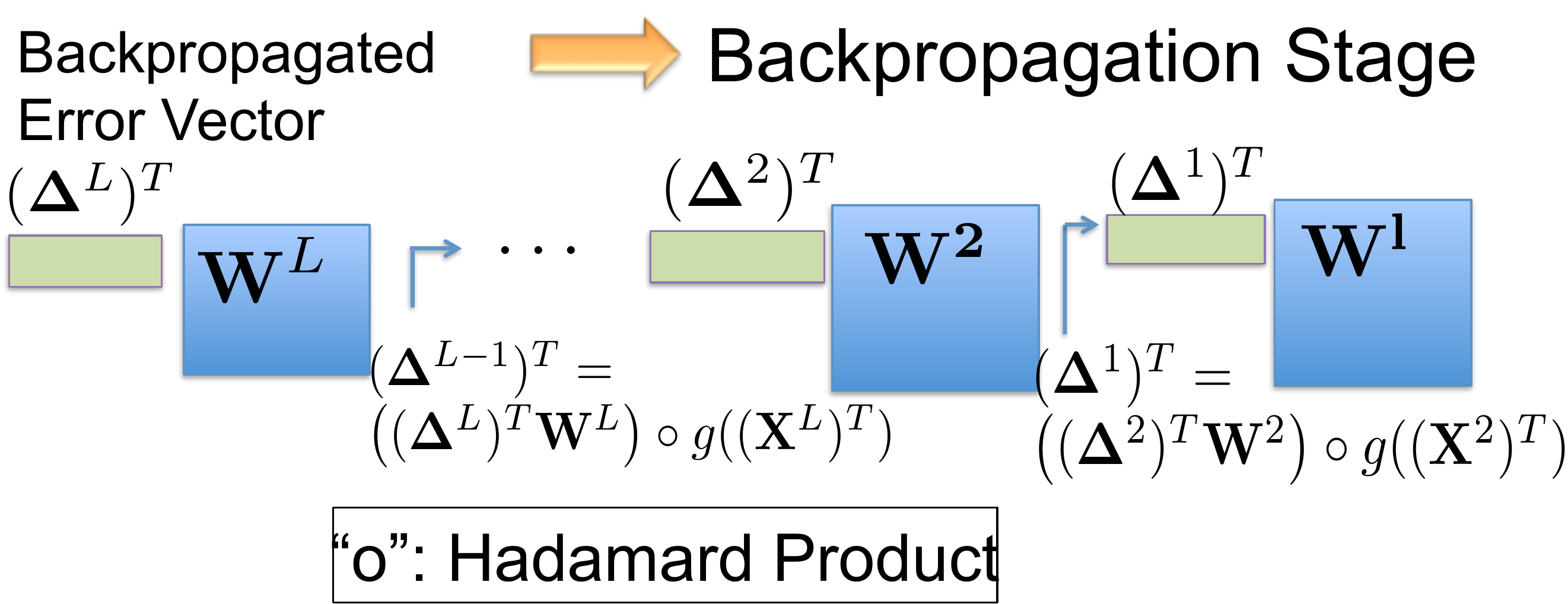}
\label{fig:DNN_step3}}} \quad
\subfloat[Update stage]
{\fbox{\includegraphics[height=3.8cm]{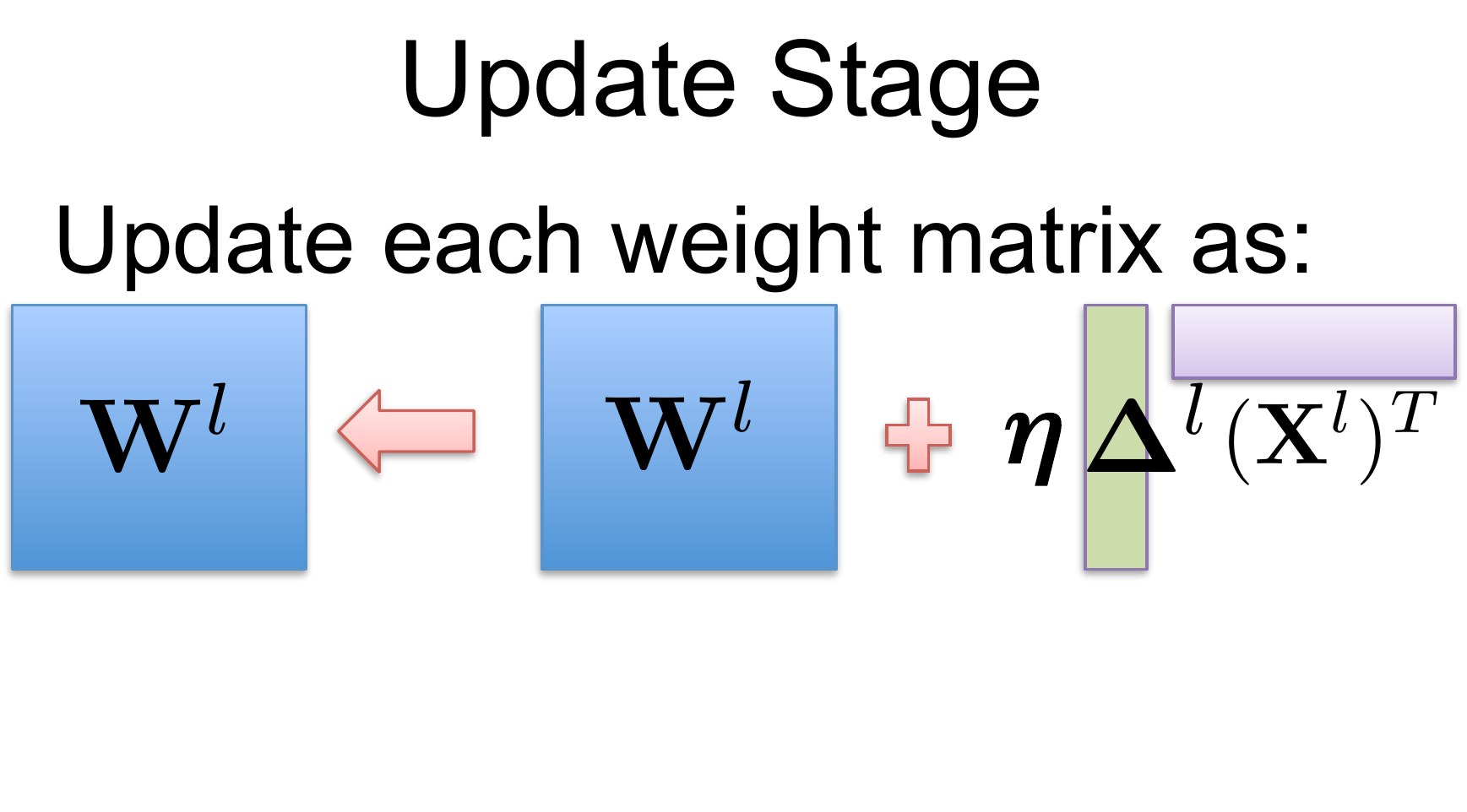}
\label{fig:DNN_step4}}}
\caption{DNN training for mini-batch size $B=1$: (a) Feedforward stage: The input data $\bm{X}^1$ is passed forward through all the layers (a matrix-multiplication followed by an element-wise nonlinear activation function $f(\cdot)$ at each layer) producing an estimate of the label $(\widehat{\bm{Y}^L})$. (b) Transition: The backpropagated error for the last layer ($\bm{\Delta}^L$) is calculated using the estimated $(\widehat{\bm{Y}}^L)$ and true label $(\bm{Y})$. (c) Backpropagation stage:  The backpropagated error $\bm{\Delta}^L$ propagates backward across the layers (a matrix-multiplication followed by a Hadamard product), generating the backpropagated error vector for every layer recursively from the previous layer. (d) Update stage: Each layer also updates itself using its backpropagated error $\bm{\Delta}^l$ and its own input $\bm{X}^l$. \label{fig:DNN_training}}
\end{figure*}

\textbf{System model:} We assume that there is a decentralized system of $P$ memory-constrained nodes that are prone to \emph{soft-errors} during computation. Soft-errors cause the node to produce entirely garbage outputs. We introduce two error-models in \Cref{sec:background}. \textcolor{black}{Under Error-Model $1$, which is an adversarial model, soft-errors only occur during the most computationally intensive operations, \textit{i.e.}, steps $O1$, $O2$ and $O3$ but the number of erroneous nodes is bounded. Under Error-Model $2$, which is a probabilistic model, soft-errors can occur during the steps $O1$, $O2$, $O3$, $C1$, $C2$ as well as encoding/decoding. There is no upper bound on the number of erroneous nodes or any specific assumption on where they can occur, but when errors occur, the output of that node is assumed to have an additive continuous-valued random noise. We elaborate upon these two models in \Cref{sec:background}. }

There is no single reliable master node, and all nodes can be unreliable. After an initial error-free setup (a pre-processing step\footnote{We assume that the cost of initial setup or pre-processing before the start of training is amortized across large number of iterations.}), all the nodes may begin their computational tasks in parallel, and proceed with the iterations of training. These memory-constrained and unreliable nodes may also communicate with each other, as required by the DNN training algorithm, or even perform some functions of a master node, such as, gathering outputs from other nodes, encoding/decoding etc., while respecting their storage constraints. After completing the required number of iterations, the final computational results (in this case, the $L$ trained parameter matrices) remain stored locally across the multiple nodes in a distributed manner. The data set and their labels are stored in a separate reliable memory unit and are communicated to all the nodes at each iteration when they access it. 




\textbf{Problem formulation}: Design an error-resilient DNN training strategy using $P$ nodes, such that:
\begin{itemize}
\item Each node can store only a $\frac{1}{K}$ fraction of each weight matrix $\bm{W}^l$ for each layer. Thus, for each layer, there is a per-node storage of $\frac{N^2}{K} + o \left( \frac{N^2}{K}\right)$ where the small additional storage of $o \left( \frac{N^2}{K}\right)$ is to store additional quantities that are negligible in storage size as compared to the $\frac{1}{K}$ fraction of $\bm{W}^l$, e.g., matrices $\bm{S}^l$, $\bm{X}^l$, $\bm{C}^l$ and $\bm{\Delta}^l$ which are all of dimensions $N \times B$ where $B \ll N$ by our assumption. In particular, we assume that $B=o\left(\frac{N}{K} \right)$ to satisfy this storage constraint.
\item All additional overheads per node including the communication complexity as well as the computational complexity of encoding/decoding \emph{in an error-free iteration}\footnote{For error-resilience, some operations such as error detection, encoding etc. are required to be performed at each iteration even though most iterations of training are actually error-free. The purpose of this assumption is only to ensure that the additional overheads introduced for error-resilience in these error-free iterations is negligible. In the few iterations where errors occur and are detected, all error-resilient strategies incur some extra costs, such as, possibly regenerating the erroneous nodes or reverting to the last checkpoint etc, that is not being compared here.} should be negligible in scaling sense as compared to the computational complexity of the local matrix multiplications and updates, \textit{i.e.}, steps $O1$, $O2$ and $O3$ \emph{at each node after parallelization}. 
\end{itemize}


Our goal is to achieve maximum error tolerance in the steps of DNN training, \textit{i.e.}, maximize the number of errors that can be corrected during training in a single iteration under both the error models. Under Error model $1$ (see \Cref{sec:background}), the number of erroneous nodes are bounded and we require the number of errors that can be corrected in any step to be higher than the maximum number of errors that can occur in that step in the worst-case. Under Error model $2$ (see \Cref{sec:background}), which allows for unbounded number of errors but their values being drawn from a continuous-valued distribution, in addition to maximizing the number of errors we can correct, our goal is also to be able to detect the occurrence of errors even when they are too many to be corrected.



For this problem formulation, we also assume that $N \gg P$, the number of parallel nodes.

\begin{rem}
In practice, the nodes also perform ``checkpointing'' at large intervals, i.e., storing the entire DNN (the $L$ parameter matrices) at a reliable disk from which the values can be retrieved when errors cannot be corrected. However, checkpointing is very expensive, even though it is assumed to be error-free, as the nodes have to access the disk, and thus can only be performed at large time intervals.
\end{rem}




\begin{rem}
Existing coded computing techniques require encoding of the matrices being multiplied. If we were to extend them naively to the problem of coded DNN training, then the matrix $\bm{W}^l$ has to be encoded afresh in each iteration because the matrix $\bm{W}^l$ is updated in each iteration during step $O3$. Encoding matrix $\bm{W}^l$ with non-sparse codes in each iteration has a huge computational cost, and is thus a major challenge for the problem of coded DNN training (violates the last criterion in problem formulation). Thus, one key contribution in this work is in proposing a unified strategy such that the matrix $\bm{W}^l$, once initially encoded, remains encoded during updates in each iteration, obviating the need to encode afresh.
\end{rem}

\noindent \textbf{Matrix Partitioning Notations:} Throughout this paper, matrices and vectors are denoted in bold font. When we block-partition a matrix $\bm{A}$ both row-wise and column-wise into $m \times n$ equal-sized blocks for any integers $m$ and $n$, we let $\bm{A}_{i,j}$ denote the block with row index $i$ and column index $j$, where $i=0,1,\dots,m-1$ and $j=0,1,\dots n-1$. Similarly, when we partition a vector $\bm{a}$ into $m$ equal parts for any integer $m$, the sub-vectors are denoted as $\bm{a}_0,\bm{a}_1,\dots,\bm{a}_{n-1}$ respectively. E.g., for $m=n=2$, the partitioning is as follows:
$$
\bm{A}= \begin{bmatrix}
            \bm{A}_{0,0} & \bm{A}_{0,1} \\
            \bm{A}_{1,0} & \bm{A}_{1,1}
          \end{bmatrix} \text{  and  }
          \bm{a}= \begin{bmatrix} \bm{a}_0 \\\bm{a}_1 \end{bmatrix}.
$$
Also, note that when a matrix $\bm{A}$ is split only horizontally or only vertically into $m$ blocks, we denote the sub-matrices as $\bm{A}_{i,:}$ or $\bm{A}_{:,i}$ respectively for $i=0,1,\ldots,m-1$.


\section{Motivating Example for Coded Matrix Multiplication}
\label{sec:motivating_examples}
In this section, we introduce a motivating example to first understand both Polynomial codes~\cite{yu2017polynomial} and MatDot codes~\cite{allerton17}, that have been proposed for distributed matrix-matrix multiplication and then, we introduce the key idea of garbage alignment in the PolyDot framework for matrix-matrix multiplication~\cite{allerton17}.  We choose $K=K'=4$. 

The main intuition behind these polynomial-based strategies is to carefully design two polynomials $\widetilde{\bm{W}}(v)$ and $\widetilde{\bm{X}}(v)$ (may also be multivariate) whose coefficients are sub-matrices of $\bm{W}$ and $\bm{X}$ respectively, such that, different sub-matrices of the final result $\bm{S}(=\bm{W}\bm{X})$ show up as coefficients of the product $\widetilde{\bm{W}}(v)\widetilde{\bm{X}}(v)$. 
The $p$-th processing node stores a unique evaluation of $\widetilde{\bm{W}}(v)$ and $\widetilde{\bm{X}}(v)$ at $v=b_p$ for $p=0,1,\ldots,P-1$, and then computes the product $\widetilde{\bm{W}}(b_p)\widetilde{\bm{X}}(b_p)$, essentially producing a unique evaluation of the polynomial $\widetilde{\bm{S}}(v):=\widetilde{\bm{W}}(v)\widetilde{\bm{X}}(v)$ at $v=b_p$. From a sufficient number of unique evaluations of the polynomial $\widetilde{\bm{S}}(v)$, the decoder is able to interpolate back all its coefficients, which include the different sub-matrices of the result $\bm{S}$. Thus, the recovery threshold, \textit{i.e.}, the number of nodes to wait for out of $P$ is $(1+ Degree(\widetilde{\bm{S}}(v)))$, which is the total number of unknowns in the interpolation of the polynomial $\widetilde{\bm{S}}(v)$.

\textbf{Matrix Multiplication using Polynomial Codes:}
In Polynomial codes~\cite{yu2017polynomial}, the first matrix $\bm{W}$ is split horizontally and the second matrix $\bm{X}$ is split vertically into $K (=4)$ blocks each, as follows:
\begin{align*}
&\bm{W}=\begin{bmatrix} \bm{W}_{0,:} \\ \bm{W}_{1,:} \\ \bm{W}_{2,:} \\ \bm{W}_{3,:}  \end{bmatrix} \text{ and}\quad \bm{X}= \begin{bmatrix}  \bm{X}_{:,0} & \bm{X}_{:,1} & \bm{X}_{:,2} & \bm{X}_{:,3} \end{bmatrix}. 
\end{align*}
Note that, the resultant matrix $\bm{S}=\bm{W}\bm{X}$ therefore takes the following form: 
\begin{align*}& \bm{S} = \begin{bmatrix} \bm{S}_{0,0} & \bm{S}_{0,1} & \bm{S}_{0,2} & \bm{S}_{0,3}\\
\bm{S}_{1,0} &  \bm{S}_{1,1}&  \bm{S}_{1,2} & \bm{S}_{1,3} \\
\bm{S}_{2,0} &  \bm{S}_{2,1}&  \bm{S}_{2,2} & \bm{S}_{2,3} \\
\bm{S}_{3,0} &  \bm{S}_{3,1}&  \bm{S}_{3,2} & \bm{S}_{3,3} \\ \end{bmatrix}
 = \begin{bmatrix}
 \bm{W}_{0,:}\bm{X}_{:,0} & \bm{W}_{0,:}\bm{X}_{:,1} & \bm{W}_{0,:}\bm{X}_{:,2} & \bm{W}_{0,:}\bm{X}_{:,3}\\
\bm{W}_{1,:}\bm{X}_{:,0} &  \bm{W}_{1,:}\bm{X}_{:,1}&  \bm{W}_{1,:}\bm{X}_{:,2} & \bm{W}_{1,:}\bm{X}_{:,3} \\
\bm{W}_{2,:}\bm{X}_{:,0} &  \bm{W}_{2,:}\bm{X}_{:,1}&  \bm{W}_{2,:}\bm{X}_{:,2} & \bm{W}_{2,:}\bm{X}_{:,3} \\
\bm{W}_{3,:}\bm{X}_{:,0} &  \bm{W}_{3,:}\bm{X}_{:,1}&  \bm{W}_{3,:}\bm{X}_{:,2} & \bm{W}_{3,:}\bm{X}_{:,3} 
 \end{bmatrix}.
\end{align*}
For this coding strategy, two polynomials are chosen as follows:
\begin{align*}
& \widetilde{\bm{W}}(v) = \sum_{i=0}^{K-1}\bm{W}_{i,:} v^{i} = \bm{W}_{0,:} + \bm{W}_{1,:}v + \bm{W}_{2,:}v^{2} + \bm{W}_{3,:}v^{3}, \\
 \text{and } & \widetilde{\bm{X}}(v) = \sum_{j=0}^{K-1} \bm{X}_{:,j}v^{jK} = \bm{X}_{:,0} + \bm{X}_{:,1} v^{4} + \bm{X}_{:,2} v^{8}  + \bm{X}_{:,3}v^{12}. 
\end{align*}
The $p$-th processing node stores a unique evaluation of $\widetilde{\bm{W}}(v)$ and $\widetilde{\bm{X}}(v)$ at $v=b_p$ for $p=0,1,\ldots,P-1$, and then computes the product $\widetilde{\bm{W}}(b_p)\widetilde{\bm{X}}(b_p)= \widetilde{\bm{S}}(b_p)$, which essentially produces a unique evaluation of the polynomial $\widetilde{\bm{S}}(v)=\widetilde{\bm{W}}(v) \widetilde{\bm{X}}(v)$ at $v=b_p$.
Observe the coefficients of the polynomial $\widetilde{\bm{S}}(v)$ as follows:
\begin{align*}
\widetilde{\bm{S}}(v)= \widetilde{\bm{W}}(v) \widetilde{\bm{X}}(v) = \sum_{i=0}^{K-1} \sum_{j=0}^{K-1} \underbrace{\bm{W}_{i,:} \bm{X}_{:,j}}_{\bm{S}_{i,j}}v^{i+ jK} = \sum_{i=0}^{3} \sum_{j=0}^{3} \bm{S}_{i,j} v^{i+ 4j}.
\end{align*}
Interestingly, all the different sub-matrices of $\bm{S}$, \textit{i.e.}, $\bm{S}_{i,j}$ show up as coefficients of the polynomial $\widetilde{\bm{S}}(v)$. Because the degree of this polynomial $\widetilde{\bm{S}}(v)$ is $15$ (for $K=4$), the decoder requires $16$ unique evaluations to interpolate the polynomial successfully. Thus, recovery threshold is $16$.

\textbf{Matrix Multiplication using MatDot Codes:}
Contrary to Polynomial codes, in MatDot codes~\cite{allerton17} the first matrix $\bm{W}$ is split vertically and the second matrix $\bm{X}$ is split horizontally into $K$ blocks each as follows:
\begin{align*}
&\bm{W}=\begin{bmatrix} \bm{W}_{:,0} & \bm{W}_{:,1} & \bm{W}_{:,2} & \bm{W}_{:,3}  \end{bmatrix}\text{ and} \quad \bm{X}= \begin{bmatrix}  \bm{X}_{0,:} \\ \bm{X}_{1,:} \\ \bm{X}_{2,:} \\ \bm{X}_{3,:} \end{bmatrix}. 
\end{align*}

Now we carefully choose the two polynomials $\widetilde{\bm{W}}(v)$ and $\widetilde{\bm{X}}(v)$ as follows: 
\begin{align*}
& \widetilde{\bm{W}}(v) = \sum_{i=0}^{K-1}\bm{W}_{:,i} v^{i} = \bm{W}_{:,0} + \bm{W}_{:,1}v + \bm{W}_{:,2}v^{2} + \bm{W}_{:,3}v^{3} \\
 \text{and, } & \widetilde{\bm{X}}(v) = \sum_{j=0}^{K-1} \bm{X}_{j,:}v^{K-1-j} = \bm{X}_{0,:}v^3 + \bm{X}_{1,:} v^{2} + \bm{X}_{2,:} v  + \bm{X}_{3,:} .
\end{align*}
As before, the $p$-th node stores a unique evaluation of $\widetilde{\bm{W}}(v)$ and $\widetilde{\bm{X}}(v)$ at $v=b_p$ for $p=0,1,\ldots,P-1$, and then computes the product $\widetilde{\bm{W}}(b_p)\widetilde{\bm{X}}(b_p)= \widetilde{\bm{S}}(b_p)$, essentially producing a unique evaluation of the polynomial $\widetilde{\bm{S}}(v):=\widetilde{\bm{W}}(v) \widetilde{\bm{X}}(v)$. Now observe the coefficients of the polynomial $\widetilde{\bm{S}}(v)$ as follows:
\begin{align*}
\widetilde{\bm{S}}(v) & = \widetilde{\bm{W}}(v) \widetilde{\bm{X}}(v) = \sum_{i=0}^{K-1} \sum_{j=0}^{K-1} \bm{W}_{:,i} \bm{X}_{j,:}v^{i+ K-1-j} = \sum_{i=0}^{3} \sum_{j=0}^{3} \bm{W}_{:,i} \bm{X}_{j,:} v^{i+3-j} \\
& = (...)+ (...)v^1 + (...)v^2 + \bm{S}v^3 + (...)v^4 + (...)v^5 + (...)v^6 .
\end{align*}
Note that, the coefficient of $v^{K-1} (= v^3)$ gives the result $\bm{S}$. All the other coefficients are practically of no use, and hence, are referred to as ``garbage.'' Because the degree of the polynomial $\bm{S}(v)$ is $6$, the decoder would require only $7$ unique evaluations to interpolate the polynomial successfully. Thus, the recovery threshold is $7$.

\textbf{Matrix Multiplication using Generalized PolyDot Codes (with Garbage Alignment):} Before introducing Generalized PolyDot Codes, we review the PolyDot framework~\cite{allerton17} for the matrix-matrix multiplication problem. The PolyDot framework splits both the matrices horizontally and vertically into $4$ sub-matrices each, as follows:
\begin{align*}
\bm{W}=\begin{bmatrix}
\bm{W}_{0,0}  &\bm{W}_{0,1} \\
\bm{W}_{1,0}  &\bm{W}_{1,1} 
\end{bmatrix} \qquad \bm{X}=\begin{bmatrix}
\bm{X}_{0,0}  &\bm{X}_{0,1} \\
\bm{X}_{1,0}  &\bm{X}_{1,1} 
\end{bmatrix}.
\end{align*}
The resultant matrix $\bm{S}$ therefore takes the form:
\begin{align*}
\bm{S}= \begin{bmatrix}
\bm{S}_{0,0} & \bm{S}_{0,1}\\
\bm{S}_{1,0} & \bm{S}_{1,1}
\end{bmatrix} = \begin{bmatrix}
\bm{W}_{0,0}\bm{X}_{0,0} + \bm{W}_{0,1}\bm{X}_{1,0} & \bm{W}_{0,0}\bm{X}_{0,1}  + \bm{W}_{0,1}\bm{X}_{1,1} \\
\bm{W}_{1,0}\bm{X}_{0,0} + \bm{W}_{1,1}\bm{X}_{1,0} & \bm{W}_{1,0}\bm{X}_{0,1}  + \bm{W}_{1,1}\bm{X}_{1,1}
\end{bmatrix}.
\end{align*}
Now the PolyDot framework (see \Cref{fig:polydot}) encodes these sub-matrices of $\bm{W}$ into a polynomial in two variables with each variable corresponding to either the row or column dimension, as follows:
\begin{align*}
\widetilde{\bm{W}}(u,v)= \sum_{i=0}^2 \sum_{j=0}^2 \bm{W}_{i,j}u^i v^j = \bm{W}_{0,0} + \bm{W}_{0,1}v + \bm{W}_{1,0}u + \bm{W}_{1,1}uv.
\end{align*}
The sub-matrices of $\bm{X}$ are encoded as follows:
 \begin{align*}
\widetilde{\bm{X}}(v,w)= \sum_{j=0}^2 \sum_{k=0}^2 \bm{X}_{j,k}v^{(1-j)} w^k = \bm{X}_{0,0}v + \bm{X}_{0,1}wv + \bm{X}_{1,0} + \bm{X}_{1,1}w.
\end{align*}
Now, observe the coefficients of the product of the two polynomials, \textit{i.e.}, $\widetilde{\bm{S}}(u,v,w)= \widetilde{\bm{W}}(u,v)\widetilde{\bm{X}}(v,w) $ as follows:
\begin{align*}
\widetilde{\bm{S}}(u,v,w) & = \widetilde{\bm{W}}(u,v)\widetilde{\bm{X}}(v,w) = \sum_{i=0}^1 \sum_{j=0}^1 \sum_{j'=0}^1 \sum_{k=0}^1 \bm{W}_{i,j}u^i v^j \bm{X}_{j',k}v^{(1-j')} w^{k} \\
& = \underset{j = j'}{\sum_{i=0}^1 \sum_{j=0}^1 \sum_{j'=0}^1 \sum_{k=0}^1} \bm{W}_{i,j} \bm{X}_{j',k} u^i v^{(1+j-j')} w^{k} + \underset{j \neq j'}{\sum_{i=0}^1 \sum_{j=0}^1 \sum_{j'=0}^1 \sum_{k=0}^1} \bm{W}_{i,j} \bm{X}_{j',k} u^iv^{(1+j-j')} w^{k} \\
& = \sum_{i=0}^1 \sum_{k=0}^1 \underbrace{\left( \sum_{j=0}^1 \bm{W}_{i,j}\bm{X}_{j,k} \right)}_{\bm{S}_{i,k} } u^i v w^{k} + \underset{j \neq j'}{\sum_{i=0}^1 \sum_{j=0}^1 \sum_{j'=0}^1 \sum_{k=0}^1} \bm{W}_{i,j} \bm{X}_{j',k} u^iv^{(1+j-j')} w^{k}. 
\end{align*}
The coefficient of $u^ivw^k$ corresponds to $\bm{S}_{i,k}$. The total number of unknowns or coefficients in this multivariate polynomial is $2\times 3 \times 2 = 12$. If we convert this multivariate polynomial into a polynomial of a single variable, e.g., using substitution $v=u^2, w=u^6$, so that there exists a bijection~\cite{allerton17} between the coefficients of the multivariate polynomial and the polynomial of a single variable, then we would require $12$ unique evaluations of the polynomial to be able to interpolate all its $12$ unknown coefficients, including the ones that contribute towards $\bm{S}$.

In this work, one of our key observations is that even though the polynomial has $12$ coefficients, the number of coefficients that are useful to us is only $4$, \textit{i.e.}, only the coefficients of $u^ivw^k$ for $i,k=0,1$ while the others are garbage. Thus, we instead propose the following variable substitution: $u=v^2, w=v^4$ in our previously proposed PolyDot framework~\cite{allerton17}. Observe the product now:

\begin{align*}
\widetilde{\bm{S}}(v)= \widetilde{\bm{S}}(u,v,w)|_{u=v^2, w=v^4} = \sum_{i=0}^1 \sum_{k=0}^1 \underbrace{\left( \sum_{j=0}^1 \bm{W}_{i,j}\bm{X}_{j,k} \right)}_{\bm{S}_{i,k} } v^{2i+4k+1}  + \underset{j \neq j'}{\sum_{i=0}^1 \sum_{j=0}^1 \sum_{j'=0}^1 \sum_{k=0}^1} \bm{W}_{i,j} \bm{X}_{j',k} v^{(2i+4k+1+j-j')}. 
\end{align*}
Note that the coefficient of $v^{2i+4k+1}$ in $\widetilde{\bm{S}}(v)$ correspond exactly to the coefficient of $u^ivw^k$ in $\widetilde{\bm{S}}(u,v,w)$, \textit{i.e.}, $\bm{S}_{i,k}$. However, some of the garbage coefficients have now aligned with each other to reduce the total number of unknown garbage coefficients, e.g., coefficient of both $uv^2$ and $w$ in $\widetilde{\bm{S}}(u,v,w)$ now get added up and form the coefficient of $v^4$ in $\widetilde{\bm{S}}(v)$. This is the key idea of garbage alignment that lies at the core of the design of Generalized PolyDot codes, as we discuss in details in the next section. The polynomial resulting after substitution is only of degree $8$. Thus it only requires $9$ unique evaluations to be able to interpolate all its coefficients. Thus, garbage alignment reduces the recovery threshold from $12$ to $9$. 

\begin{figure}
\centering
\includegraphics[height=4cm]{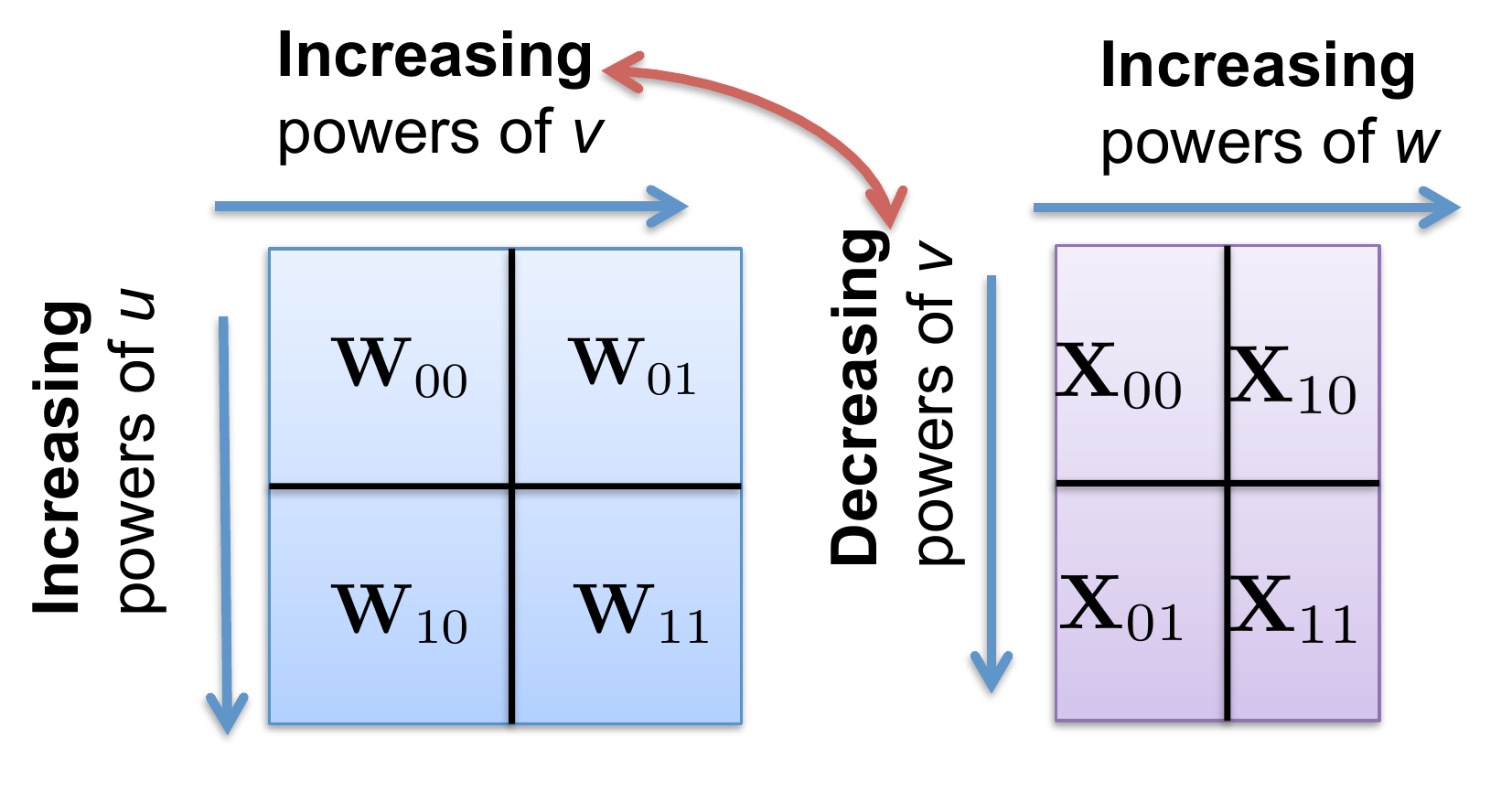}
\caption{Encoding in the PolyDot framework~\cite{allerton17} for matrix multiplication $\bm{W}\bm{X}$: Increasing powers of $v$ are chosen for the blocks along the column dimensions of $\bm{W}$, and consequently decreasing powers of $v$ are chosen for the blocks of $\bm{X}$ along the row dimension.}
\label{fig:polydot}
\end{figure}

\section{Generalized PolyDot codes for coded matrix multiplication}
\label{sec:gen_poly_codes}
In this section, we describe our Generalized PolyDot code construction for the problem of coded matrix multiplication $\bm{W}_{N_1 \times N_0}\bm{X}_{N_0 \times B}$ using $P$ memory-constrained nodes, as discussed in \Cref{sec:problem_formulation}. We choose integers $m$, $n$ and $d$ such that $mn=K$ and $nd=K'$ respectively. Next, we partition the matrix $\bm{W}$ both horizontally and vertically into an $m \times n$ grid of smaller sub-matrices of dimensions $\frac{N_1}{m} \times \frac{N_0}{n}$ each. Note that, in this type of partitioning, each sub-matrix contains a $\frac{1}{K}$ fraction of $\bm{W}_{N_1 \times N_0}$. Similarly, $\bm{X}$ is also partitioned into an $n \times d$ grid of sub-matrices of dimensions $\frac{N_0}{n} \times \frac{B}{d}$ each ($\frac{1}{K'}$ of $\bm{X}$). After this, we encode these sub-matrices of $\bm{W}$ and $\bm{X}$ by taking appropriate linear combinations and store encoded sub-matrices of $\bm{W}$ and $\bm{X}$ at each node that satisfy the storage constraints.



\Cref{thm:gen_poly_MV} states our achievability result for the problem of matrix-vector products where we are required to perform $\bm{s}=\bm{W}\bm{x}$ using $P$ nodes, such that every node can only store an $\frac{N_1}{m}\times \frac{N_0}{n}$ sub-matrix ($\frac{1}{K}$ fraction) of $\bm{W}$ and an $\frac{N_0}{n} \times 1 $ sub-vector of $\bm{x}$. 
\begin{thm}[Achievability for matrix-vector]
\label{thm:gen_poly_MV}
Generalized PolyDot codes for computing matrix-vector multiplication $\bm{W}_{N_1 \times N_0}\bm{x}_{N_0 \times 1}$ using $P$ nodes, each storing an $\frac{N_1}{m} \times \frac{N_0}{n}$ sub-matrix of $\bm{W}$ and an $\frac{N_0}{n} \times 1 $ sub-vector of $\bm{x}$, has a recovery threshold of $mn+n-1$. Thus, it can tolerate at most $P-mn-n+1$ erasures. 
\end{thm}
\begin{proof}[Proof of \Cref{thm:gen_poly_MV}]
Recall the PolyDot framework for matrix multiplication. We first block-partition $\bm{W}$ into $m \times n$ sub-matrix where $\bm{W}_{i,j}$ denotes the sub-matrix at location $(i,j)$ for $i=0,1,\ldots,m-1$ and $j=0,1,\ldots,n-1$. Let the $p$-th node ($p=0,1,\ldots,P-1$) store an encoded sub-matrix of $\bm{W}$, which is a polynomial in $u$ and $v$, as follows:
\begin{equation}
\widetilde{\bm{W}}(u,v)= \sum_{i=0}^{m-1} \sum_{j=0}^{n-1} \bm{W}_{i,j}u^i v^j,
\label{eq:encoded_W}
\end{equation}
evaluated at some $(u,v)=(a_p,b_p)$. The choice of these variables will be clarified later. The vector $\bm{x}$ is also partitioned into $n$ equal sub-vector denoted by $\bm{x}_0, \bm{x}_1, \ldots, \bm{x}_{n-1} $, each of dimensions $\frac{N_0}{n} \times 1$. Each node stores an encoded sub-vector as follows:
\begin{equation}
\widetilde{\bm{x}}(v)= \sum_{j=0}^{n-1} \bm{x}_j v^{n-j-1},
\label{eq:encoded_x}
\end{equation}
evaluated at $v=b_p$. Now, each node \emph{computes the smaller matrix-vector multiplication} $ \widetilde{\bm{W}}(a_p,b_p) \widetilde{\bm{x}}(b_p)$ which effectively results in the evaluation, at $(u,v)=(a_p,b_p)$, of the following polynomial:
\begin{equation}
\widetilde{\bm{s}}(u,v) = \widetilde{\bm{W}}(u,v) \widetilde{\bm{x}}(v) =  \sum_{i=0}^{m-1} \sum_{j=0}^{n-1} \sum_{j'=0}^{n-1} \bm{W}_{i,j} \bm{x}_{j'} u^i v^{n-j'+j-1},
\label{eq:s_polynomial1}
\end{equation}
even though the node is \textit{not explicitly evaluating} it from all its coefficients. Observe that the coefficient of $u^i v^{n-1}$ for $i=0,1,\dots,m-1$ turns out to be $\sum_{j=0}^{n-1} \bm{W}_{i,j} \bm{x}_j = \bm{s}_i$. This is obtained by fixing $j'=j$. Thus, these $m$ coefficients constitute the $m$ sub-vectors of $\bm{s}(=\bm{W}\bm{x})$. Therefore, $\bm{s}$ can be recovered by the decoder if it can interpolate these $m$ coefficients of the polynomial $\widetilde{\bm{s}}(u,v)$ in \Cref{eq:s_polynomial1} from its evaluations.
As an example, consider the case where $m=n=2$. 
\begin{align}
\widetilde{\bm{s}}(u,v)  & = (\bm{W}_{0,0} + \bm{W}_{1,0}u +  \bm{W}_{0,1}v + \bm{W}_{1,1} uv   )(\bm{x}_0 v + \bm{x}_1)  \nonumber \\
& = \bm{W}_{0,0}\bm{x}_1 +  \bm{W}_{1,0}\bm{x}_1 u +  \bm{W}_{0,1} \bm{x}_0 v^2 + \bm{W}_{1,1}\bm{x}_0 uv^2 \nonumber \\
& + \underbrace{(\bm{W}_{0,0} \bm{x}_0 + \bm{W}_{0,1} \bm{x}_1)}_{\bm{s}_0} v  + \underbrace{(\bm{W}_{1,0} \bm{x}_0 + \bm{W}_{1,1} \bm{x}_1)}_{\bm{s}_1} uv.
\label{eq:s_polynomial2}
\end{align}
More generally, we use the substitution $u=v^n$ to convert $\widetilde{\bm{s}}(u,v)$ into a polynomial in a single variable. Therefore, $a_p=b_p^n$ and each $b_p$ is unique for $p=0,1,\ldots,P-1$. Some of the unwanted, \textbf{garbage coefficients align with each other} (e.g. $u$ and $v^2$ in \Cref{eq:s_polynomial2}), but the coefficients of $u^i v^{n-1}$, \textit{i.e.}, $\bm{s}_i$  remain unchanged and now correspond to the coefficients of $v^{ni+n-1}$ for $i=0,1,\ldots,m-1$. After the substitution $u=v^n$ in $\widetilde{\bm{s}}(u,v)$, the resulting polynomial of single variable $v$ is of degree $mn+n-2$. Thus, the decoder needs to wait for $mn+n-1$ nodes, each providing a unique evaluation, to be able to interpolate all the $mn+n-1$ unknown coefficients. The recovery threshold is thus $mn+n-1$. 
\end{proof}

Now, we extend the coding strategy to the problem of matrix-matrix multiplication.

\begin{thm}[Achievability for matrix-matrix]
\label{thm:gen_poly_MM}
Generalized PolyDot codes for computing matrix-matrix multiplication $\bm{W}_{N_1 \times N_0}\bm{X}_{N_0 \times B}$ using $P$ nodes, each storing an $\frac{N_1}{m} \times \frac{N_0}{n}$ sub-matrix of $\bm{W}$ and an $\frac{N_0}{n} \times \frac{B}{d}$ sub-matrix of $\bm{X}$ has a recovery threshold of $mnd+n-1$. Thus, it can tolerate at most $P-mnd-n+1$ erasures. 
\end{thm}

\begin{proof}[Proof of \Cref{thm:gen_poly_MM}]
The matrix-matrix multiplication strategy is very similar to the matrix-vector case. The $p$-th node ($p=0,1,\ldots,P-1$) stores an encoded sub-matrix of $\bm{W}$ which is the same polynomial in $u$ and $v$, as follows:
\begin{equation}
\widetilde{\bm{W}}(u,v)= \sum_{i=0}^{m-1} \sum_{j=0}^{n-1} \bm{W}_{i,j}u^i v^j,
\end{equation}
evaluated at $(u,v)=(a_p,b_p)$. The matrix $\bm{X}$ is also block-partitioned into $n \times d$ sub-matrices, where the sub-matrix at location $(j,k)$ is denoted as $\bm{X}_{j,k}$, for $j=0,1,\ldots,n-1$ and $k=0,1,\ldots,d-1$. As per the PolyDot framework~\cite{allerton17}, the $p$-th node also stores an encoded sub-matrix of $\bm{X}$, as a polynomial in $(v,w)$, as follows:
\begin{equation}
\widetilde{\bm{X}}(v,w)= \sum_{j=0}^{n-1} \sum_{k=0}^{d} \bm{X}_{j,k}v^{n-1-j}w^{k},
\end{equation}
evaluated at $(v,w)=(b_p,c_p)$. 
Next, each node \textit{computes the smaller matrix-matrix product}: $\widetilde{\bm{S}}(a_p,b_p,c_p)=\widetilde{\bm{W}}(a_p,b_p) \widetilde{\bm{X}}(b_p,c_p)$ which effectively results in the evaluation, at $(u,v,w)=(a_p,b_p,c_p)$, of the polynomial:
$$ \widetilde{\bm{S}}(u,v,w)= \widetilde{\bm{W}}(u,v)\widetilde{\bm{X}}(v,w) = \sum_{i=0}^{m-1}\sum_{j=0}^{n-1} \sum_{j'=0}^{n-1}\sum_{k=0}^{d-1} \bm{W}_{i,j}\bm{X}_{j',k}u^i v^{n-1+j-j'}w^{k}, $$
even though the node is not \textit{explicitly evaluating it} from its coefficients. Now, fixing $j'=j$, we observe that the coefficient of $u^iv^{n-1}w^{k} $ for $i=0,1,\dots, m-1$ and $k=0,1,\dots,d-1$ turns out to be $ \sum_{j=0}^{n-1} \bm{W}_{i,j}\bm{X}_{j,k} =\bm{S}_{i,k}$. These $md$ coefficients constitute the $m \times d$ sub-matrices (or blocks) of $\bm{S}=\bm{W}\bm{X}$. Therefore, $\bm{S}$ can be recovered at the decoder if all these $md$ coefficients of the polynomial $\widetilde{\bm{S}}(u,v,w)$ can be interpolated from its evaluations at different nodes.

For garbage alignment, we propose the substitutions $(u =v^n, w=v^{mn})$ to convert $\widetilde{\bm{S}}(u,v,w)$ into a polynomial of a single variable $v$. Thus, $a_p=b_p^n$, $c_p=b_p^{mn}$ and $b_p$ is unique for $p=0,1,\ldots,P-1$. The coefficient of $u^iv^{n-1}w^{k}$ in $\widetilde{\bm{S}}(u,v,w)$ exactly correspond to the coefficient of $v^{ni+mnk+n-1}$ in $ \widetilde{\bm{S}}(v)=\widetilde{\bm{S}}(u,v,w)|_{u=v^n,w=v^mn}$, while some of the garbage terms align with each other, reducing the total number of unknowns. Observe the polynomial: $$ \widetilde{\bm{S}}(v)=\widetilde{\bm{S}}(u,v,w)|_{u=v^n,w=v^{mn}} = \sum_{i=0}^{m-1}\sum_{j=0}^{n-1} \sum_{j'=0}^{n-1}\sum_{k=0}^{d-1} \bm{W}_{i,j}\bm{X}_{j',k} v^{ni+mnk+ n-1+j-j'},$$ which is a polynomial in a single variable $v$. Its degree is given by $n(m-1)+mn(d-1)+n-1+n-1 = mnd+n-2$. Thus, the decoder needs to wait for $mnd+n-1$ nodes, each producing a unique evaluation, to be able to interpolate all its $mnd+n-1$ coefficients. 
\end{proof}

\begin{rem} Under erasures, the master node only waits for $mnd+n-1$ nodes to finish and the decoding reduces to the problem of solving a linear system of equations (polynomial interpolation). If the outputs are corrupted by errors instead of erasures, the master node gathers outputs from all the $P$ nodes and then solves a sparse reconstruction problem to decode the correct output, using the techniques in \cite{candes2005decoding}, as we also discuss in \Cref{appendix:decoding}. 
\end{rem}



Now we discuss the communication and computation costs of Generalized PolyDot Codes (with $mn=K$, $nd=K'$) in the centralized setup under erasures.
\begin{itemize}
\item Computational complexity of encoding the sub-matrices at the master node: $\mathcal{O}(P(N_1N_0+N_0B))$.
\item Total communication complexity of sending different encoded sub-matrices from the master node to each of the $P$ nodes: $\Theta(\frac{PN_1N_0}{K} + \frac{PN_0B}{K'})$.
\item Computational complexity at each worker node for the matrix multiplication: $\Theta(\frac{N_1N_0B}{mnd})$.
\item Total communication complexity of gathering different outputs at the master node from the first $mnd+n-1$ workers: $\Theta((mnd+n-1)\frac{N_1B}{md})$.
\item Computational complexity of decoding at master node:  $\mathcal{O}( (mnd+n-1)^3\frac{N_1B}{md})$.
\end{itemize}

\begin{figure}[!htbp]
\centering
\includegraphics[height=4.6cm]{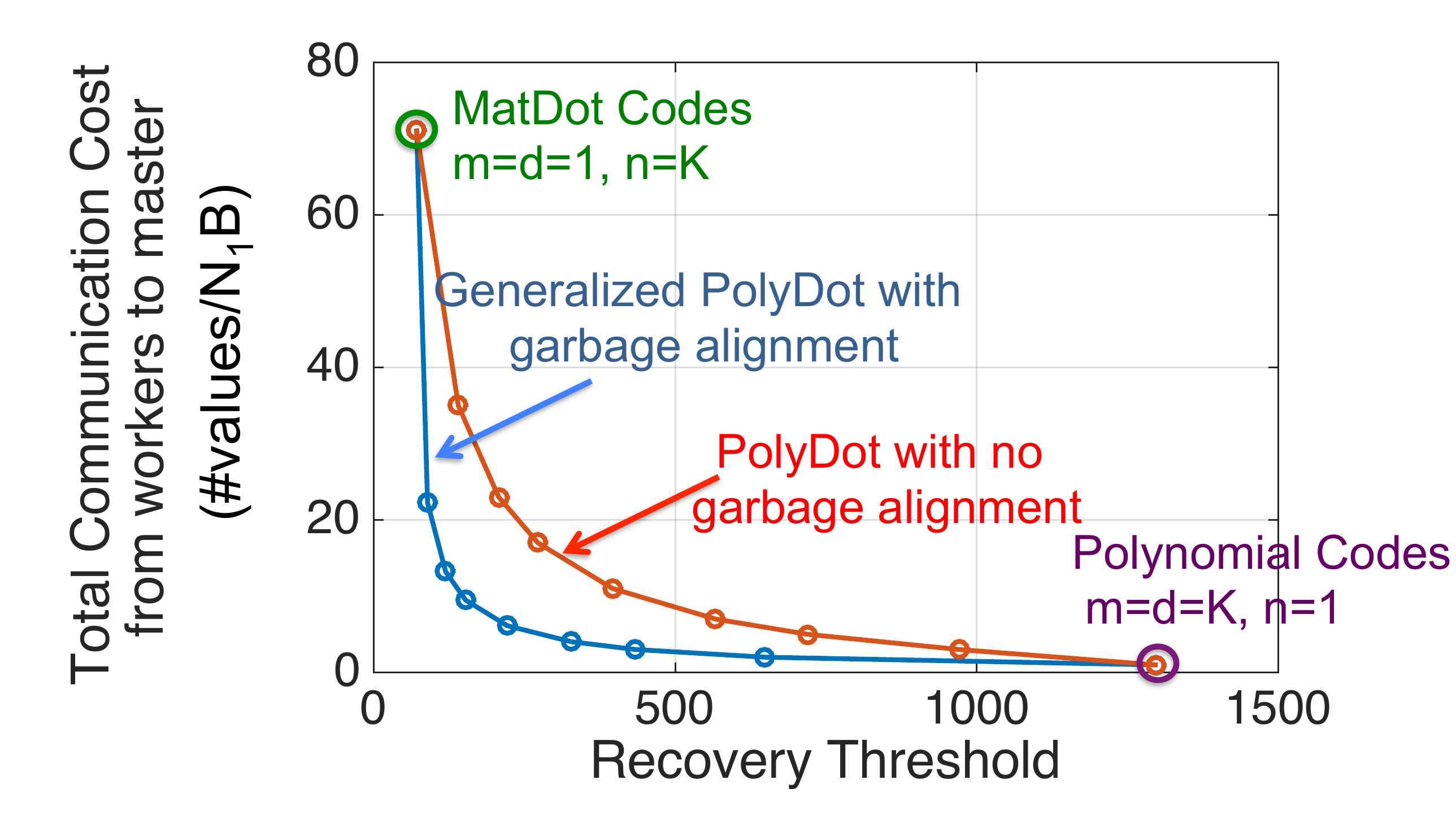}
\caption{Tradeoff between communication cost (from the workers to the master node) and recovery threshold of Generalized PolyDot codes by varying $m$, $n$ and $d$ for a fixed $K=K'=36$. MatDot codes have the lowest recovery threshold $2K-1=71$. The minimum communication cost is $N_1B$, corresponding to Polynomial codes, that have the largest recovery threshold $K^2=1296$. Generalized PolyDot codes bridge between these two strategies, improving the tradeoff using garbage alignment. The communication from the master node to the workers is not included as it is the same for all the strategies.\vspace{-0.5cm}}
\label{fig:RecThres_Comm}
\end{figure}

\noindent \textbf{Tradeoff between communication cost and recovery threshold}: In \Cref{fig:RecThres_Comm} we illustrate the tradeoff between recovery threshold and communication costs for Generalized PolyDot codes by varying $m$, $n$ and $d$. When we choose $n=1,m=K,d=K'$, the Generalized PolyDot codes reduce to Polynomial codes with recovery threshold $KK'$. On the other hand, in the regime where $K=K'$, the Generalized PolyDot codes reduce to MatDot codes when we choose $m=d=1, n=K$, resulting in a recovery threshold of $2K-1$.

Now, we move on to our Problem Formulation $2$, \textit{i.e.}, coded DNNs. 
\section{Modelling Assumptions for coded DNNs with a Result on Real Number Error Correction}
\label{sec:background}



{\color{black}{In this section, we will elaborate upon a few modeling assumptions for coded DNNs, such as, defining the two error models and communication complexity in decentralized settings. The error models introduced here lead to an interesting theoretical result on real number error correction, as stated in \Cref{thm:main}.}}

{\color{black}{
\subsection{Adversarial and Probabilistic Error Models in channel coding.}
   Let $\bm{q}$ be a $Q \times 1$ vector consisting of $Q$ real-valued symbols. The received output vector is as follows:
\begin{equation}
\bm{z}=\bm{G}^T\bm{q} + \bm{e}.
\end{equation}
Here $\bm{G}$ is the generator matrix of a $(P,Q)$ real number MDS Code and $\bm{e}$ is the $P \times 1$ error vector that corrupts the codeword $\bm{G}^T\bm{q}$. The locations of the codeword that are affected by errors is a subset $\mathcal{A} \subseteq \{0,1,\ldots,P-1\}$, and the rest are $0$. We use the notation $\mathbb{Q}$, $\mathbb{Z}$, $\mathbb{E}$ and $\widehat{\mathbb{E}}$ to denote random vectors corresponding to the symbol vector, output vector, the true error vector and the estimated error vector respectively. 

\begin{defn}[Adversarial Error Model]The subset $\mathcal{A}$ satisfies $|\mathcal{A} | \leq \lfloor \frac{P-Q}{2} \rfloor$, with no specific assumptions on the locations or values of the errors and they may be chosen advarsarially.
\end{defn}

\begin{defn}[Probabilistic Error Model] The subset $\mathcal{A}$ can be of any cardinality from $0$ to $P$, and these locations may be chosen adversarially. However, given $\mathcal{A}$, the elements of $\mathbb{E}$ indexed in $\mathcal{A}$ are drawn from iid Gaussian distributions and the rest are $0$. Also note that $\mathbb{Q}$ and $\mathbb{E}$ are independent.
\end{defn}

\begin{thm}[Real Number Error Correction under Probabilistic Error Model] Under the Probabilistic Error Model for channel coding, the decoder of a $(P,Q)$ MDS Code can perform the following:
\begin{itemize}
\item [1.] It can detect the occurrence of errors with probability $1$, irrespective of the number of errors that occurred.
\item [2.] If the number of errors that occurred is less than or equal to $P-Q-1$, then all those errors can be corrected with probability $1$, even without knowing in advance that how many errors actually occurred.
\item [3.] If the number of errors that occurred is more than $P-Q-1$, then the decoder is able to determine that the errors are too many to be corrected and declare a ``decoding failure'' with probability $1$.
\end{itemize}
\label{thm:main}
\end{thm}

This result is interesting as it essentially means that in real number error correction, one can theoretically correct $P-Q-1$ errors with probability $1$ which is more than the well-known adversarial error tolerance of $\lfloor \frac{P-Q}{2} \rfloor$ for MDS coding. A detailed proof of this result is provided in \Cref{appendix:decoding}. Here we provide the main intuition.

Let us first consider the simple case of replication. Replication is essentially a $(P,1)$ MDS Code as one single real-valued symbol is replicated $P$ times. Under an adversarial error model, one could use a majority voting among the $P$ received values and thus correct upto $\lfloor \frac{P-1}{2}\rfloor$ errors. However, under the probabilistic error model, the probability that two replicas are affected by the same value of error is $0$. Thus, as long as not all the $P$ received values are equal, one can detect that errors have occurred. Moreover, if at least two received values match out of $P$, it is most likely the original symbol unaffected by errors. Thus, one can correct $P-2$ errors with probability $1$. If no symbols match, then the decoder is able to declare a decoding failure.

This idea also extends to any $(P,Q)$ MDS Code. For a $(P,Q)$ MDS Code, the minimum Hamming distance between two codewords is $d_{min}=P-Q+1$. Thus, if the number of errors are within $\lfloor \frac{d_{min}-1}{2}\rfloor=\lfloor \frac{P-Q}{2}\rfloor$, the received output vector lies within a Hamming ball of radius $\lfloor \frac{d_{min}-1}{2}\rfloor$ around the original codeword, and is thus closest in Hamming distance to the original codeword as compared to any other codeword. If one allows for more than $\lfloor \frac{d_{min}-1}{2}\rfloor$ errors, the received output might fall within the $\lfloor \frac{d_{min}-1}{2}\rfloor$ Hamming ball of another codeword, and hence may be decoded incorrectly. However, what \Cref{thm:main} says is that if the error values are not adversarially chosen but allowed to be probabilistic, then even if we go slightly beyond $\lfloor \frac{P-Q}{2}\rfloor$, \textit{i.e.}, upto a Hamming radius of $P-Q-1$, the probability of the received output being closer to a different codeword is $0$. In other words, for a given codeword, the set of all possible outputs that are closer to other codewords in a Hamming sense has probability measure $0$ in the space of all possible values that the output can take, \textit{i.e.}, the Hamming ball of radius $P-Q-1$.

Let $\bm{H}$ be the $(P-Q)\times P$ sized parity check matrix of the MDS code, such that $\bm{H}\bm{G}^T=\bm{0}$. We first propose the following decoding algorithm (see Algorithm~\ref{algo:decoding}) to produce $\widehat{\bm{e}}$, as an estimate of $\bm{e}$, for a given channel output $\bm{z}$. If an $\widehat{\bm{e}}$ is obtained, the decoder can uniquely solve for $\widehat{\bm{q}}$ from the linear set of equations $\bm{G}^T\widehat{\bm{q}}=\bm{z}-\widehat{\bm{e}}$.

\begin{algorithm}
\caption{Decoding Algorithm to produce $\widehat{\bm{e}}$, as an estimate of $\bm{e}$, for a given channel output $\bm{z}$}
\begin{algorithmic}[1]
\STATE \textbf{If} $\bm{H}\bm{z}=\bm{0}$, \textbf{then} declare ``no errors detected'' and produce $\widehat{\bm{e}}=\bm{0}$. 
\STATE \hspace{1cm} \textbf{Else} find an $\widehat{\bm{e}}$ as follows: $\widehat{\bm{e}} =
\arg \min ||\bm{e}||_0 \text{ such that } \bm{H}\bm{e}=\bm{H}\bm{z}.$ 
\STATE \hspace{1.5cm} \textbf{If} the obtained $\widehat{\bm{e}}$ is such that $||\widehat{\bm{e}}||_0\leq P-Q-1$, \textbf{then} produce this estimate $\widehat{\bm{e}}$.
\STATE \hspace{2cm} \textbf{Else} declare a ``decoding failure.'' 
\end{algorithmic}
\label{algo:decoding}
\end{algorithm}

Now we will show that the three claims of \Cref{thm:main} hold using this proposed decoding algorithm. Let $Null(\cdot)$ denote the null-space of a matrix, and $||\cdot||_0$ denote the number of non-zero elements of a vector. We first claim that Algorithm~\ref{algo:decoding} is able to detect the occurrence of errors with probability $1$ when it checks if $\bm{H}\mathbb{Z} = \bm{0}$.  

\noindent {\bf Claim $1$:}   $\Pr{(\bm{H}\mathbb{Z} = \bm{0}|\ \mathbb{E}\neq \bm{0} )}=0$. 


\noindent \textbf{Proof Sketch of Claim $1$:} Observe that $\bm{H}\bm{z}=\bm{H}\bm{e}$. As $\bm{H}$ is also the transpose of the generator matrix of a $(P, P-Q)$ MDS Code, every $(P-Q)$ columns of $\bm{H}$ are always linearly independent. If the error locations are such that $|\mathcal{A}| \leq (P-Q)$, then $\bm{e}$ can never lie in $Null(\bm{H})$. Alternately, if $|\mathcal{A}| > (P-Q)$ and $\bm{e}$ lies in $Null(\bm{H})$, then $\bm{H}_{\mathcal{A}}\bm{e}_{\mathcal{A}}=\bm{0}$ where $\bm{H}_{\mathcal{A}}$ is a sub-matrix of $\bm{H}$ consisting of columns indexed in $\mathcal{A}$ and $\bm{e}_{\mathcal{A}}$ is a sub-vector of $\bm{e}$ consisting of elements indexed in $\mathcal{A}$. Any such vector $\bm{e}_{\mathcal{A}}$ in the $Null(\bm{H}_{\mathcal{A}})$ lies in a subspace of dimension $|\mathcal{A}|- (P-Q)$ which becomes a measure $0$ subset for a random vector $\mathbb{E}_{\mathcal{A}}$ whose all $|\mathcal{A}|$ entries are iid Gaussian. We show this rigorously in \Cref{appendix:decoding}. 

The next two claims show the error correction capability of Algorithm~\ref{algo:decoding}. Note that, for a particular realization of $\mathbb{E}=\bm{e}$, Algorithm~\ref{algo:decoding} can have three possible outcomes: it either produces $\hat{\mathbb{E}}=\bm{e}$, or $\hat{\mathbb{E}}=\bm{e}' \neq \bm{e}$, or it declares a decoding failure.

\noindent {\bf Claim $2$:} $\Pr{(\hat{\mathbb{E}} = \mathbb{E}\ | \ || \mathbb{E}||_0 \leq P-Q-1 )} =1.$  Note that, given $||\bm{e}||_0 \leq P-Q-1$, there is at least one vector, which is the true $\bm{e}$, which lies in the search-space of the decoding algorithm and hence the declaration of a decoding failure does not arise. Thus, it is sufficient to show that 
$\Pr{(\hat{\mathbb{E}} \neq \mathbb{E}\ | \ || \mathbb{E}||_0 \leq P-Q-1 )} =0.$ 

\noindent {\bf Claim $3$:} $\Pr{(\text{Decoding Failure }| \ || \mathbb{E}||_0 > P-Q-1 )} =1.$ Because the case of $\hat{\mathbb{E}}=\bm{e}$ cannot arise given $|| \bm{e}||_0 > P-Q-1$, it is sufficient to show $\Pr{(\hat{\mathbb{E}} \neq \mathbb{E}| \ || \mathbb{E}||_0 > P-Q-1 )} =0.$


\noindent \textbf{Proof Sketch of Claims $2$ and $3$:} Essentially, to prove both the claims, it is sufficient to show that $\Pr{(\hat{\mathbb{E}} \neq \mathbb{E})}=0.$ We consider all the possible error patterns $\mathcal{A}$ separately. For a particular realization of $\mathbb{E}=\bm{e}$ with a particular error pattern $\mathcal{A}$, the event of producing a wrong outcome is a strict subset of the event that there exists another $\bm{e}'\neq \bm{e}$ such that $||\bm{e}'||_0\leq P-Q-1$, $\bm{H}\bm{e}=\bm{H}\bm{e}'$ and $||\bm{e}'||_0 \leq |\mathcal{A}|$. We also fix $\mathcal{A}'$ as the set of non-zero indices for $\bm{e}'$ and show that the probability of the event goes to zero for all possible $\mathcal{A}'$. 

Note that, for $\bm{H}\bm{e}=\bm{H}\bm{e}'$ to hold, $\bm{e}=\bm{e}' + \bm{h},$ where $\bm{h}\in Null(\bm{H}) \backslash \{\bm{0} \}$. Thus  $|\mathcal{A}'\cup \mathcal{A} |- |\mathcal{A}'|$ indices of $\bm{h}$, indexed in the set $(\mathcal{A}'\cup \mathcal{A})\backslash \mathcal{A}'$ match exactly with $\bm{e}$ (see \Cref{fig:e_prime}).

\begin{figure}[!ht]
\centering
\includegraphics[height=2.8cm]{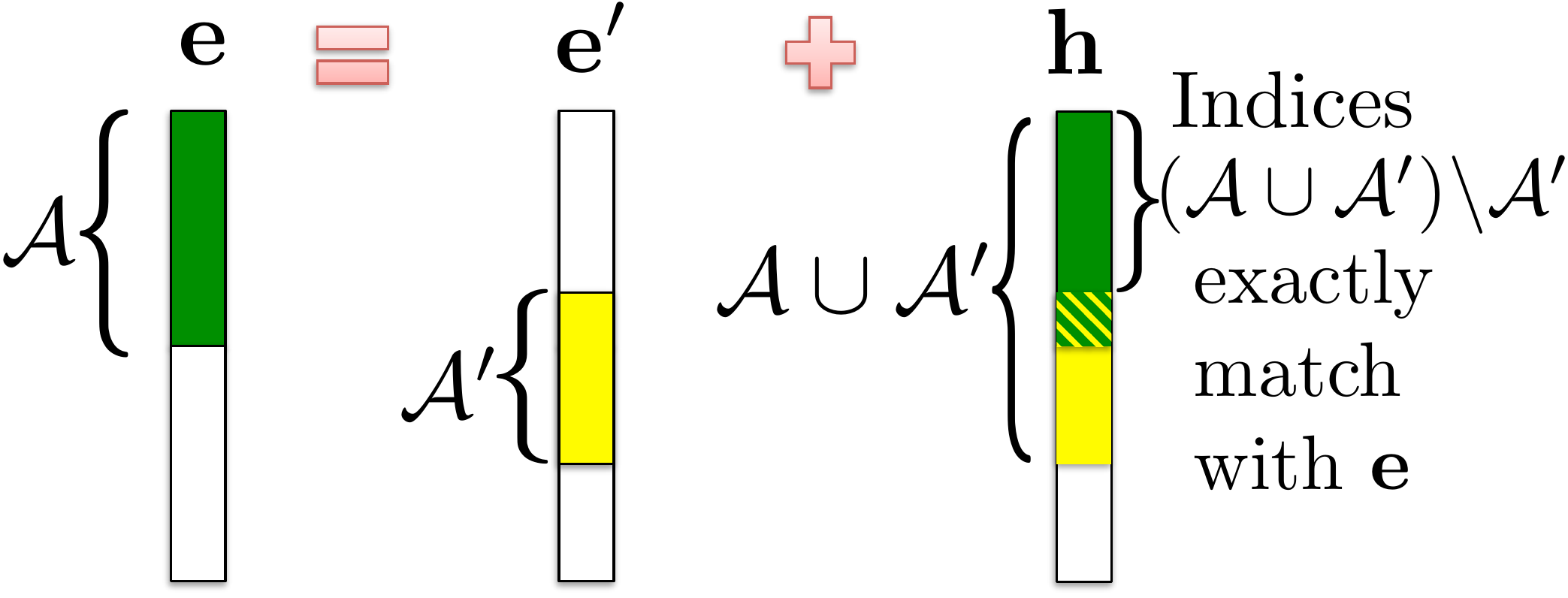}
\caption{Key intuition behind the proof of \Cref{thm:main}: Let $\mathcal{A}$ be the locations of errors in true error vector $\bm{e}$. Suppose there exists another vector $\bm{e}'$ with non-zero indices $\mathcal{A}'$ such that $\bm{e}=\bm{e}'+\bm{h}$ where $\bm{h}\in Null \ Space(\bm{H}) \backslash \{\bm{0} \}$. Then the elements of $\bm{h}$ indexed in the set $(\mathcal{A}\cup\mathcal{A}')\backslash\mathcal{A}'$ exactly match with $\bm{e}$. Thus, $\bm{e}_{(\mathcal{A}\cup\mathcal{A}')\backslash\mathcal{A}'}$, being a sub-vector of $\bm{h}_{(\mathcal{A}\cup\mathcal{A}') }$, lies in a subspace of dimension $|(\mathcal{A}\cup\mathcal{A}') |-(P-Q)$ which becomes a measure $0$ subset for a random vector $\mathbb{E}_{(\mathcal{A}\cup \mathcal{A}')\backslash\mathcal{A}'}$ whose all $|(\mathcal{A}\cup \mathcal{A}')|- |\mathcal{A}'|$ elements are from iid Gaussian distributions.
\label{fig:e_prime}}
\end{figure}


\textbf{The key intuition behind this proof is that once a certain number of elements of a vector $\bm{h}$ are allowed to be chosen from iid Gaussian distributions, the probability of the vector $\bm{h}$ still lying in $Null(\bm{H})$ becomes $0$.} To understand this better, observe that, 
\begin{equation}
\bm{0}=\bm{H}\bm{h}=
\bm{H}_{(\mathcal{A}\cup \mathcal{A}')} \bm{h}_{(\mathcal{A}\cup \mathcal{A}')} =
\begin{bmatrix} \bm{H}_{(\mathcal{A}\cup \mathcal{A}')\backslash\mathcal{A}' } & \bm{H}_{\mathcal{A}'}\end{bmatrix} \begin{bmatrix} \bm{h}_{(\mathcal{A}\cup \mathcal{A}')\backslash\mathcal{A}'} \\
\bm{h}_{\mathcal{A}'}
\end{bmatrix} = \begin{bmatrix} \bm{H}_{(\mathcal{A}\cup \mathcal{A}')\backslash\mathcal{A}' } & \bm{H}_{\mathcal{A}'}\end{bmatrix} \begin{bmatrix} \bm{e}_{(\mathcal{A}\cup \mathcal{A}')\backslash\mathcal{A}'} \\
\bm{h}_{\mathcal{A}'}
\end{bmatrix}.
\label{eq:null_space1}
\end{equation}
Therefore, given  $\bm{e}_{\mathcal{A}}$  and a particular choice of non-zero indices $\mathcal{A}'$ for $\bm{e}'$, we need to show that the probability that there exists an $\bm{h}=\bm{e}'-\bm{e}$ (and hence an $\bm{h}_{\mathcal{A}'}$) such that \Cref{eq:null_space1} holds is $0$. Any vector $\begin{bmatrix} \bm{e}_{(\mathcal{A}\cup \mathcal{A}')\backslash\mathcal{A}'} \\
\bm{h}_{\mathcal{A}'}
\end{bmatrix}$ satisfying \Cref{eq:null_space1} lies in a subspace of dimension $|(\mathcal{A}\cup \mathcal{A}')|-(P-Q)$, and $\bm{e}_{(\mathcal{A}\cup \mathcal{A}')\backslash\mathcal{A}'}$ being a sub-vector of this vector also lies in a sub-space of dimension at most $|(\mathcal{A}\cup \mathcal{A}')|-(P-Q)$. However, this becomes a measure $0$ subspace for a random vector $\mathbb{E}_{(\mathcal{A}\cup \mathcal{A}')\backslash\mathcal{A}'}$ whose all $|(\mathcal{A}\cup \mathcal{A}')|- |\mathcal{A}'|$ elements are drawn from iid Gaussian distributions. We show this rigorously in \Cref{appendix:decoding}.

\begin{rem}
While we theoretically show that real number MDS coding can correct $P-Q-1$ errors, our proposed decoding algorithm requires sparse reconstruction for undetermined systems which is NP Hard~\cite{candes2005decoding}. For practical purposes, one might consider using an L1-norm relaxation~\cite{candes2005decoding} or other kinds of polynomial-time sparse reconstruction algorithms proposed in the compressed sensing literature~\cite{tropp2010computational}, which are known to be reasonably accurate under various restrictions on matrix $\bm{H}$.
\end{rem}

Let us now understand what these models mean in the context of coded DNNs. 
}}
\subsection{Error Models for the coded DNN problem.}
Recall from our problem formulation (\Cref{sec:problem_formulation}) that we are interested in model-parallel architectures that parallelize each layer across $P$ error-prone nodes (that can be reused across layers) because the nodes cannot locally store the entire matrix $\bm{W}^l$. As the steps $O1$, $O2$ and $O3$ are the most computationally intensive ($\Theta(N^2B)$) steps at each layer, we restrict ourselves to schemes where these three steps for each layer are parallelized across the $P$ nodes\footnote{The steps $C1$ and $C2$ are lower in computational complexity, \textit{i.e.}, $\Theta(NB)$ which is lower in scaling sense as compared to $\Theta(N^2B)$ and hence may or may not be parallelized across multiple nodes.}. In such schemes, communication will be required after steps $O1$ and $O2$ as the partial computation outputs of steps $O1$ and $O2$ at one layer might be required at another node to compute the input $\bm{X}^{(l+1)}$ or backpropagated error $\bm{\Delta}^{(l-1)}$ for another layer\footnote{Note that there could be alternate parallelization schemes where the different layers of the network are parallelized across different nodes instead of each layer being parallelized across all nodes. These schemes might have lower communication but some of the nodes stay idle and under-utilized, \textit{i.e.}, while computations are being performed in one layer, the nodes containing the other layers stay idle. It will be an interesting future work to explore the computation-communication tradeoffs among these alternate parallelization schemes.}. We define Error Models $1$ and $2$, which are essentially realizations of the probabilistic and adversarial models for the coded DNN problem, with some additional assumptions.


\begin{defn}[Error Model $1$: Adversarial Error Model] 
\textit{Any} node can have soft-errors but \textit{only} during the steps $O1$, $O2$ and $O3$, which are the most computationally intensive operations in DNN training. Encoding, error-detection, decoding, nonlinear activation and Hadamard product are assumed to be error-free. These operations require negligible time and number of operations\footnote{The shorter the computation, the lower is the probability of soft-errors. E.g.,  a Poisson process of soft errors~\cite{li2007memory} makes the number of soft-errors have mean proportional to the interval length.} because most of the time and resources are spent on steps $O1$, $O2$ and $O3$. There is no specific assumption on the locations of the erroneous nodes and they may be adversarial. However, the total number of erroneous nodes at any layer during $O1$, $O2$ and $O3$ are known to be bounded by $t_1$, $t_2$ and $t_3$ respectively. There is also no assumption on the distribution of the errors for this model, but all the output values of an erroneous node,~e.g., all the values of the output matrix or vector are affected by errors.
\end{defn}
\begin{defn}[Error Model $2$: Probabilistic Error Model]
\textit{Any} node can have soft-errors during any primary operation such as encoding, decoding, nonlinear activation, Hadamard product as well as steps $O1$, $O2$ and $O3$, and there is no bound on the number of errors. The locations of the erroneous nodes can still be adversarial. The entire output of an erroneous node (all the values of the output matrix or vector) is assumed to be corrupted by additive iid Gaussian noise. However, under Error Model $2$, we use verification steps to check for decoding errors that have very low complexity (compared to the primary steps), and hence those verification steps are assumed to be error-free.
\end{defn}

\begin{rem} 
\label{rem:error_model} 
Error Model $1$ is a ``worst-case'' abstraction, which is useful when it is difficult to place probabilistic priors on errors. It essentially means that the number of errors that can occur in the longer steps is bounded. Hence, if we choose a strategy with higher error tolerance, we can correct all errors. On the other hand, Error Model $2$ allows for errors in all primary operations and also has no upper bound on the number of errors. However, it makes one simplifying assumption. Specifically, the continuous distribution of noise simplifies our analyses by avoiding complicated probability distributions that arise in finite number of bits representations. We acknowledge that this simplification can lead to optimistic conclusions, e.g., it allows us to correct more errors than the adversarial model (see \Cref{thm:main}) with probability $1$ and also detect the occurrence of errors (``garbage outputs'') with probability $1$ (because the noise takes any specific value with probability zero; see~\Cref{appendix:decoding}). This model is only accurate in the limit of large number of bits of precision. In practical implementations, our probability $1$ results should be interpreted as holding with high probability (e.g. it is unlikely, but possible, that two erroneous nodes produce the exact same garbage output). Note that because both replication and coding can exploit Error Model $2$ for error-correction and detection, it does not bias our results towards coding relative to replication.
\end{rem}

\subsection{Error Tolerance Goals for the coded DNN Strategy.}  

We would like to be able to correct as many erroneous nodes as possible after the steps $O1$ and $O2$, because outputs are communicated to other nodes after these two steps.


\begin{defn}[Error Tolerances $(t_f, t_b)$]
Under Error Models $1$, for any layer $l$, the error tolerances are $(t_f, t_b)$ if $t_f$ and $t_b$ erroneous node outputs can be detected and corrected in the worst case immediately after steps $O_1$ and after step $O_2$ respectively. Similarly, under Error Model $2$, the error tolerances are $(t_f, t_b)$ if $t_f$ and $t_b$ erroneous node outputs can be detected and corrected with probability $1$ immediately after steps $O_1$ and after step $O_2$ respectively.  
\end{defn}

Our goal is to maximize the values of these error tolerances under both the error models.

For any coding strategy, the achievable $t_f$ and $t_b$'s depend on the number of nodes available ($P$) and other parameters of the coding strategy, e.g., $m,n$ etc. as derived in  \Cref{sec:comparison_DNN_MV}. Note that, after steps $O1$ and $O2$, we do not necessarily correct only the errors that occur during those steps. Depending on the coding strategy used, errors occurring in other steps could also get corrected after either $O1$ or $O2$.  Under Error model $1$, the values of the achievable $t_f$ and $t_b$'s will be required to be greater than appropriate functions of $t_1,t_2$ and $t_3$ based on the coding strategy being used, as we also elaborate in \Cref{sec:comparison_DNN_MV}.

\subsection{Communication Complexity.}


In this work, we use standard definition of communication complexity for fully distributed and decentralized architectures, as mentioned in~\cite{chan2007collective}.

\begin{defn}[Communication Complexity,~\cite{chan2007collective}]
The communication cost of sending a message of $N$ items between two nodes will be modeled by $\alpha+\beta N$, in the absence of network conflicts. Here $\alpha$ and $\beta$ are two constants representing the message startup time and per data item transmission time respectively. 

\end{defn}


\begin{rem} 
\label{rem:broadcast}
We assume that each node can communicate simultaneously to at most a constant number (say $2$) of nodes. Thus, when one node has to broadcast the same $N$ values to $P$ other nodes, it usually initiates communication link (or startup) with the other nodes in the form of a spanning-tree in $\log{P}$ rounds and then starts communicating the $N$ values across this tree-like transmission network of the nodes. For more details, the reader is referred to \cite{van1997summa,chan2007collective}. Following~\cite{van1997summa,chan2007collective}, the communication cost for this type of broadcast is given by: $\alpha \log{P} + \beta N $. Here, the first term arises because the communication link (or startup) between all the nodes is set up in $\Theta(\log_2{P})$ rounds, and then the second term denotes the cost of sending the $N$ values across this tree-like network of nodes. When multiple nodes have to communicate with each other, the communication cost can be efficiently managed using collective communication protocols, as suggested in~\cite{chan2007collective}. For instance, when all nodes send their own, unique message of $N$ values to all other nodes, a communication protocol called All-Gather~\cite{chan2007collective} is used whose communication cost is $\alpha \log{P} + 2\beta PN$. We will use Broadcast and All-Gather protocols to prove our results on communication complexities in \Cref{appendix:complexity}.

\end{rem}



\section{Applying Existing Strategies to the coded DNN Problem}
\label{sec:existing_strategies}
Before we introduce our new coded DNN training strategy using Generalized PolyDot codes, let us review the application of two existing strategies for coded DNN training for the case of $B=1$. Later in this paper, we will compare the error tolerance of these strategies with our proposed strategy. The case of mini-batch $B>1$ is similar and can be obtained as an extension of the case of $B=1$, as we discuss in \Cref{sec:extension}. Because the operations are similar across layers, henceforth, we will omit the superscript $(\cdot)^l$  and will only use the notations $\bm{W}$, $\bm{S}$ (or vector $\bm{s}$), $\bm{X}$ (or vector $\bm{x}$), $\bm{C}^T$ (or vector $\bm{c}^T$), and $\bm{\Delta}$ (or vector $\bm{\delta}^T$) respectively for a particular layer.

The problem formulation $2$ stated in \Cref{sec:problem_formulation} discusses our goals. Essentially, for $B=1$, we are required to design a coded DNN training strategy, that we denote as $\mathcal{C}(N,K,P)$, which performs distributed ``post'' and ``pre'' multiplication of the same matrix $\bm{W}$ with vectors $\bm{x}$ and $\bm{\delta}^T$ respectively at each layer and a distributed update ($\bm{W}+\eta \bm{\delta}\bm{x}^T$), along with all the other operations using $P$ memory-constrained nodes.\\

\noindent \textbf{Replication ($\mathcal{C}_{\mathrm{rep}}(K,N,P)$)}:  For every layer, the matrix $\bm{W}$ is block-partitioned across a grid of $m\times n$ nodes where $K=mn$, and $\frac{P}{mn}$ replicas of this system is created using a total of $P$ nodes (assume $mn$ divides $P$). For computing $\bm{s}=\bm{W}\bm{x}$, the node with grid index $(i,j)$ accesses $\bm{x}_j$ and computes $\bm{W}_{i,j}\bm{x}_j$. Then, the first node in every row aggregates and computes the sum $\sum_{j=0}^{n-1}\bm{W}_{i,j}\bm{x}_j = \bm{s}_i$ for $i=0,1,\dots, m-1$. For the example with $m=n=2$, observe the two sub-vectors of $\bm{s}$ that are required to be reconstructed:
\begin{equation*}
\bm{s}=
\begin{bmatrix}
\bm{s}_0 \\
\bm{s}_1
\end{bmatrix}  =  \begin{bmatrix}
           \bm{W}_{0,0} & \bm{W}_{0,1} \\
     \bm{W}_{1,0} & \bm{W}_{1,1}   
          \end{bmatrix} \begin{bmatrix} 
          \bm{x}_0 \\
          \bm{x}_1
          \end{bmatrix}
          =    
          \begin{bmatrix}
           \bm{W}_{0,0} \bm{x}_0 + \bm{W}_{0,1} \bm{x}_1\\
     \bm{W}_{1,0} \bm{x}_0 + \bm{W}_{1,1} \bm{x}_1  \end{bmatrix}.
\end{equation*}
After these computations, all the replicas computing the same sub-vector, \textit{i.e.}, say $\bm{s}_i$, exchange their computational outputs for error detection and correction. Note that the communication cost for exchanging outputs is only a function of the number of nodes, $P$, and does not depend on $N$. This is because the nodes can just exchange a single value of their computation result among each other instead of the entire sub-vector. Thus, this communication cost is much lower as compared to the computational cost of matrix multiplication in the regime $N \ll P$.

Under Error Model $1$, any $t=\lfloor\frac{P-mn}{2mn}\rfloor$ errors can be tolerated in the worst case. However under Error Model $2$, the probability of two outputs having exactly same error is $0$. As long as an output occurs at least twice, it is almost surely the correct output. Thus, any $t=\frac{P}{mn}-2$ errors can be detected and corrected. Then, the correct sub-vectors ($\bm{s}_i$'s) are communicated to the respective nodes that require it for generating their input for the next layer, and the sub-matrices stored in the erroneous nodes are \textbf{regenerated} by accessing other nodes known to be correct.

\noindent \textit{Additional Steps:} At regular intervals, the system also \textbf{checkpoints}, i.e., sends the entire DNN to a \textit{disk} for storage. This disk-storage, although time-intensive to retrieve from, can be assumed to be error-free. Under Error Model $2$, if more than $t$ errors occur, then with probability $1$, none of the outputs match. The system detects the occurrence of errors even though it is unable to correct them. So, it retrieves the DNN from the disk and reverts the computation to the last checkpoint. 

A similar technique is applied for backpropagation. The the node with index $(i,j)$ accesses $\bm{\delta}^T_i$ and computes $\bm{\delta}^T_i \bm{W}_{i,j}$. Finally the last node in every column aggregates and computes $\sum_{i=0}^{m-1}\bm{\delta}^T_i \bm{W}_{i,j} =  \bm{c}^T_j $ for $j=0,1,\dots,n-1$.  Error check occurs similarly. If errors can be corrected, then $\bm{c}^T_j $'s are communicated to the respective nodes that require it to compute backpropagated error for the next layer, along with $\bm{x}_j$.  Interestingly, after these operations, the node with index $(i,j)$ has $\bm{x}_j$ and $\bm{\delta}^T_i$, and is thus able to update itself as $\bm{W}_{i,j}\leftarrow \bm{W}_{i,j}+ \eta \bm{\delta}_i \bm{x}^T_j $ respectively.
\begin{lem}[Error Tolerances for Replication Strategy] The error tolerances for the replication strategy are $t_f=t_b=\lfloor\frac{P-mn}{2mn}\rfloor$ under Error Model $1$ and $t_f=t_b=\frac{P}{mn}-2$ under Error Model $2$, assuming $mn$ divides $P$. 
\label{lem:replication_error_tolerance}
\end{lem}

\noindent \textbf{Preliminary MDS-code-based strategy ($\mathcal{C}_{\mathrm{mds}}(K,N,P)$)}:
Another strategy (details in \cite{dutta2018DNN1}) is to use two systematic MDS codes to encode the block-partitioned matrix $\bm{W}$. 
The total number of nodes used by this strategy is $P=P_f+P_b-mn$ (see \Cref{fig:MDS_strategy}), of which only $mn$ nodes are common for both steps $O1$ and $O2$. At each layer, the matrix $\bm{W}$ is block-partitioned into $m \times n$ blocks and arranged across a grid of processors as shown in \Cref{fig:MDS_strategy}. Then, these blocks are coded using a $(\frac{P_f}{n},m)$ systematic MDS code along the row dimension and a $(\frac{P_b}{m},n)$ systematic MDS code along the column dimension as follows:
\begin{align}
\left(\bm{G}_r^T  \otimes \bm{I}_{N/m}  \right) \bm{W} \left( \bm{G}_c \otimes \bm{I}_{N/n}\right).
\end{align}
where $\bm{G}_r$ and $\bm{G}_c$ are the generator matrices of the two systematic MDS codes for the row and column dimensions, $\otimes$ denotes the Kronecker product and $\bm{I}_{N/m}$ and $\bm{I}_{N/n}$ denote identity matrices of the corresponding dimensions.

In step $O1$, only $P_f$ nodes corresponding to the  $(\frac{P_f}{n},m)$ code are active. Similarly, in step $O2$, only $P_b$ nodes are active corresponding to the $(\frac{P_b}{m},n)$ code. Errors that happen in the update step $O3$ corrupt the updated sub-matrices and are detected and corrected the next time those sub-matrices are used to produce an output to be sent to another node, which could be either after step $O1$ or step $O2$ of the next iteration at that layer. Thus, errors of step $O3$ are corrected either after step $O1$ or step $O2$ at that layer, in the next iteration. 

Under Error Model $1$, in the worst case, the strategy thus requires $t_f \geq t_1+t_3$ and $t_b \geq t_2+t_3$ to be able to detect and correct all the errors. Under Error Model $2$, when the number of errors are greater than $t_f$ or $t_b$, they can only be detected with probability $1$ but cannot be corrected, as elaborated in \cite{dutta2018DNN1}.

\begin{figure}[!ht]
\centering
\subfloat[Original Matrix]{\includegraphics[height=3cm]{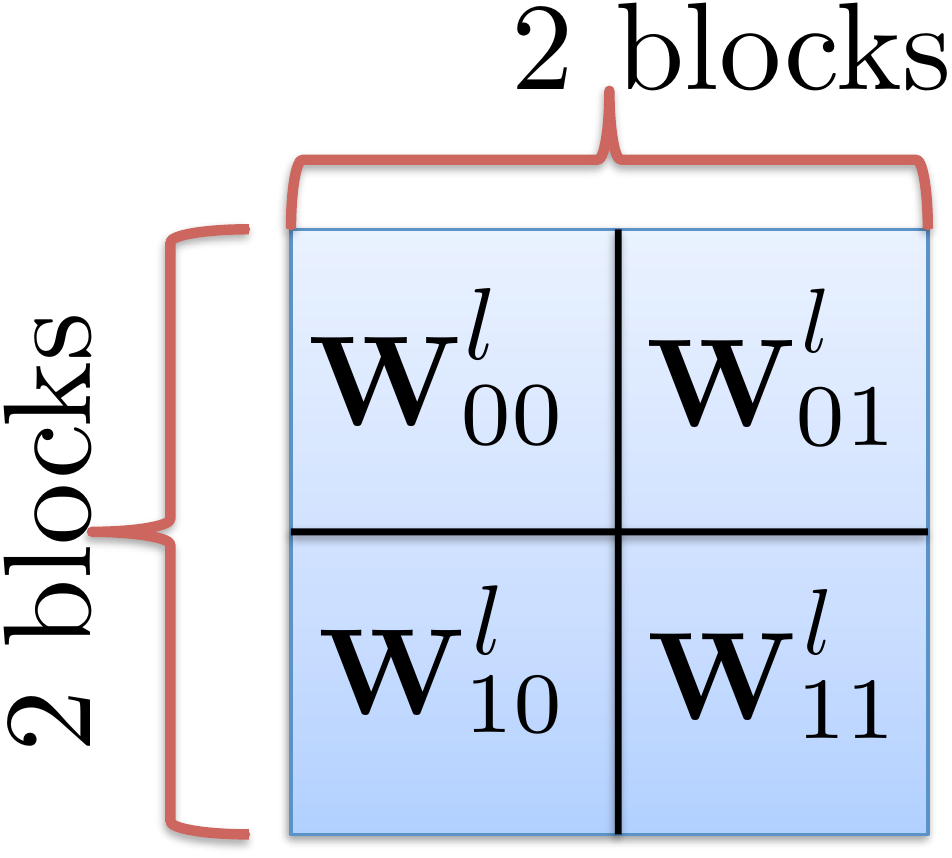}}
\hspace{0.2cm}
\subfloat[Entire Coded Matrix]{\includegraphics[height=3cm]{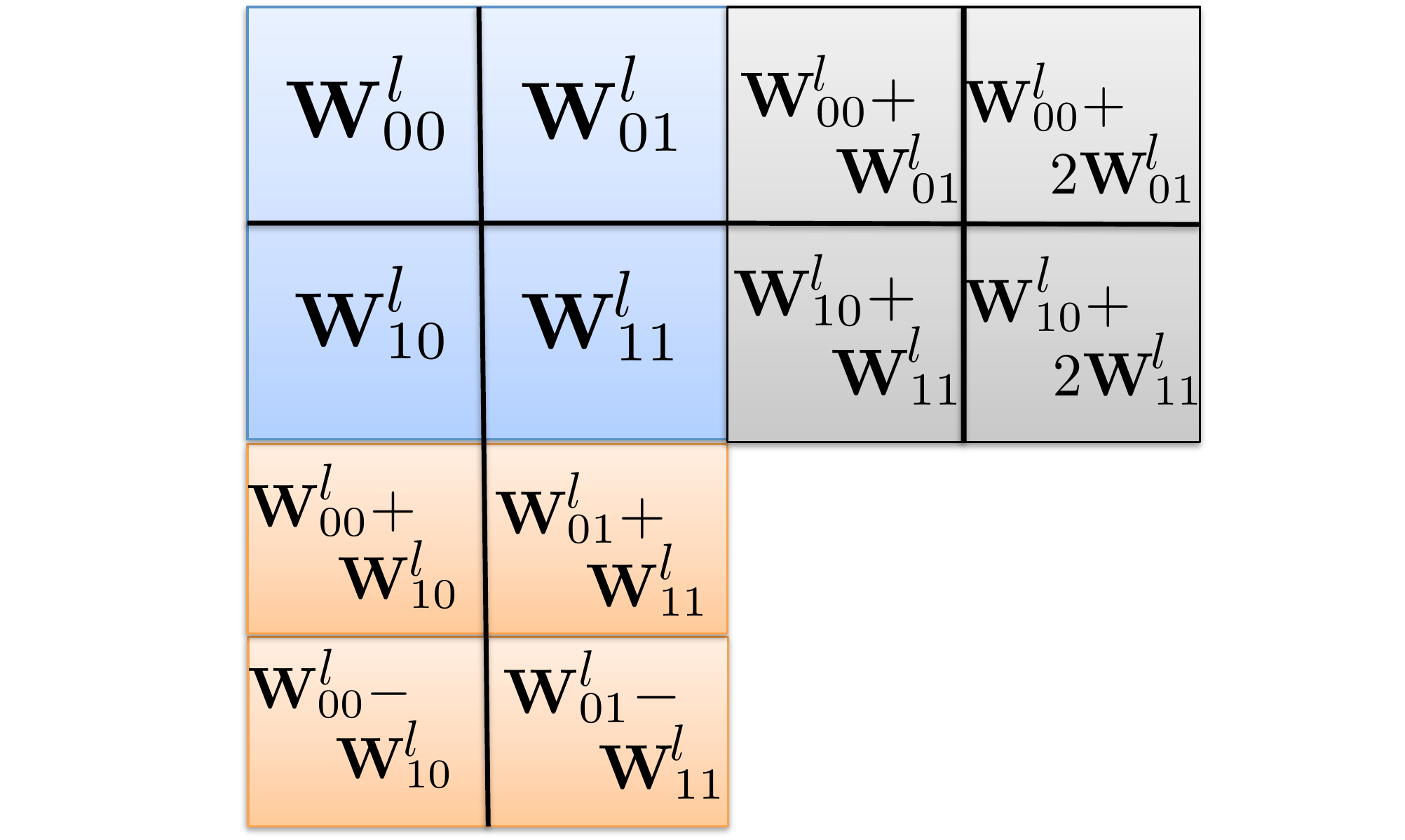}}
\hspace{0.2cm}
\subfloat[Active (Feedforward)]{\includegraphics[height=3cm]{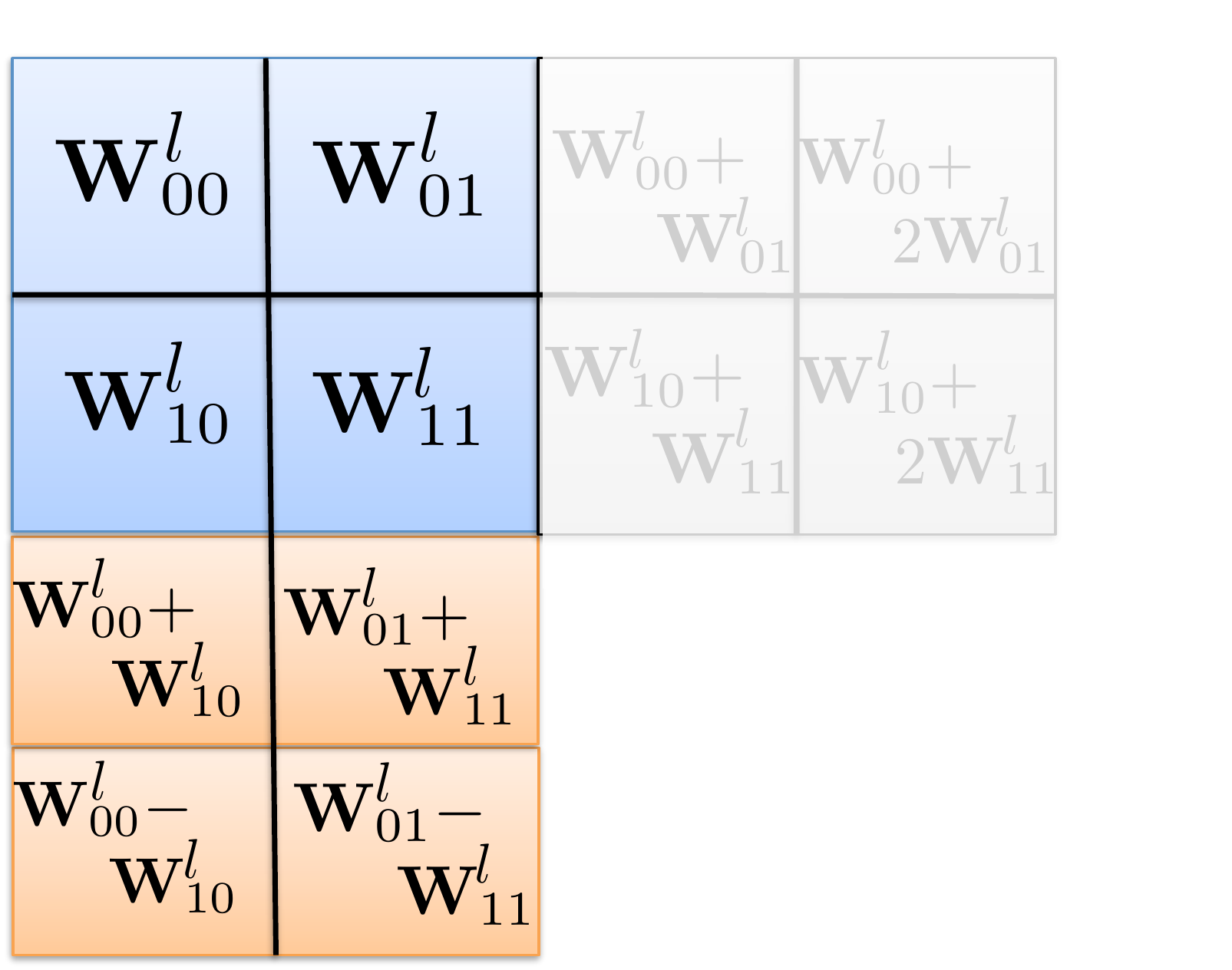}}
\hspace{0.2cm}
\subfloat[Active (Backpropagation)]{\includegraphics[height=3cm]{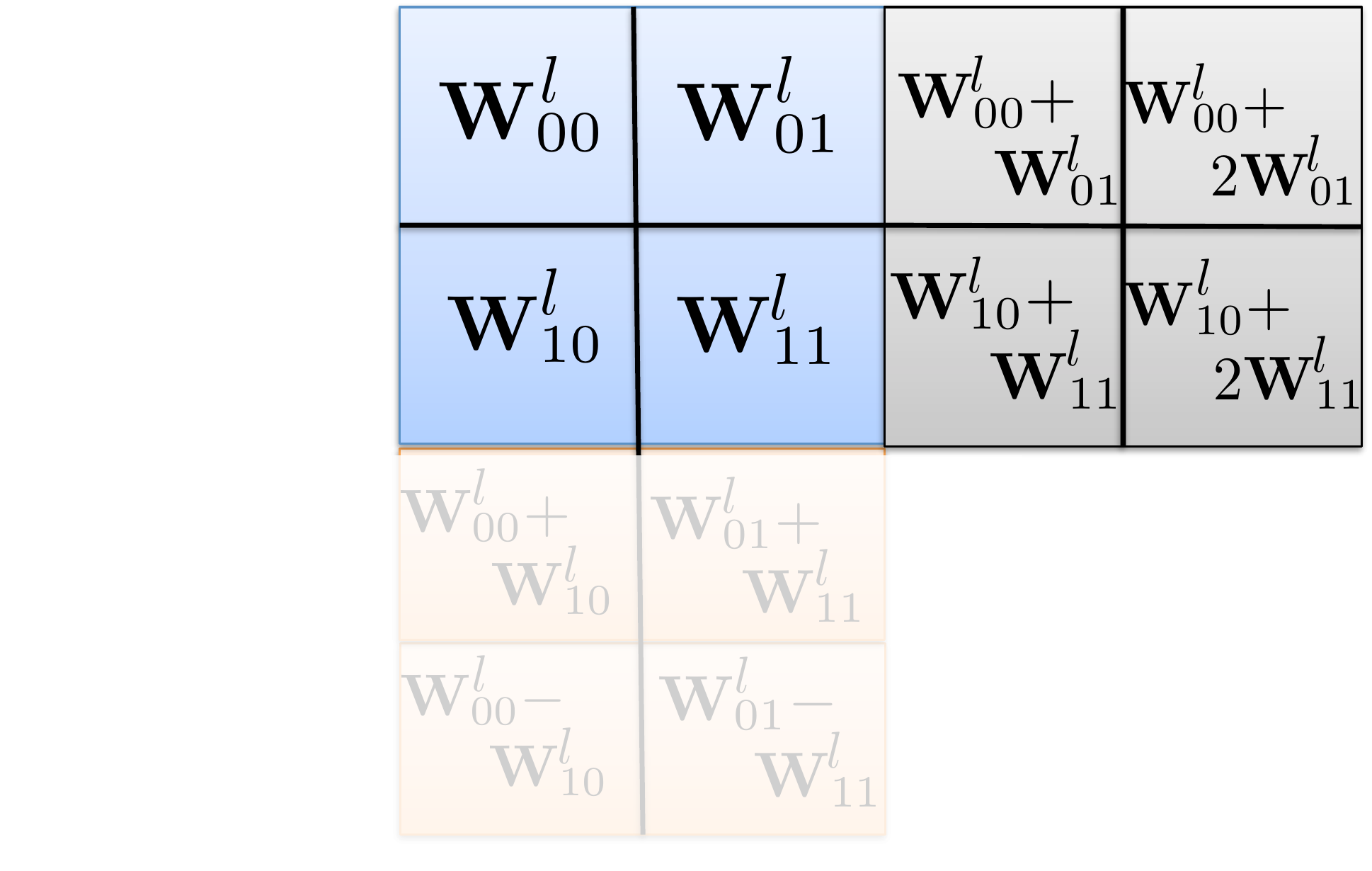}}
\caption{MDS-code-based strategy for DNN training \cite{dutta2018DNN1}: Original Matrix $\bm{W}$ is divided into $2 \times 2$ blocks and then encoded using two systematic MDS codes for the row and column blocks respectively. The redundant nodes due to any one MDS code is active in the feedforward or backpropagation stage.}
\label{fig:MDS_strategy}
\end{figure}
\begin{lem}[Error Tolerances for MDS-code-based Strategy]
The error tolerances for the MDS-code-based strategy are as follows: $t_f=\lfloor \frac{P_f-mn}{2n} \rfloor ,\ t_b=\lfloor \frac{P_b-mn}{2m} \rfloor $ under Error Model $1$ and $t_f= \frac{P_f-mn-n}{n},\  t_b= \frac{P_b-mn-m}{m}  $ under Error Model $2$.
\label{lem:mds_error_tolerance}
\end{lem}

\section{Our Proposed coded DNN Training Strategy for  mini-batch $B=1$}
\label{sec:coded_DNN_MV}



In this section, we introduce our proposed unified coded DNN training strategy. We propose an initial encoding scheme for $\bm{W}$ at each layer such that the same encoding allows us to perform coded ``post'' and ``pre'' multiplication of $\bm{W}$ with vectors $\bm{x}$ and $\bm{\delta}^T$ respectively at each layer in every iteration. The key idea is that we encode $\bm{W}$ only for the first iteration. For all subsequent iterations, we encode and decode vectors (hence complexity $o(\frac{N^2}{K})$ as we show in \Cref{thm:complexity_MV}) instead of matrices. As we will show, the encoded weight matrix $\bm{W}$ is able to update itself, maintaining its coded structure at very low additional overhead. 


\noindent \textbf{Initial Encoding of $\bm{W}$ (Pre-processing Step):} 
    Every node stores an $\frac{N}{m}\times \frac{N}{n}$ sub-matrix (or block) of $\bm{W}$ encoded using Generalized PolyDot. Recall from \Cref{eq:encoded_W}) that, 
    \begin{equation}
    \label{eq:initial_encoded_W}
    \widetilde{\bm{W}}(u,v)= \sum_{i=0}^{m-1} \sum_{j=0}^{n-1} \bm{W}_{i,j}u^i v^j.
    \end{equation}
For $p=0,1,\ldots,P-1$, node $p$ stores $\widetilde{\bm{W}}_p := \widetilde{\bm{W}}(u,v)|_{u=a_p,v=b_p}$, \textit{i.e.}, the evaluation of $\widetilde{\bm{W}}(u,v)$ at $(u,v)=(a_p,b_p)$,  at the beginning of the training. This coded sub-matrix has $\frac{N^2}{K}$ entries. Thus every node stores a sub-matrix $\widetilde{\bm{W}}_p = \widetilde{\bm{W}}(a_p,b_p)$ at the beginning of the training which has $\frac{N^2}{K}$ entries. 

\textit{Encoding of matrix $\bm{W}$ is done only before the first iteration.}
 

\noindent \textbf{Feedforward stage:} 
Assume that the entire input $\bm{x}$ to the layer is made available at every node by the previous layer (this assumption is justified at the end of this paragraph). Also assume that the updated $\widetilde{\bm{W}}_p$ of the previous iteration is available at every node (this assumption will be justified when we show that the encoded sub-matrices of $\bm{W}$ are able to update themselves, preserving their coded structure).

For $p=0,1,\ldots,P-1$, node $p$ first block-partitions $\bm{x}$ into $n$ equal parts, and encodes them using the polynomial:
\begin{equation}
\widetilde{\bm{x}}(v)= \sum_{j=0}^{n-1} \bm{x}_j v^{n-j-1}.
\end{equation}
For $p=0,1,\dots,P-1$, the $p$-th node \textbf{evaluates} the polynomial $\widetilde{\bm{x}}(v)$ at $v=b_p$, yielding $\widetilde{\bm{x}}_p:=\widetilde{\bm{x}}(b_p)$. E.g., for $n=2$, $\bm{x}$ is encoded as $ \widetilde{\bm{x}}(v) = \bm{x}_0 v + \bm{x}_1  $.

Next, each node \textbf{computes} the matrix-vector product: $\widetilde{\bm{s}}_p:=\widetilde{\bm{W}}_p\widetilde{\bm{x}}_p$. The computation of $\widetilde{\bm{s}}_p$ at node $p$ is equivalent to the evaluation, at $(u,v)=(a_p,b_p)$, of the following polynomial:
\begin{equation}
\widetilde{\bm{s}}(u,v) = \widetilde{\bm{W}}(u,v) \widetilde{\bm{x}}(v) =  \sum_{i=0}^{m-1} \sum_{j=0}^{n-1} \sum_{j'=0}^{n-1}\bm{W}_{i,j} \bm{x}_{j'} u^i v^{n-1+j-j'}
\label{eq:s_polynomial}
\end{equation}
even though the node is \textit{not explicitly evaluating it by accessing all its coefficients separately}.
Now, fixing $j'=j$, observe that the coefficient of $u^i v^{n-1}$ for $i=0,1,\dots,m-1$ is $\sum_{j=0}^{n-1} \bm{W}_{i,j} \bm{x}_j = \bm{s}_i$. Thus, these $m$ coefficients constitute the $m$ sub-vectors of $\bm{s}=\bm{W}\bm{x}$. Therefore, $\bm{s}$ can be recovered at any node if it can reconstruct all the coefficients of the polynomial $\widetilde{\bm{s}}(u,v)$ in \Cref{eq:s_polynomial}, or rather \textit{just these $m$ coefficients}.
The matrix-vector product computed at the $p$-th node results in the evaluation of this polynomial at $(u,v)=(a_p,b_p)$. Every node then sends this product to every other node\footnote{Recall that there is no single master node, and every node replicates the functioalities of the decoder. Using efficient all-to-all communication protocols~\cite{chan2007collective,bruck1997efficient} popular in parallel computing, the communication cost of all nodes broadcasting its own sub-vector $\widetilde{\bm{s}}_p$ of length $\frac{N}{m}$ to all other nodes has a communication cost of $\alpha \log{P} +2\beta \frac{N}{m}P  = \Theta(\frac{N}{m}P )$ using All-Gather protocol. We are currently examining strategies to reduce this cost further.} where some of these products may be erroneous. Now, if every node can still reconstruct the coefficients of $u^i v^{n-1}$ from these evaluations, then it can successfully decode $\bm{s}_0,\bm{s}_1,\ldots,\bm{s}_{m-1}$. 

We use one of the substitutions $u=v^n$ or $v=u^m$ (elaborated in \Cref{appendix:error_tolerance}), to convert $\widetilde{\bm{s}}(u,v)$ into a polynomial in a single variable and then use standard decoding techniques\cite{candes2005decoding} to interpolate the coefficients of a polynomial in one variable from its evaluations at $P$ arbitrary points when some  evaluations have an additive error. Once $\bm{s}$ is decoded at each node, the nonlinear function $f(\cdot)$ is applied element-wise to generate the input for the next layer. This also makes $\bm{x}$ available at every node at the start of the next feedforward layer, justifying our assumption.

\textit{Regeneration:} If the number of errors are few ($\leq$ error tolerance), the nodes are not only able to decode the vectors correctly but also locate which nodes were erroneous (see \Cref{appendix:error_tolerance}). Thus, the encoded $\bm{W}$ stored at those nodes are \textbf{regenerated}\footnote{The encoded matrix at any node is the evaluation of a polynomial whose coefficients correspond to the original sub-matrices $\bm{W}_{i,j}$. Thus, the number of nodes required by an error-prone node is the degree of this polynomial $ + 1$. Substituting $u=v^n$ (alternatively, $v=u^m$), this degree is $mn-1$, and thus an error-prone node needs to access $mn$ correct nodes to regenerate itself.} by accessing some of the nodes that are known to be correct and the algorithm proceeds forward.  

\begin{rem} It might appear that regeneration violates the storage constraint of each node, \textit{i.e.}, a storage of only a $\frac{1}{K}$ fraction of matrix $\bm{W}$. However, note that the coded sub-matrix can also be computed element-wise rather than all at once, while adhering to the storage constraint. E.g.~if the stored sub-matrix $\widetilde{\bm{W}}(a_p, b_p)$ of size $\frac{N}{m}\times \frac{N}{n}$ is required to be regenerated, the node can access only the first element (location $(0,0)$) of the stored sub-matrix of any $mn$ nodes to regenerate the element $(0,0)$ of $\widetilde{\bm{W}}(a_p, b_p)$. Then it deletes all the gathered values, and moves on to location $(0,1)$ and so on. This process only requires an additional storage of $o\left( \frac{N^2}{K} \right)$.
\end{rem}

\textit{Additional Steps (Under Error Model $2$):} Under Error Model $2$, errors can occur in all the primary steps, and are unbounded. Similar to replication and MDS-code-based strategy, the DNN is \textbf{checkpointed} at a disk at regular intervals. If there are more errors than the error tolerance after steps $O1$ or $O2$, the nodes are unable to decode correctly. However, as the error is assumed to be additive and drawn from real-valued, continuous distributions, the occurrence of errors is still detected with probability $1$ even though they cannot be located or corrected, and thus the entire DNN can again be \textit{restored from the last checkpoint}.

To allow for \textbf{decoding errors}, we need to include one more verification step. This step is similar to the replication strategy where all nodes exchange some particular values of the decoded vector, \textit{i.e.}, say any $P$ pre-decided values of the decoded vector $\bm{s}$ and compare (additional communication overhead of $\alpha \log{P} + 2\beta P^2 $ and computation overhead of $\Theta(P^2)$ as discussed in \Cref{appendix:complexity}). Again, it is unlikely that two nodes will have the exact same decoding error. If there is a disagreement at one or more nodes during this process, we assume that there has been errors during the decoding, and the entire neural network is restored from the last checkpoint. Because the total complexity of this verification step is low in scaling sense compared to encoding/decoding or communication (because it does not depend on $N$), we assume that it is error-free since the probability of soft-errors occurring within such a small duration is negligible as compared to other computations of longer duration.


Under Error Model $2$, errors can also occur in the step $C1$ or during encoding. If an error occurs during step $C1$, the vector $\bm{x}$ for the next layer is corrupted, which ultimately corrupts the encoded sub-vector of that node, \textit{i.e.}, $\widetilde{\bm{x}}_p$ for the next layer. Errors during encoding will also corrupt this sub-vector $\widetilde{\bm{x}}_p$. Then, the error propagates into $\widetilde{\bm{s}}_p$ during the matrix-vector product $\widetilde{\bm{s}}_p=\widetilde{\bm{W}}_p\widetilde{\bm{x}}_p$ at the next layer and is finally detected and if possible corrected after step $O1$ of the next layer in the same iteration, when every node attempts to decode $\bm{s}$ from all its received $\widetilde{\bm{s}}_p$ sub-vectors, some of which may be erroneous.

\textbf{Backpropagation stage:} 
The backpropagation stage is very similar to the feedforward stage. The backpropagated error (transpose) $\bm{\delta}^T$ is available at every node. Each node partitions the row-vector $\bm{\delta}^T$ into $m$ equal parts and encodes them using the polynomial: 
\begin{equation}
\widetilde{\bm{\delta}}^T(u) = \sum_{i=0}^{m-1} \bm{\delta}^T_i u^{m-i-1}.
\end{equation}
For $p=0,1,\dots,P-1$, the $p$-th node \textbf{evaluates} $\widetilde{\bm{\delta}}^T(u)$ at $ u= a_p$, yielding $\widetilde{\bm{\delta}}^T_p:= \widetilde{\bm{\delta}}^T(a_p) $. Next, it \textbf{performs} the computation $\widetilde{\bm{c}}^T_p:= \widetilde{\bm{\delta}}^T_p \widetilde{\bm{W}}_p$ and sends the product to all other nodes, of which some products may be erroneous. Consider the polynomial:
$$ \widetilde{\bm{c}}^T(u,v)= \widetilde{\bm{\delta}}^T(u)\widetilde{\bm{W}}(u,v) 
= \sum_{i'=0}^{m-1} \sum_{i=0}^{m-1} \sum_{j=0}^{n-1} \bm{\delta}_{i'}^T\bm{W}_{i,j}  u^{m-1+i-i'} v^j.
$$
The products computed at each node result in the evaluations of this polynomial $\widetilde{\bm{c}}^T(u,v)$ at $(u,v)=(a_p,b_p)$. Similar to feedforward stage, each node then decodes the coefficients of $u^{m-1}v^j$ in the polynomial for $j=0,1,\dots,n-1$, and thus reconstructs $n$ sub-vectors forming $\bm{c}^T$.

\textit{Additional Steps (Under Error Model $2$):} The additional steps of checkpointing and verification step to check for decoding errors is similar. Errors during step $C2$ or during encoding will corrupt the encoded sub-vector $\widetilde{\bm{\delta}}^T_p$, which will eventually show up after the computation $\widetilde{\bm{c}}^T_p= \widetilde{\bm{\delta}}^T_p \widetilde{\bm{W}}_p$, and will be corrected after step $O2$ when each node attempts to reconstruct $\bm{c}^T$ from the outputs $\widetilde{\bm{c}}^T_p$ of all the nodes, of which some may be erroneous.

\textbf{Update stage:} 
The key part is \emph{updating} the coded $\widetilde{\bm{W}}_p$. Observe that since $\bm{x}$ and $\bm{\delta}$ are both available at each node, it can encode the vectors as $\sum_{i=0}^{m-1}\bm{\delta}_i u^i $ and $\sum_{j=0}^{n-1}\bm{x}_j v^j $ at $u=a_p$ and $v=b_p$ respectively, and then update itself as follows:
\begin{align}
\widetilde{\bm{W}}_p &\leftarrow \widetilde{\bm{W}}_p + \eta (  \sum_{i=0}^{m-1}\bm{\delta}_i a_p^i )( \sum_{j=0}^{n-1}\bm{x}_j b_p^j  )^T \nonumber \\
& = \sum_{i=0}^{m-1}\sum_{j=0}^{n-1} \underbrace{( \bm{W}_{i,j} + \eta \bm{\delta}_i \bm{x}^T_j )}_{\text{Update of } \bm{W}_{i,j}} a_p^i b_p^j. \label{eq:coded_update}
\end{align}
Thus, the update step preserves the coded nature of the weight matrix, with negligible additional overhead (see \Cref{thm:complexity_MV}). 

\textbf{Update with Regularization:} When training is performed using L2 regularization, recall that the original update rule is modified as: $\bm{W}\leftarrow (1-\eta\lambda)\bm{W} + \bm{\delta}\bm{x}^T$. Adding the weight decay term only requires a minor change in the update step of coded DNN training. Now, the update equation in \eqref{eq:coded_update} can be modified as:
\begin{align}
\widetilde{\bm{W}}_p &\leftarrow (1-\eta\lambda) \widetilde{\bm{W}}_p + \eta (  \sum_{i=0}^{m-1}\bm{\delta}_i a_p^i )( \sum_{j=0}^{n-1}\bm{x}_j b_p^j  )^T \nonumber \\
& = \sum_{i=0}^{m-1}\sum_{j=0}^{n-1} \underbrace{( (1-\eta\lambda) \bm{W}_{i,j} + \eta \bm{\delta}_i \bm{x}^T_j )}_{\text{Update of } \bm{W}_{i,j}} a_p^i b_p^j. 
\end{align}
We do not need additional encoding or decoding for introducing the weight decay term. Instead, we only have to shrink the already encoded $\bm{W}$ matrix, $\widetilde{\bm{W}}_p$, by $(1-\eta\lambda)$ after each iteration. For a detailed discussion on regularization in DNN training, the reader is referred to \Cref{subsec:app_reg}.

\textit{Errors during Update stage:} Errors can occur in the update stage under both the error models. These errors corrupt the updated sub-matrix $\widetilde{\bm{W}}_p$ and then show up in the computation $\widetilde{\bm{s}}_p=\widetilde{\bm{W}}_p\widetilde{\bm{x}}_p$ at the same layer in the next iteration. So, these errors are finally detected and if possible corrected after step $O1$ at that layer in the next iteration when every node attempts to decode $\bm{s}$ from all its received $\widetilde{\bm{s}}_p$ sub-vectors, which may be erroneous.






\section{Results on Performance of our Proposed Strategy}
\label{sec:comparison_DNN_MV}
In this section, we show that $\mathcal{C}_{\mathrm{GP}}(K,N,P)$ has better worst case error tolerances than $\mathcal{C}_{\mathrm{mds}}(K,N,P)$ and $\mathcal{C}_{\mathrm{rep}}(K,N,P)$ by a factor that can diverge to infinity while the communication and computation overheads of the proposed strategy remains negligible. The comparison is formalized in \Cref{thm:error_tolerance_MV} followed by the characterization of the additional overheads in \Cref{thm:complexity_MV}.
{\color{black}{
\begin{thm}[Error tolerances $(t_f,t_b)$]
\label{thm:error_tolerance_MV}
The error tolerances $(t_f,t_b)$ at each layer for the three strategies $\mathcal{C}_{\mathrm{GP}}(K,N,P)$, $\mathcal{C}_{\mathrm{mds}}(K,N,P)$ and $\mathcal{C}_{\mathrm{rep}}(K,N,P)$ under Error Models $1$ and $2$ are given by \Cref{table_example}.
\end{thm}

\begin{table}[!htbp]
\centering
\caption{Error Tolerances $(t_f,t_b)$ under fixed number of nodes $P$ }
\label{table_example}
\centering
\begin{tabular}{|M{2.7cm}|M{5cm}|M{5cm}|}
\hline
Strategy & Error Model $1$ ($t_f,t_b$) & Error Model $2$ ($t_f,t_b$)\\
\hline
$\mathcal{C}_{\mathrm{GP}}(K,N,P)$ with $u=v^{n} $ & $\left(\frac{P-mn-n+1}{2}, \frac{P-2mn+n}{2} \right)$  & $\left(P-mn-n, P-2mn+n-1 \right)$
\\
\hline
$\mathcal{C}_{\mathrm{GP}}(K,N,P)$ with $v=u^{m}$ &   $\left( \frac{P-2mn+m}{2},\frac{P-mn-m+1}{2} \right) $ & $\left( P-2mn+m-1,P-mn-m \right)$ \\
\hline
$\mathcal{C}_{\mathrm{mds}}(K,N,P)$ where $P=P_f+P_b-mn$ & $\left(\frac{P_{f}-mn}{2n},   \frac{P_{b}-mn}{2m} \right)$& $\left(\frac{P_{f}-mn-n}{n},   \frac{P_{b}-mn-m}{m} \right)$ \\
\hline
$\mathcal{C}_{\mathrm{rep}}(K,N,P)$ & $\left(\frac{P-mn}{2mn}, \frac{P-mn}{2mn}\right) $ & $\left(\frac{P-2mn}{mn}, \frac{P-2mn}{mn}\right) $ \\
\hline
\end{tabular}
\end{table}
}}
\begin{rem} We ignore integer effects here as we are primarily interested in error tolerance with scaling $P$. However, strictly speaking, we need a floor function $ \lfloor \cdot \rfloor $ applied to all of the expressions for Error Model $1$ above, and $mn$ to divide $P$ for replication for both the error models.
\end{rem}

\begin{rem} 
Note that, for the proposed coded DNN training strategy, the errors in the update stage (step $O3$) are also corrected after step $O1$ at that layer, in the next iteration along with the other errors during step $O1$. Thus, under Error Model $1$, we require $t_f$ to be greater than $t_1+t_3$ while $t_b$ is only required to be greater than $t_2$. Because the computational complexities of steps $O1,O2$ and $O3$ are similar, one might expect that $t_1,t_2$ and $t_3$ are nearly the same. Under Error Model $2$, when the errors are more than $t_f$ or $t_b$, the strategy is able to detect errors with probability $1$ but not correct them. However, even under Error Model $2$, it is more desirable to have $t_f > t_b$ since more errors are likely to occur after step $O1$ as compared to step $O2$ since the total duration of computation that is covered is more.
\end{rem}

\begin{coro}[Scaling Sense Comparison]
\label{coro:scaling_sense_comparison}
Consider the regime $m=n=\sqrt{K}$. Then, for both the error models, the ratio of $t_f$ (or $t_b$) for $\mathcal{C}_{\mathrm{GP}}(K,N,P)$ with $\mathcal{C}_{\mathrm{mds}}(K,N,P)$ and $\mathcal{C}_{\mathrm{rep}}(K,N,P)$ scales as $\Theta(\sqrt{K})$ and $\Theta(K)$ respectively as $P \to \infty $.
\end{coro}

The proofs of \Cref{thm:error_tolerance_MV} and \Cref{coro:scaling_sense_comparison} are in \Cref{appendix:error_tolerance}. In \Cref{fig:tolerance_region}, we show that Generalized PolyDot achieves the best $(t_f,t_b)$ tradeoff compared to the other existing schemes. Now we formally show in \Cref{thm:complexity_MV} that $\mathcal{C}_{\mathrm{GP}}(K,N,P)$ also satisfies the desired properties of adding negligible overhead at each node.

\begin{figure}
\centering
\includegraphics[width=6cm]{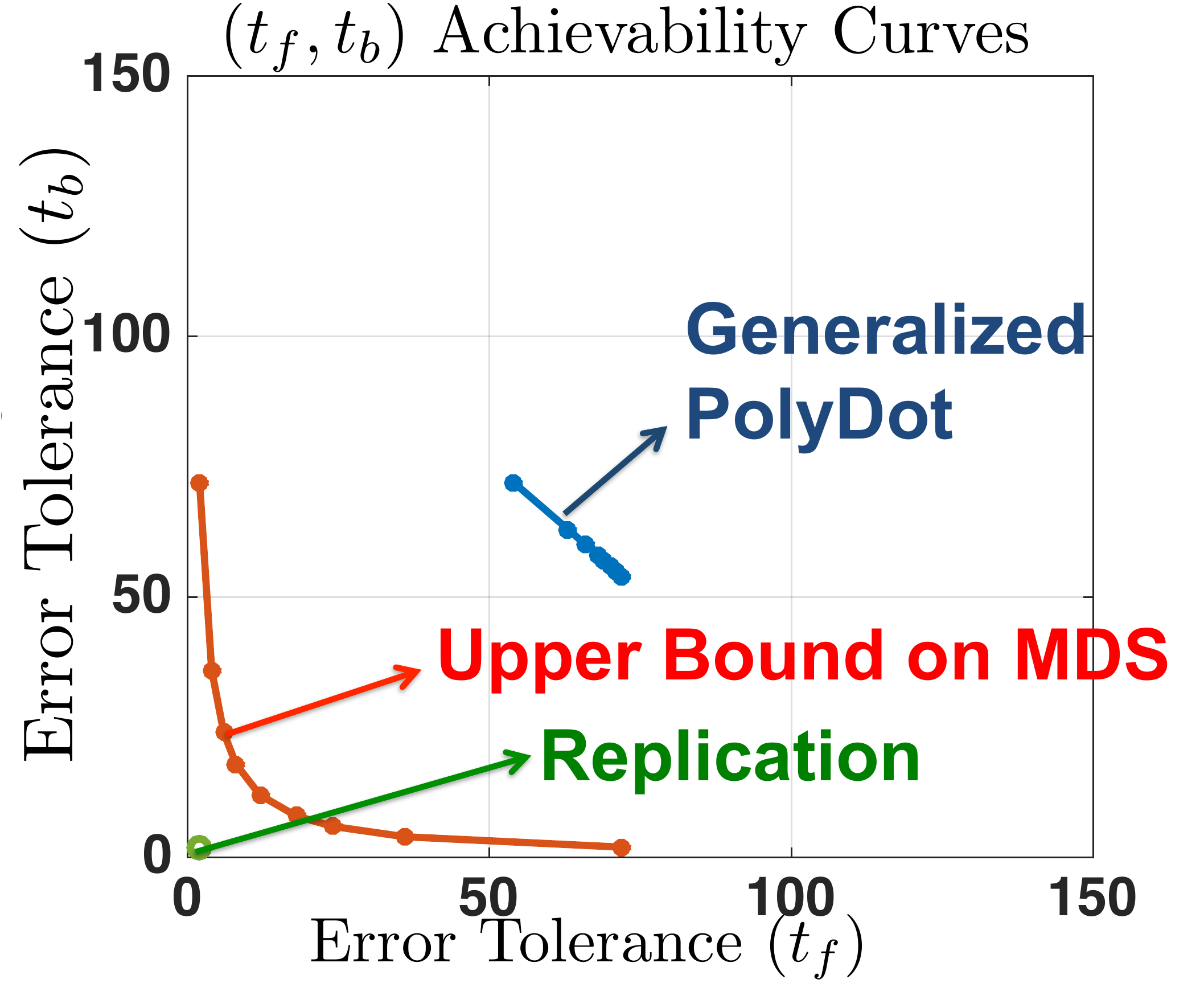}
\hspace{1cm}
\includegraphics[width=6cm]{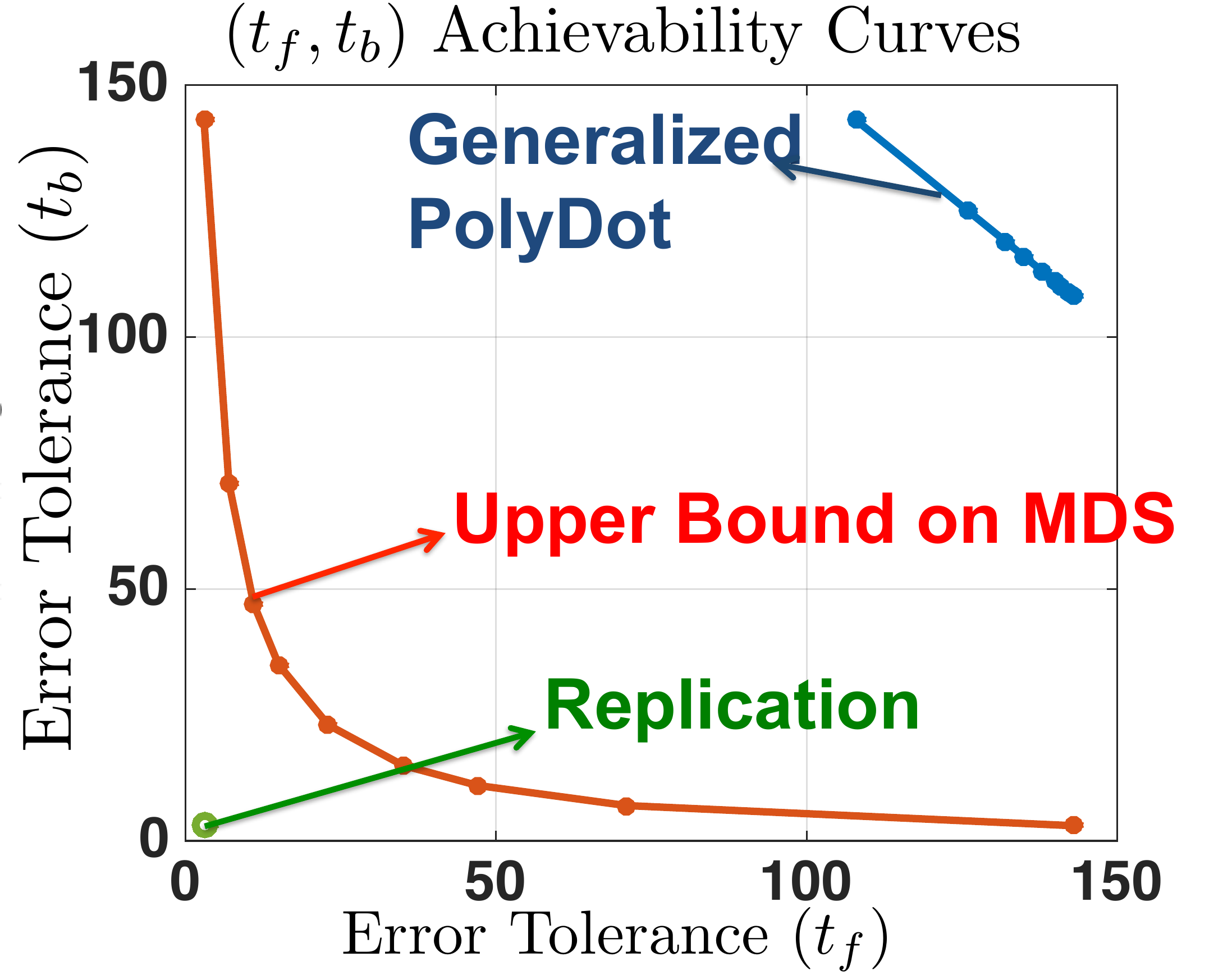}
\caption{Error tolerance regions under Error Models $1$ (Left) and $2$ (Right): We choose $P= 180$, $K=36$ and vary $m$ and $n$. For the MDS-code-based strategy, we plot an upper bounds on $t_f,t_b$ using $P_f,P_b \leq P$. Generalized PolyDot (with $u=v^n$) achieves the best $(t_f,t_b)$ tradeoff. Choosing $v=u^m$ only interchanges $(t_f,t_b)$ and thus it also gives same curve.}
\label{fig:tolerance_region}
\end{figure}

\begin{thm}[Complexity of $\mathcal{C}_{\mathrm{GP}}(K,N,P)$]
\label{thm:complexity_MV}
For $\mathcal{C}_{\mathrm{GP}}(K,N,P)$  at any layer in a single iteration, the ratio of the total complexity of all the steps including encoding, decoding, communication, nonlinear activation, Hadamard product etc. to the most complexity intensive steps (matrix-vector products $O1$ and $O2$ and update $O3$) tends to $0$ as $K,N,P \to \infty$ if the number of nodes satisfy $P^4=o(N)$.
\end{thm}

The proof is provided in \Cref{appendix:complexity}. For this proof, we assume a pessimistic bound ($\Theta(P^3)$) on the decoding of a code of block length $P$ under errors, based on sparse reconstruction algorithms\cite{candes2005decoding}. We are currently examining the reduction of this complexity using other algorithms, which would also relax the condition of \Cref{thm:complexity_MV}. 

In \Cref{table_complexity}, we finally list out the storage, communication and computation costs of $\mathcal{C}_{\mathrm{GP}}(K,N,P)$. The derivation of communication and computation costs is further elaborated in \Cref{appendix:complexity}, in the proof of \Cref{thm:complexity_MV}.

\begin{table}[!ht]
\centering
\caption{Storage, Communication and Computation Costs for Each Layer (All Stages Combined)}
\label{table_complexity}
\centering
\begin{tabular}{|M{2.7cm}|M{3cm}|M{4cm}|M{5cm}|}
\hline
 Storage & Communication Complexity & Computational Complexity of steps $O1$, $O2$ and $O3$   & Computational Complexity of all other steps (including Encoding/Decoding) \\
\hline
 $\frac{N^2}{K}+ \Theta(\frac{PN}{m}+\frac{PN}{n})  $  & $\Theta((\frac{N}{m}+\frac{N}{n})P)$   & $ \Theta(\frac{N^2}{K}) $ &  $\mathcal{O}((\frac{N}{m}+\frac{N}{n})P^3)$  \\
\hline
\end{tabular}
\end{table}

\begin{itemize}
\item Storage: In $\mathcal{C}_{\mathrm{GP}}(K,N,P)$, each node stores a fraction $\frac{1}{K}$ of the matrix $\bm{W}$ of size $N \times N$ which contributes the term $\frac{N^2}{K}$. However, during decentralized decoding in $\mathcal{C}_{\mathrm{GP}}(K,N,P)$, all nodes receive partial computation results of sizes $\frac{N}{m}$ (Feedforward stage - before step $C1$) or $\frac{N}{n}$ (Backpropagation stage - before step $C2$) from all other $P$ nodes, leading to the term $\Theta(\frac{PN}{m}+\frac{PN}{n})$. 
\item Communication: For $\mathcal{C}_{\mathrm{GP}}(K,N,P)$, every node broadcasts its own partial computation result $\widetilde{\bm{s}}_p$ (or $\widetilde{\bm{c}}^T_p$) to all other $P$ nodes. This leads to a communication cost of $\left(\alpha \log{P} + 2\beta \frac{NP}{m} \right) + \left(\alpha \log{P} +2\beta \frac{NP}{n}\right) = \Theta \left( (\frac{N}{m}+\frac{N}{n})P \right)$ when performed using an efficient All-Gather communication protocol (see \Cref{appendix:complexity} and also in \cite{chan2007collective}). Thus, the communication cost is smaller in scaling sense as compared to the computational complexity of the steps $O1$, $O2$ and $O3$ which is what we had desired. As a future work, we are exploring the reduction of this communication cost further by using efficient implementation strategies.

\item Computation: The most dominant computational complexity, \textit{i.e.}, the complexity of matrix-vector products and rank-1 update at each node is $\Theta(\frac{N^2}{K})$. Among the additional steps,  the decoding is the most dominant in terms of computational complexity. The decoding of codewords of length $P$ requires a complexity of $\mathcal{O}(P^3)$, and for  $\mathcal{C}_{\mathrm{GP}}(K,N,P)$, this is repeated $\frac{N}{m}$ or $\frac{N}{n}$ times in the feedforward or backpropagation stages respectively. 
\end{itemize}

\begin{rem}
Note that we have used pessimistic bounds to characterize the communication and computational overheads of our strategy and in spite of that, we are able to show that the additional overheads are negligible in scaling sense as compared to the computational complexities of steps $O1$, $O2$ and $O3$ respectively. We are currently exploring the reduction of these additional overheads further using efficient implementation strategies. 
\end{rem}

\section{Extension to mini-batch size $B>1$: coded matrix-matrix multiplication}
\label{sec:extension}
So far, we have discussed DNN training using SGD for the special case of mini-batch size $B=1$, \textit{i.e.}, when the DNN accesses only a single data point at a time. In this section, we extend our proposed strategy to the general case of mini-batch size $B>1$, which leads to coded \emph{matrix-matrix} multiplication.


The operations of DNN training for the general case of $B>1$ have already been stated in \Cref{sec:problem_formulation}. The main difference from the case of $B=1$ is that the input to layer $l$ is a matrix $\bm{X}^l$ of dimensions $N \times B$. Similarly, the backpropagated error is also a matrix $\bm{\Delta}^l$ of dimensions $N \times B$. Because the operations are the same across all the layers, we will again omit the superscript $(\cdot)^l$ in the subsequent discussion.

\textbf{Goal:} The goal in this section is to design a coded DNN training strategy for mini-batch size $B>1$, denoted by $\mathcal{C}(N,K,P,B)$, using $P$ nodes such that every node can store only a $ \frac{1}{K}$ fraction of the entries of $\bm{W}$ for each layer. We assume that $B=o(\frac{N}{K})$ to ensure that the sizes of $\bm{X}$, $\bm{S}(=\bm{W}\bm{X})$, $\bm{\Delta}^T$ and $\bm{C}^T(= \bm{\Delta}^T\bm{W})$ are smaller in scaling sense than the size of $\frac{1}{K}$ fraction of $\bm{W}$. Thus, every node has a total storage of $\frac{L N^2}{K}+ o(\frac{L N^2}{K})$ where the small additional storage of $o(\frac{L N^2}{K})$ is for storing $\bm{X}$, $\bm{S}(=\bm{W}\bm{X})$, $\bm{\Delta}^T$ and $\bm{C}^T(= \bm{\Delta}^T\bm{W})$ respectively for every layer\footnote{If we do not assume an upper bound on $B$, then as $B$ increases, the allowed total storage per node would also be required to increase. Then, it may become possible to store more than $\frac{1}{K}$ fraction of $\bm{W}$ leading to alternative strategies altogether. This problem may be considered as a future work.}. Similar to the coded DNN strategy for $B=1$, the additional computation and communication complexities including encoding/decoding overheads in each iteration should be negligible in scaling sense as compared to the local computational complexity of the steps $O1$, $O2$ and $O3$ parallelized across each node, at any layer.

Thus, essentially we are required to perform distributed ``post'' and ``pre'' multiplication of the same matrix $\bm{W}$ of dimensions $N \times N$ with matrices $\bm{X}$ and $\bm{\Delta}^T$ respectively and distributed update $\bm{W}+\eta \bm{X}\bm{\Delta}^T$, along with all the other operations of training. Similar to the case of $B=1$, because outputs are only communicated to other nodes after steps $O1$ and $O2$ respectively, we aim to correct as many erroneous nodes as possible after these two steps, before moving to another layer.


\begin{rem}
The assumption that $B=o(\frac{N}{K})$ is required only to satisfy the storage constraints. Irrespective of whether $B=o(\frac{N}{K})$ or $B=\Omega(\frac{N}{K})$, when considering the total computational or communication complexity, the matrix-matrix products and updates are still the most significant costs, and all other additional overheads including communication, encoding/decoding etc. add negligible overhead, as we will also formally show in \Cref{thm:complexity_MM}, provided that the number of processing nodes are not too large, \textit{i.e.}, $P^4=o(N)$.
\end{rem}

\subsection{Modification to Existing Strategies.}

\textbf{Replication ($\mathcal{C}_{\mathrm{rep}}(K,N,P,B)$):} 

We modify the replication strategy of \Cref{sec:existing_strategies} for $B>1$. The matrix $\bm{W}$ is block-partitioned and stored in a manner similar to that for  mini-batch $B=1$ with $\frac{P}{mn}$ replicas. However, instead of a single vector $\bm{x}$, now the entire data matrix $\bm{X}$ is divided vertically into $n$ blocks, each of size $\frac{N}{n} \times B$. Similarly, instead of vector $\bm{\delta}^T$, now the matrix $\bm{\Delta}^T$ is divided into $m$ equal parts horizontally, each of size $ B \times \frac{N}{m}$. All the operations that were performed on the $i$-th part of $\bm{x}$ or $\bm{\delta}^T$, are now performed on the corresponding block of $\bm{X}$ or $\bm{\Delta}^T$ respectively. The additional computational and communication complexity of this strategy is obviously higher than the case with $B=1$, but they still satisfy the desired constraint that they should be smaller in scaling sense than the per-node computational complexity of the steps $O1$, $O2$ and $O3$.

\textbf{MDS-code-based Strategy ($\mathcal{C}_{\mathrm{mds}}(K,N,P,B)$):}

The proposed MDS-code-based strategy can also be modified in a similar manner as replication. The matrix $\bm{W}$ is encoded and stored in a manner similar to that for the case of mini-batch size $B=1$. However, the matrices $\bm{X}$ and $\bm{\Delta}^T$ are now divided into $n$ and $m$ parts respectively, similar to replication, and the same operations are performed as in the case of mini-batch size $B=1$. The additional computational and communication complexity of this strategy is still smaller in scaling sense than the per-node computational complexity of $O1$, $O2$ and $O3$.

\subsection{Our Proposed coded DNN Training Strategy for  mini-batch size $B>1$.}
\label{sec:coded_DNN_MM}



Here, we apply Generalized PolyDot codes for the matrix-matrix multiplication, as discussed previously in \Cref{thm:gen_poly_MM}. Let us begin by reviewing the notations for matrix partitioning.

The weight matrix $\bm{W}$ is partitioned into a grid of $m \times n$ blocks as before. The main difference from the case of $B=1$ is that the input to the layer $\bm{X}$ and backpropagated error $\bm{\Delta}$ are matrices now. The matrix $\bm{X}$ is also block partitioned into a grid of $n \times d_1$ blocks, each sub-matrix (or block) being of size $\frac{N}{n} \times \frac{B}{d_1} $. 
This results in the matrix $\bm{S}=\bm{W}\bm{X}$ to be partitioned into a grid of $m \times d_1$ blocks, with each sub-matrix (or block) being of size $\frac{N}{m}\times \frac{B}{d_1} $, such that $\bm{S}_{i,k}=\sum_{j=0}^{n-1}\bm{W}_{i,j}\bm{X}_{j,k}.$


Similarly, $\bm{\Delta}^T$ is also partitioned into a grid of $d_2 \times m $ blocks. 
This in turn results in $\bm{C}^T=\bm{\Delta}^T \bm{W}$ to be partitioned into $d_2 \times n $ blocks, such that $[\bm{C}^T]_{k,j}=\sum_{i=0}^{m-1}[\bm{\Delta}^T]_{k,i} \bm{W}_{i,j}.$

\textbf{Initial Encoding of $\bm{W}$ (Pre-processing step):} The initial encoding of $\bm{W}$ is similar to that for the case of mini-batch size $B=1$ (see \Cref{eq:initial_encoded_W} in \Cref{sec:coded_DNN_MV}).

\textbf{Feedforward stage:} Similar to the case of $B=1$, we assume that the entire input $\bm{X}$ is made available at every node by the previous layer\footnote{We will discuss storage and communication tradeoffs in \Cref{subsec:comparison_DNN_MM}}. Each node first partitions the matrix $\bm{X}$ into a grid of $n \times d_1$ blocks, each sub-matrix being of size $\frac{N}{n} \times \frac{B}{d_1} $, and then encodes as follows:
$$\widetilde{\bm{X}}(v,w)= \sum_{j=0}^{n-1}\sum_{k=0}^{d_1-1}\bm{X}_{j,k}v^{n-1-j}w^k  .$$

For $p=0,1,\dots,P-1$, the $p$-th node \textbf{evaluates} the polynomial $\widetilde{X}(v,w)$ at $(v,w)=(b_p,c_p)$, yielding $\widetilde{\bm{X}}_p=\widetilde{\bm{X}}(b_p,c_p)$. E.g., for $n=d_1=2$, $\bm{X}$ is encoded as $ \widetilde{\bm{X}}(v,w)=\bm{X}_{0,0}v + \bm{X}_{0,1} +\bm{X}_{1,0}wv +\bm{X}_{1,1}w$. Next, each node \textbf{computes} the matrix-matrix product: $\widetilde{\bm{S}}_p:=\widetilde{\bm{W}}_p \widetilde{\bm{X}}_p  $. Computing $\widetilde{\bm{S}}_p$ is effectively resulting in the evaluation, at $(u,v,w)=(a_p,b_p,c_p)$, of the following polynomial:
$$\widetilde{\bm{S}}(u,v,w)= \widetilde{\bm{W}}(u,v)\widetilde{\bm{X}}(v,w) = \sum_{i=0}^{m-1}\sum_{j=0}^{n-1} \sum_{j'=0}^{n-1}\sum_{k=0}^{d_1-1} \bm{W}_{i,j}\bm{X}_{j',k}u^i v^{n-1+j-j'}w^{k}, $$
even though the node is not explicitly evaluating it by accessing all its coefficients separately.

Recall that following the basic idea of Generalized PolyDot codes, by fixing $j'=j$, we can evaluate the coefficient of $u^iv^{n-1}w^{k} $ for $i=0,1,\dots, m-1$ and $k=0,1,\dots,d_1-1$, which turns out to be $ \sum_{j=0}^{n-1} \bm{W}_{i,j}\bm{X}_{j,k} =\bm{S}_{i,k}$. Thus, these $md_1$ coefficients constitute the $m\times d_1$ sub-matrices (or blocks) of the desired result $\bm{S}=\bm{W}\bm{X}$. Therefore, $\bm{S}$ can be recovered at any node if it can reconstruct all the coefficients, or rather these $md_1$ coefficients, of the polynomial $\widetilde{\bm{S}}(u,v,w)$.

We use substitutions $(u =v^n, w=v^{mn})$ or $(v=u^m,w=u^{mn})$ as elaborated in \Cref{subsec:comparison_DNN_MM} to convert $\widetilde{\bm{S}}(u,v,w)$ into a polynomial of a single variable, and then use standard decoding techniques \cite{candes2005decoding} to interpolate the coefficients of a polynomial in a single variable from its evaluations at $P$ arbitrary points where some evaluations may be erroneous. Once $\bm{S}$ is decoded, a nonlinear activation function is applied element-wise to generate the input for the next layer. This also makes $\bm{X}$ available at each node at the start of the next feedforward layer, justifying our assumption. 

The additional steps under Error Model $2$ are similar to the case of $B=1$.


\textbf{Backpropagation stage:}
The backpropagated error (transpose) $\bm{\Delta}^T$ is available at every node. Each node partitions $\bm{\Delta}^T$ into a grid of $ d_2 \times m$ blocks, and encodes them using the polynomial:
$$\widetilde{\bm{\Delta}}^T(w,u)= \sum_{k=0}^{d_2-1} \sum_{i=0}^{m-1} \bm{\Delta}^T_{k,i} w^k u^{m-1-i}.$$

For $p=0,1,\dots P-1$, the $p$-th node \textbf{evaluates} $\widetilde{\bm{\Delta}}^T(w,u)$ at $(w,u)=(c_p,a_p)$, yielding $\widetilde{\bm{\Delta}}^T_p=\widetilde{\bm{\Delta}}^T(c_p,a_p)$.

Next, it \textbf{performs} the computation $\widetilde{\bm{C}}^T_p:= \widetilde{\bm{\Delta}}^T_p \widetilde{\bm{W}}^T_p $ and sends the product to all the other nodes, of which, some products may be erroneous. Consider the polynomial:
$$ \widetilde{\bm{C}}^T(w,u,v)=  \widetilde{\bm{\Delta}}^T(w,u) \widetilde{\bm{W}}(u,v)  = \sum_{k=0}^{d_2-1}\sum_{i'=0}^{m-1}\sum_{i=0}^{m-1}\sum_{j=0}^{n-1} \bm{\Delta}^T_{k,i'}\bm{W}_{i,j} w^k u^{m-1+i-i'} v^{j}. $$

The products computed at each node effectively results in the evaluations of this polynomial $\widetilde{\bm{C}}(w,u,v)$ at $(w,u,v)=(c_p,a_p,b_p)$. The coefficients of $w^ku^{m-1}v^{j}$ for $k=0,1,\dots d_2-1$ and $j=0,1,\dots, n-1 $ in this polynomial actually correspond to the $d_2 \times n$ grid of sub-matrices of $\bm{C}^T$. Thus, if every node is able to decode these coefficients from the different evaluations of the polynomial at $P$ nodes, then every node can reconstruct $\bm{C}^T$. Observe that, since both $\bm{C}^T$ and $\bm{X}$ are available at the node, the backpropagated error for the consecutive layer can be computed at each node by computing the Hadamard product $\bm{C}^T \circ g(\bm{X}^T)$.


\textbf{Update stage:}
Since both $\bm{X}$ and $\bm{\Delta^T}$ are available at each node, it can now encode the $b$-th column of $\bm{X}$ and the $b$-th row of $\bm{\Delta^T}$ for $b=0,1,\dots,B-1$, in a manner similar to that of mini-batch size $B=1$.

Let
$$ \bm{\Delta^T} = \begin{bmatrix}
- \bm{\delta}^T_{(0)} - \\
-\bm{\delta}^T_{(1)} -\\
\vdots \\
-\bm{\delta}^T_{(B-1)} -
\end{bmatrix} \text{ and }  \bm{X}= \begin{bmatrix} 
| & | &  & | \\
\bm{x}_{(0)} & \bm{x}_{(1)} & \ldots & \bm{x}_{(B-1)} \\
| & | &  & |
\end{bmatrix}. $$

We also let each individual vector $\bm{\delta}_{(b)}
=\begin{bmatrix}  \bm{\delta}_{(b)0} \\
\bm{\delta}_{(b)1} \\
\vdots \\
\bm{\delta}_{(b)(m-1)}
\end{bmatrix} $ be partitioned into $m$ equal parts (similar to $\bm{\delta}$ being partitioned into $m$ equal parts for the case of mini-batch size $B=1$) and $\bm{x}_{(b)}= \begin{bmatrix}  \bm{x}_{(b)0} \\
\bm{x}_{(b)1} \\
\vdots \\
\bm{x}_{(b)(n-1)}
\end{bmatrix} $ be partitioned into $n$ equal parts (similar to $\bm{x}$).

Now, each node encodes them as follows:
$$ \sum_{i=0}^{m-1} \bm{\delta}_{(b)i} u^i \text{ and } \sum_{j=0}^{n-1} \bm{x}_{(b)j} v^j \text{  at  } (u,v)=(a_p,b_p) \text{  respectively.} $$

Then, the coded $\widetilde{\bm{W}}_p$ can be updated as:
\begin{align}
\widetilde{\bm{W}}_p &\leftarrow \widetilde{\bm{W}}_p + \sum_{b=0}^{B-1} \eta ( \sum_{i=0}^{m-1} \bm{\delta}_{(b)i} a_p^i  )( \sum_{j=0}^{n-1} \bm{x}_{(b)j} b_p^j )^T  \\
& = \sum_{i=0}^{m-1}\sum_{j=0}^{n-1} \underbrace{(\bm{W}_{i,j} + \eta  \sum_{b=0}^{B-1}  (\bm{\delta}_{(b)i})   (\bm{x}_{(b)j})^T ) }_{\text{Update of } \bm{W}_{i,j}} \ a_p^i b_p^j.
\end{align}

Thus, the update step preserves the coded nature of the weight matrix $\bm{W}$.

\textbf{Update with regularization:} As in the case of mini-batch size $B=1$, the coded update can be easily extended to coded update with regularization as follows:
\begin{align}
\widetilde{\bm{W}}_p &\leftarrow (1-\eta \lambda)\widetilde{\bm{W}}_p + \sum_{b=0}^{B-1} \eta ( \sum_{i=0}^{m-1} \bm{\delta}_{(b)i} a_p^i  )( \sum_{j=0}^{n-1} \bm{x}_{(b)j} b_p^j )^T  \\
& = \sum_{i=0}^{m-1}\sum_{j=0}^{n-1} \underbrace{((1-\eta \lambda)\bm{W}_{i,j} + \eta  \sum_{b=0}^{B-1}  (\bm{\delta}_{(b)i})   (\bm{x}_{(b)j})^T ) }_{\text{Update of } \bm{W}_{i,j}} \ a_p^i b_p^j.
\end{align}

\subsection{Comparison with existing strategies.}
\label{subsec:comparison_DNN_MM}
{\color{black}{
\begin{thm}[Error tolerances $(t_f,t_b)$]
\label{thm:error_tolerance_MM}
The error tolerances $(t_f,t_b)$ in the feedforward and backpropagation stages for a layer for the three strategies $\mathcal{C}_{\mathrm{GP}}(K,N,P,B)$, $\mathcal{C}_{\mathrm{mds}}(K,N,P,B)$ and $\mathcal{C}_{\mathrm{rep}}(K,N,P,B)$ are given by \Cref{table_example_MM}.
\end{thm}

\begin{table}[!htbp]
\centering
\caption{Error Tolerances $(t_f,t_b)$ under fixed number of nodes $P$ }
\label{table_example_MM}
\centering
\begin{tabular}{|M{2.7cm}|M{4.5cm}|M{6cm}|}
\hline
Strategy & Error Model $1$ ($t_f,t_b$) & Error Model $2$  ($t_f,t_b$)\\
\hline
$\mathcal{C}_{\mathrm{GP}}(K,N,P,B)$ with $u=v^{n} $ and $ r=v^{mn} $ & $\left(\frac{P-d_1mn-n+1}{2},\frac{P-(d_2+1)mn+n}{2}\right)$  & $(P-d_1mn-n,P-(d_2+1)mn+n-1)$
\\
\hline
$\mathcal{C}_{\mathrm{GP}}(K,N,P,B)$ with $v=u^{m}$ and $r=u^{mn}$ &   $\left(\frac{P-(d_1+1)mn+m}{2},\frac{P-d_2mn-m+1}{2}\right)$ & $(P-(d_1+1)mn+m-1,P-d_2mn-m) $ \\
\hline
$\mathcal{C}_{\mathrm{mds}}(K,N,P,B)$ where $P=P_f+P_b-mn$ & $\left( \frac{P_{f}-mn}{2n},\frac{P_{b}-mn}{2m}\right) $   & $\left(\frac{P_{f}-mn-n}{n},\frac{P_{b}-mn-m}{m} \right) $ \\
\hline
$\mathcal{C}_{\mathrm{rep}}(K,N,P,B)$ & $\left( \frac{P-mn}{2mn},\frac{P-mn}{2mn} \right)$  & $\left(\frac{P-2mn}{mn},\frac{P-2mn}{mn}\right)$  \\
\hline
\end{tabular}
\end{table}
Here $d_1$ and $d_2$ are two integers that divide $B$. Choosing higher values of $d_1$ and $d_2$ reduce the computational and communication complexity as discussed in \Cref{table_complexity_MM}, but also reduce the error tolerance.
}}

The proof is provided in \Cref{appendix:error_tolerance}.

\begin{thm}[Complexity Analysis]
\label{thm:complexity_MM} 
For $\mathcal{C}_{\mathrm{GP}}(K,N,P,B)$ at any layer in a single iteration, the ratio of the total complexity of all the steps including encoding, decoding, communication, nonlinear activation, Hadamard product etc. to the most complexity intensive steps (steps $O1$, $O2$ and $O3$) tends to $0$ as $K,N,P \to \infty$ if the number of nodes satisfy $P^4=o(N)$.
\end{thm}
 
The proof is provided in \Cref{appendix:complexity}. Now we include a table characterizing the storage, communication and computation costs of our proposed strategy in \Cref{table_complexity_MM}. Note that, the parameters $d_1$ and $d_2$ may be chosen accordingly to vary the communication cost as required.
 
 \begin{table}[!ht]
\centering
\caption{Storage, Communication and Computation Costs for Each Layer (All Stages Combined)}
\label{table_complexity_MM}
\centering
\begin{tabular}{|M{3cm}|M{3cm}|M{4cm}|M{5cm}|}
\hline
 Storage & Communication Complexity & Computational Complexity of steps $O1$, $O2$ and $O3$   & Computational Complexity of all other steps (including Encoding/Decoding) \\
\hline
 $\frac{N^2}{K}+ \Theta(\frac{PNB}{md_1}+\frac{PNB}{nd_2})  $  & $\Theta((\frac{NB}{md_1}+\frac{NB}{nd_2})P)$   & $ \Theta(\frac{N^2B}{Kd_1}+\frac{N^2B}{Kd_2} ) $ &  $\mathcal{O}((\frac{NB}{md_1}+\frac{NB}{nd_2})P^3)$  \\
\hline
\end{tabular}
\end{table}

The derivation of these terms is elaborated in \Cref{appendix:complexity}. Once again, we use pessimistic bounds for the additional overheads in our proposed strategy, and in spite of that, we are able to show that these additional overheads are negligible as compared to the complexities of the steps $O1$, $O2$ and $O3$. We are now examining strategies to reduce the overheads further.

\section{Coded Autoencoder}
\label{sec:coded_autoencoder}
In this section, we show that our coded DNN technique can be easily extended to other commonly used architectures through the example of sparse autoencoders. 
Sparse autoencoders (see \cite{ranzato07sparse,ranzato08sparse2} or \Cref{subsec:app_autoenc}) are a specific type of DNN for learning sparse representations of given data in an unsupervised fashion. Sparse autoencoders usually have only one hidden layer, and training autoencoders with more than one hidden layer are treated as multiple one-hidden-layer autoencoders stacked together. Hence, we will only consider training an autoencoder with one hidden layer. 

We first give a short description of sparse autoencoders here. For a comprehensive overview, see \Cref{subsec:app_autoenc}. Note that, we use the index $k$ to denote the iteration index here and a superscript denotes the index of the layer. E.g.~$\bm{W}^l(k)$ denotes the weight matrix $\bm{W}$ of layer $l$ at iteration $k$. 

The major variation in numerical steps arises because of its different loss function $E(\bm{W}^1(k),\bm{W}^2(k))$. In mini-batch SGD, for calculating the gradient $\frac{\partial E(\bm{W}^1(k),\bm{W}^2(k)) }{\partial W^l_{i,j}(k)}$, the loss function over one mini-batch of data is as follows:
\begin{equation}
E(\bm{W}^1(k),\bm{W}^2(k)) = \frac{1}{B}\sum_{b=0}^{B-1} \epsilon^2(b,\bm{W}^1(k),\bm{W}^2(k)) + \lambda (|| \bm{W}^{1}(k) ||_F^2 + ||\bm{W}^{2}(k)||_F^2) + \beta \sum_{i=0}^{N_1-1} \textnormal{KL}(\rho || \hat{\rho}_i),
\end{equation}
where $\hat{\rho}_i$ is the sample sparsity of the second layer's activation averaged over the mini-batch, \emph{i.e.,} $\hat{\rho}_i = \frac{1}{B} \sum_{b=0}^{B-1} Y^1_{i,b} (k)$ where $\bm{Y}^1(k)= \bm{X}^2(k) =f(\bm{S}^1(k)) =f(\bm{W}^1(k)\bm{X}^1(k)) $ applied element-wise and $\bm{X}^1(k) \in \mathcal{R}^{N_0 \times B}$ is a matrix whose columns represent the $B$ data points chosen at the $k$-th iteration.

As the difference is only in the loss function, the feedforward stage follow the same procedures given in Section~\ref{sec:extension}. Now let us examine  coded backpropagation and update stages. First, notice that updating $\bm{W}^2$ at the final layer is simply an update with L$2$ regularization, explained in \Cref{subsec:app_autoenc}. Then the major difference arises in updating $\bm{W}^1$ due to the last term in the loss function. Updating $\bm{W}^1$ follows:
\begin{equation}
\label{eq:autoencoder_update}
\bm{W}^1(k+1) = (1-\eta \lambda) \bm{W}^1(k) + \eta \bm{\Delta}_{\textnormal{auto}}^1(k) [\bm{X}^1(k)]^T,
\end{equation}
where 
\begin{equation}
\label{eq:autoencoder_delta}
\bm{\Delta}_{\textnormal{auto}}^1(k) = \left((\bm{W}^2(k))^T \bm{\Delta}^{2}(k) - \bm{Q}_{\bm{\hat{\rho}}}(k) \right) \circ f'(\bm{S}^1(k)) = \left(\bm{C}^2(k) - \bm{Q}_{\bm{\hat{\rho}}}(k) \right) \circ f'(\bm{S}^1(k)).
\end{equation}
If we compared \Cref{eq:autoencoder_delta} with the update rule of a generic DNN, where
\begin{equation*}
\bm{\Delta}^1(k) = \left((\bm{W}^2(k))^T \bm{\Delta}^{2}(k) \right) \circ f'(\bm{S}^1(k)) = \bm{C}^2(k) \circ f'(\bm{S}^1(k)),
\end{equation*}
we observe that the only additional term is $\bm{Q}_{\bm{\hat{\rho}}}(k)$. The derivation of these autoencoder equations are given in Appendix~\ref{subsec:app_autoenc}. The only additional term in $\bm{\Delta}_{\textnormal{auto}}^1$ compared to $\bm{\Delta}^1$ is $\bm{Q}_{\bm{\hat{\rho}}}$. Hence, we only have to analyze how this term can be incorporated into our coded DNN framework. $\bm{Q}_{\bm{\hat{\rho}}}$ is a matrix with the following form:
\begin{equation}
\bm{Q}_{\bm{\hat{\rho}}} = \left[ \vphantom{\begin{array}{c}1\\1\\1\\1\end{array}}
                  \smash{\underbrace{
                      \begin{array}{cccc}
                             Q(\hat{\rho}_0) & Q(\hat{\rho}_0) & \cdots & Q(\hat{\rho}_0)\\
\vdots & \vdots & \ddots & \vdots \\
Q(\hat{\rho}_{N_1-1}) & Q(\hat{\rho}_{N_1-1}) & \cdots & Q(\hat{\rho}_{N_1-1})\\
                      \end{array}
                      }_{B}}
              \right], \; \textnormal{ for } \; Q(\hat{\rho}_i) = -\frac{\rho}{\hat{\rho}_i} + \frac{1-\rho}{1-\hat{\rho}_i}.
\end{equation}
\\
To compute this, we first need to obtain $\hat{\rho}_0, \cdots \hat{\rho}_{N_1-1}$. As each node already generates the entire input for the next layer during the feedforward stage, \textit{i.e.}, $\bm{X}^2 = \bm{Y}^1 = f(\bm{W}^1(k)\bm{X}^1(k))$, the $\hat{\rho}_i$'s can be computed at every node with computation complexity 
$O(N_1 B)$. Then, computing $Q(\hat{\rho}_i)$'s takes computation complexity of $O(N_1)$. After we complete coded multiplication, $(\widetilde{\bm{C}}^2)^T_p= (\widetilde{\bm{\Delta}}^2)^T_p \widetilde{\bm{W}}^2_p $ and decode $(\bm{C}^2)^T$ at each node, the nodes can compute $\bm{\Delta}_{\textnormal{auto}}^1$ with complexity $O(N_1 B)$. Then it can also encode the computed $\bm{\Delta}_{\textnormal{auto}}^1$ for coded update of $\bm{W}^1$. 

To summarize, the key steps are as follows: 
\begin{itemize}
\item Feedforward stage: Exactly same as before
\item Backpropagation and Update stage at layer $2$: The $\bm{W}^2$ update is only backpropagation with regularization.
\item Backpropagation stage at layer $1$: The $\bm{W}^1$ update is following equations (\ref{eq:autoencoder_update}) and (\ref{eq:autoencoder_delta}). 
\begin{enumerate}
\item Computing $\bm{\hat{\rho}}$: complexity is $\Theta(N_1 B)$. 
\item Computing $\bm{Q}_{\bm{\hat{\rho}}}$: complexity is $\Theta(N_1)$. 
\item Computing $\bm{\Delta}_{\textnormal{auto}}^1$: Computing coded matrix-matrix product $(\widetilde{\bm{C}}^2)^T_p= (\widetilde{\bm{\Delta}}^2)^T_p \widetilde{\bm{W}}^2_p $ which takes $\Theta \left(\frac{N_1 N_2 B}{d_2K} \right)$ complexity.
\item Decode $(\bm{C}^2)^T$ at every node: complexity $\Theta(P^3\frac{N_1B}{d_2n} ) $.
\end{enumerate}
\item Update stage at layer $1$: Encode $\bm{\Delta}_{\textnormal{auto}}^1$ for the update.
\end{itemize}

\section{Discussion and Conclusions}

To summarize, in this work we first proposed a novel coded computing technique for distributed matrix multiplications called Generalized PolyDot and then used it to design a unified coded computing strategy for DNN training. Our proposed coding strategy (and the concurrent strategy of \cite{yu2018entangled}) advances on the existing coded computing techniques for distributed matrix-matrix multiplication, improving the recovery threshold and error tolerances. Lastly, we also show how our unified strategy can be adapted to specific applications, such as, autoencoders.

The problem of reliable computing using unreliable elements was first posed in 1956 by von Neumann\cite{von1956probabilistic}, speculating that the efficiency and reliability of the human brain is obtained by allowing for low power but error-prone components with redundancy for error-resilience. This is also evident from the influence of McCulloch-Pitts model of a neuron~\cite{mcculloch1943logical} in his work. It is often speculated~\cite{sreenivasan2011error} that the error-prone nature of brain's hardware actually helps it be more efficient: rather than making individual components more reliable using higher power/resources, it might be more efficient to accept component-level errors, and utilize sophisticated error-correction mechanisms for overall reliability of the computation\footnote{See also the ``Efficient Coding Hypothesis'' of Barlow~\cite{barlow1961possible} and an application in~\cite{yang2017encoded}.}. It is thus surprising that this problem of training neural networks under errors has still remained open, even as massive artificial neural networks are being trained on increasingly low-cost and unreliable processing units.  We believe that this work (and our prior work~\cite{dutta2018DNN1}) might be a significant step in the design of biologically inspired neural networks with error resilience that could hold the key to significant improvements in efficiency and reduction of energy consumption during neural network training. Thus, these results could be of broader scientific interest to communities like High Performance Computing (HPC), neuroscience as well as neuromorphic computing.

\ifCLASSOPTIONcaptionsoff
  \newpage
\fi



\bibliographystyle{IEEEtran}
\bibliography{IEEEabrv,sample}
\appendices 
\crefalias{section}{app}
\crefalias{subsection}{app}
\section{DNN Background}
\label{appendix:DNN_background}
We provide a background on DNNs for unfamiliar readers. We also follow the standard description used in DNN literature~\cite{rumelhart1988learning}, so familiar readers can merely skim this part. \\

\noindent \textbf{DNN Operations:} \\
We first explain the training of a DNN with $l=1,2,\ldots,L$ layers (excluding the input layer which can be thought of as the layer $0$) using Stochastic Gradient Descent (SGD) with batch size $B=1$. Later, we will explain how this can be extended to mini-batch SGD with mini-batch size $B > 1$.
A Deep Neural Network (DNN) essentially consists of $L$ weight matrices (also called parameter matrices), one for each layer, that represent the connections between the $l$-th and $(l-1)$-th layer for $l=1,2,\ldots,L$. At the $l$-th layer, $N_{l}$ denotes the number of neurons. Thus, for layer $l$, the weight matrix to be trained is of dimension $N_{l} \times N_{l-1}$. At the $l$-th layer ($l=1,\ldots,L$), $N_{l}$ denotes the number the neurons, (\emph{i.e.,} the row-dimension of the weight matrix).

We use the index $k$ to denote the iteration number of the training. At the $k$-th iteration, the neural network is trained based on a single data point using three stages: a \textit{feedforward} stage, followed by a \textit{backpropagation} stage and an \textit{update} stage.  We use the following notations: 

\begin{enumerate}

\item $N_1,\dots,N_L :$ Number of neurons in layers $1,2,\dots,L$. We also introduce the notation $N_0$ to denote the dimension of the original data vector, which serves as the input to the first layer.

\item $W_{i,j}^l (k):$ At iteration $k$, the weight of the connection from neuron $j$ on layer $l-1$ to neuron $i$ on layer $l$ for $i=0,1,\dots,N_l-1$ and $ j=0,1,\dots,N_{l-1}-1$. Note that, the weights actually form a matrix $\bm{W}^l(k)$ of dimension $N_l \times N_{l-1}$ for layer $l$.

\item $\bm{x}^l(k) \in \mathcal{R}^{N_{l-1}}  : $ The input of layer $l$ at the $k$-th iteration. Note that, for the first layer, $\bm{x}^1(k)$ becomes the data point used for the $k$-th iteration of training.

\item $ \bm{s}^l(k): $ The summed output of the neurons of layer $l$ before a nonlinear function $f(\cdot)$ is applied on it, at the $k$-th iteration. Note that, for $i=0,1,\dots,N_l-1$, the scalar $s_i^l(k)$ is the $i-$th entry of the vector $\bm{s}^l(k)$, \textit{i.e.} the summed output of neuron $i$ on layer $l$. 

\item $\hat{\bm{y}}^l(k) \in \mathcal{R}^{N_{l}}  : $ The output of layer $l$ at the $k$-th iteration after the application of the nonlinear function. Note that, the output of the last layer $L$, \textit{i.e.}, $\hat{\bm{y}}^L(k)$ is the final estimated label generated by the neural network that can be compared with the true label.

\end{enumerate}

\noindent We now make some observations that explains the functional connectivity across the layers:
\begin{enumerate}
\item Input for any layer is the output of the previous layer except of course for the first layer whose input is the actual data vector itself:

$$\bm{x}^l(k)=
\begin{cases}
    \hat{\bm{y}}^{l-1}(k) \in \mathcal{R}^{N_{l-1}} ,& \text{if } l=2,3,\dots,L\\
    \bm{x}^1(k),              & \text{otherwise}.
\end{cases} $$

\item At each layer, the input for that layer is summed with appropriate weights of that layer (entries $W_{i,j}^l$) to produce the summed output of each neuron given by:
\begin{align}
s_i^l(k) & =\bm{W}^l_{i,:}(k)\bm{x}^{l}(k) = \sum_{j=0}^{N_{l-1}-1} W_{i,j}^{l}(k) x^{l}_j(k) \\
& =  \sum_{j=0}^{N_{l-1}-1} W_{i,j}^{l}(k) x^{l}_j(k)  = \sum_{j=0}^{N_{l-1}-1} W_{i,j}^{l}(k) \hat{y}^{l-1}_j(k).  
\end{align}

\item The final output of each layer is given by a nonlinear function applied on the summed output of each neuron as below:
\begin{align}
\hat{y}_i^l (k) = f(s_i^l(k)) = f\left(\sum_{j=1}^{N_{l-1}} W_{i,j}^{l}(k) \hat{y}^{l-1}_j(k)   \right).
\end{align}

\item Observe that $\hat{y}^L(k)$ of the last layer denotes the estimated output or label of the DNN and is to be compared with the true label vector $\bm{y}(k)$ for the corresponding data point $\bm{x}^1(k)$.\\

\end{enumerate}

    \noindent \textbf{Key Idea of Training:} \\
Let us assume we have a large data set $\chi$ consisting of several data points and their labels. The goal of training a DNN is to find the weight matrices $\bm{W}^1$ to $\bm{W}^L$ that minimize the empirical loss function defined as follows: 
\begin{equation}
\label{eq:loss1}
E (\bm{W}^1,\bm{W}^2, \ldots, \bm{W}^L) = \frac{1}{|\chi|} \sum_{b=0}^{|\chi|-1} \epsilon^2(b,\bm{W}^1,\bm{W}^2, \ldots, \bm{W}^L ) .
\end{equation}
Here $b$ denotes the index of the data point in $\chi$ and $\epsilon^2(b,\bm{W}^1,\bm{W}^2, \ldots, \bm{W}^L )$ is the loss for the $b$-th data point and its label when using the weight matrices (also called parameter matrices) $\bm{W}^1,\bm{W}^2, \ldots, \bm{W}^L$ in the neural network. We clarify this using the following example. Let $(\bm{x}^1,\bm{y})$ be a particular data point and its true label vector and suppose the network consists of only a single layer followed by an element-wise nonlinear activation function $f(\cdot)$. Then, the estimated label $\hat{\bm{y}}^L$ (here $L=1$) is given by $\hat{\bm{y}}^L= f(\bm{W}^1\bm{x}^1)$ applied element-wise. Therefore, the empirical loss function is as follows:
\begin{align*}
E (\bm{W}^1) & = \frac{1}{|\chi|} \sum_{b=0}^{|\chi|-1} \epsilon^2(b,\bm{W}^1) \\
&\overset{\text{E.g.~L2 loss}}{=} \frac{1}{|\chi|} \sum_{b=0}^{|\chi|-1} \left( ||\text{Estimated Label} - \text{True Label} ||_2^2 \right)_{\text{for the b-th data point}} \\
&= 
\frac{1}{|\chi|} \sum_{b=0}^{|\chi|-1} \left(||\hat{\bm{y}}^L-\bm{y}||_2^2 \right)_{\text{for the b-th data point}} = \frac{1}{|\chi|}\sum_{(\bm{x}^1,\bm{y})\in \chi } ||f(\bm{W}^1\bm{x}^1)-\bm{y}||_2^2.
\end{align*}
Other commonly used loss functions for $\epsilon^2(\cdot)$ are hinge loss, logistic loss, or cross-entropy loss. The technique does not depend on the specific choice of the loss function.

Gradient Descent (GD) is one way to iteratively minimize this loss function using the following update rule:
\begin{align}
W_{i,j}^l(k+1) &=  W_{i,j}^l(k)- \eta \frac{\partial E(\bm{W}^1(k),\ldots, \bm{W}^L(k)) }{\partial W_{i,j}^l(k) } \\
& =  W_{i,j}^l(k)- \eta \frac{1}{|\chi|} \sum_{b=0}^{|\chi|-1} \frac{\partial \epsilon^2(b,\bm{W}^1(k), \ldots, \bm{W}^L(k))}{\partial W_{i,j}^l(k)}.
\end{align}
Here $k$ denotes the index of the iteration, $l$ denotes the layer index of the neural network, $\bm{W}^l(k)$ is the value of the weight matrix (or parameter matrix) of layer $l$ at the $k$-th iteration, and 
$W_{i,j}^l(k) $ denotes the $(i,j)$-th scalar element of the matrix $\bm{W}^l(k)$. However, as is evident from this update rule, that GD would require access to all the data points to compute the gradient at each iteration which is computationally expensive. \\

\noindent \textbf{Stochastic Gradient Descent:}\\
To remedy this, one often uses Stochastic Gradient Descent (SGD) instead of full Gradient Descent (GD) where the gradient with respect to only a single data point is used at each iteration and the update rule is replaced as:

\begin{align}
W_{i,j}^l(k+1) =  W_{i,j}^l(k)- \eta  \frac{\partial \epsilon^2(b,\bm{W}^1(k), \ldots, \bm{W}^L(k))}{\partial W_{i,j}^l(k)}.
\end{align}

Here $b$ denotes the index of the data point which is accessed at the $k$-th iteration, and this depends on the iteration index $k$. For the ease of explanation of SGD, we can simply rewrite $\epsilon^2(b,\bm{W}^1(k), \ldots, \bm{W}^L(k))$ only as a function of the iteration index $k$ alone, \textit{i.e.}, $\epsilon^2(k)= \epsilon^2(b,\bm{W}^1(k), \ldots, \bm{W}^L(k))$. Thus, the update rule is as follows:
\begin{align}
W_{i,j}^l(k+1) =  W_{i,j}^l(k)- \eta  \frac{\partial \epsilon^2(k)}{\partial W_{i,j}^l(k)}.
\end{align}

We will also include the case of mini-batch SGD with batch-size $B$ later in \Cref{subsec:mini_batch_SGD}.\\

\noindent  \textbf{Derivation of Backpropagation and Update:} \\
\noindent Recall that, at the $k$-th iteration, the weights of every layer $l$ of the DNN are to be updated as follows:
\begin{equation}
W_{i,j}^l(k+1)=W_{i,j}^l(k) - \eta \frac{\partial \epsilon^2(k)}{\partial W_{i,j}^l(k)}.
\end{equation}

The backpropagation (and update) helps us to compute the errors and updates in a recursive form, so that the update of any layer $l$ depends only on the backpropagated error vector of its succeeding layer, \textit{i.e.} layer $l+1$ and not on the layers $\{l+1, l+2, \dots ,L\}$. Here, we provide the backpropagation and update rules. Let us define the backpropagated error vector as $\bm{\delta}^l(k)$:
\begin{equation}
\delta_i^l(k)= -\frac{\partial \epsilon^2(k)}{\partial s_i^l(k) } \ \ \forall \ i=0,1,\dots,N_l-1.
\end{equation}
For example, one might again consider the special case of the L2 loss function given by:
\begin{equation}
\epsilon^2(k) = ||\hat{\bm{y}}^L(k)- \bm{y}(k) ||_2^2 =\sum_{i=0}^{N_L-1} (\hat{y}_i^L(k)-y_i(k))^2. 
\end{equation}
For the L2 loss function, the error at the last layer is given by:
\begin{equation}
\label{error:last_layer}
\delta_i^L(k)= 2 \epsilon_i(k) f'(s_i^L(k)).
\end{equation}

\noindent Now, observe that $\delta_i^l(k)$ can be calculated from $\bm{\delta}^{l+1}(k)$ as we derive here in Lemma \ref{lem:DNN}.

\begin{lem}
\label{lem:DNN}
During the training of a neural network using backpropagation, the backpropagated error vector $\delta_i^{l}(k)$ for any layer $l$ can be expressed as a function of the backpropagated error vector of the previous layer as given by:
\begin{equation}
\delta_i^{l}(k) = \left( \sum_{j=0}^{N_{l+1}-1} \delta_j^{l+1}(k) W_{j,i}^{l+1}(k)  \right) f'(s_i^l(k)).
\end{equation}
\end{lem}
%
%
\begin{proof}[Proof of Lemma \ref{lem:DNN}:]
\begin{align}
\delta_i^l(k) & = -\frac{\partial \epsilon^2(k)}{\partial s_i^l(k) } = -\sum_{j=0}^{N_{l+1}-1} \frac{\partial \epsilon^2(k)}{\partial s_j^{l+1}(k) } \frac{\partial s_j^{l+1}(k)}{\partial s_i^l(k) } \\
& = -\sum_{j=0}^{N_{l+1}-1} \frac{\partial \epsilon^2(k)}{\partial s_j^{l+1}(k) } W^{l+1}_{j,i}(k) f'(s^l_i(k)) = \left( \sum_{j=0}^{N_{l+1}-1} \delta_j^{l+1}(k) W_{j,i}^{l+1}(k)  \right) f'(s_i^l(k)).
\end{align}
\end{proof}

Now, using the fact that $s_i^l(k)  = \sum_{j=0}^{N_{l-1}-1} W_{i,j}^{l}(k) x^{l}_j(k)$, we have
\begin{equation}
\frac{\partial \epsilon^2(k)}{\partial W_{i,j}^l(k) } = \frac{\partial \epsilon^2(k)}{\partial s_i^l(k) }\frac{\partial s_i^l(k)}{\partial W_{i,j}^l(k) } = -\delta_i^l(k) x^{l}_j(k).
\end{equation}
Thus, the update rule is derived as follows:
\begin{equation}
\label{eq:update_Wij}
W_{i,j}^l(k+1)=W_{i,j}^l(k) - \eta \frac{\partial \epsilon^2(k)}{\partial W_{i,j}^l(k) }= W_{i,j}^l(k) + \eta \delta_i^l(k) x^{l}_j(k).
\end{equation}
Note that, the update rule does not depend on any particular choice of loss function.

\subsection{Algorithmic Steps for DNN Training ($B=1$).}
We assume that a DNN with $L$ layers (excluding the input layer) is being trained using backpropagation with Stochastic Gradient Descent (SGD)\footnote{As a first step in this direction of coded neural networks, we assume that the training is performed using vanilla SGD. As a future work, we plan to extend these coding ideas to other training algorithms~\cite{ruder2016overview} such as momentum SGD, Adam etc.} with a mini-batch size of $B=1$~\cite{rumelhart1988learning}. The DNN thus consists of $L$ weight matrices (see \Cref{fig:DNN_training}), one for each layer, that represent the connections between the $l$-th and $(l-1)$-th layer for $l=1,2,\ldots,L$. At the $l$-th layer, $N_{l}$ denotes the number of neurons. Thus, the weight matrix to be trained is of dimension $N_{l} \times N_{l-1}$. For simplicity of presentation, we assume that $N_l=N$ for all layers. In every iteration, the DNN (\textit{i.e.} the $L$ weight matrices) is trained based on a single data point and its true label through three stages, namely, feedforward, backpropagation and update, as shown in Fig.~\ref{fig:DNN_training}. At the beginning of every iteration, the first layer accesses the data vector (input for layer $1$) and starts the feedforward stage which propagates from layer $l=1$ to $l=L$. For a layer $l$, let us denote the weight matrix, input for the layer and backpropagated error for iteration $k$ by $\bm{W}^l(k)$, $\bm{x}^l(k)$ and $\bm{\delta}^l(k)$ respectively. The operations performed in layer $l$ during feedforward stage (see \Cref{fig:DNN_step1}) can be summarized as:
\begin{itemize}[noitemsep,topsep=0pt,leftmargin=0.2cm]
\item $[$Step $O1]$ Compute matrix-vector product $\bm{s}^l(k)=\bm{W}^l(k)\bm{x}^l(k)$. 
\item $[$Step $C1]$ Compute input for layer $(l+1)$ given by $\bm{x}^{(l+1)}(k)=f(\bm{s}^l(k))$ where $f(\cdot)$ is a nonlinear activation function applied elementwise.
\end{itemize}
At the last layer ($l=L$), the backpropagated error vector is generated by assessing the true label and the estimated label, $f(\bm{s}^L(k))$, which is output of last layer (see \Cref{fig:DNN_step2}). Then, the backpropagated error propagates from layer $L$ to $1$ (see \Cref{fig:DNN_step3}), also updating the weight matrices at every layer alongside (see \Cref{fig:DNN_step4}). The operations for the backpropagation stage can be summarized as:
\begin{itemize}[noitemsep,topsep=0pt,leftmargin=0.2cm]
\item $[$Step $O2]$ Compute matrix-vector product $[\bm{c}^l(k)]^T=[\bm{\delta}^l(k)]^T\bm{W}^l(k)$. 
\item $[$Step $C2]$ Compute backpropagated error vector for layer $(l-1)$ given by $[\bm{\delta}^{(l-1)}(k)]^T=[\bm{c}^l(k)]^T\bm{D}^l(k)$ where $\bm{D}^l(k)$ is a diagonal matrix whose $i$-th diagonal element depends only on the $i$-th value of $\bm{x}^l(k)$. More specifically, $\bm{D}^l(k)$ is a diagonal matrix whose $i$-th diagonal element is a function $g(\cdot)$ of the $i$-th element of $\bm{x}^l(k)$, such that, $g(f(u))=f'(u)$ for the chosen nonlinear activation function $f(\cdot)$ in the feedforward stage. This is equivalent to computing the Hadamard product: $[\bm{\delta}^{(l-1)}(k)]^T=[\bm{c}^l(k)]^T \circ g([\bm{x}^l(k)]^T ) $.

\end{itemize}
Finally, the step in the Update stage is as follows:
\begin{itemize}[noitemsep,topsep=0pt,leftmargin=0.2cm]
\item  $[$Step $O3]$ Update as: $\bm{W}^l(k+1) \leftarrow \bm{W}^l(k) + \eta \bm{\delta}^l(k)[\bm{x}^l(k)]^T$ where $\eta$ is the learning rate. Sometimes a regularization term is added with the loss function in DNN training (elaborated in \Cref{subsec:app_reg}). For L2 regularization, the update rule is modified as: $\bm{W}^l(k+1) \leftarrow (1-\eta \lambda)\bm{W}^l(k) + \eta \bm{\delta}^l(k)[\bm{x}^l(k)]^T$ where $\eta$ is the learning rate and $\frac{\lambda}{2}$ is the regularization constant.
\end{itemize}

\textbf{Important Note:} As the primary, computationally intensive operations in the DNN training remains roughly the same across layers, we simplify our notations in the main part of the paper. For any layer, we denote the feedforward input $\bm{x}^l(k)$, the weight matrix $\bm{W}^l(k)$ and the backpropagated error $\bm{\delta}^l(k)$ as $\bm{x}$, $\bm{W}$ and $\bm{\delta}$ respectively. We also assume $N_0=N_1=\dots = N_{L}=N$, and thus the dimension of $\bm{W}$ is $N \times N$.

\subsection{Extension to mini-batch SGD with batch size $B>1$.}
\label{subsec:mini_batch_SGD}
Note that while SGD is faster in terms of computations as compared to GD which requires access to the whole data set at each iteration, it is usually less accurate and has noise. So, often one uses mini-batch SGD with a batch size of $B>1$ which is somewhat intermediate between SGD (batch size $B=1$) and full GD. In mini-batch SGD, the gradient is computed over a mini-batch of $B$ data points together than only a single data point, though $B$ is usually much much smaller than the total size of the data set $|\chi|$. Thus, the update rule is given by:

\begin{align}
W_{i,j}^l(k+1) &=  W_{i,j}^l(k)- \frac{\eta}{B} \sum_{b=0}^{B-1}  \frac{\partial \epsilon^2(b,\bm{W}^1(k), \ldots, \bm{W}^L(k))}{\partial W_{i,j}^l(k)} \\
&= W_{i,j}^l(k)- \eta \frac{\partial \left( \frac{1}{B} \sum_{b=0}^{B-1} \epsilon^2(b,\bm{W}^1(k), \ldots, \bm{W}^L(k)) \right)}{\partial W_{i,j}^l(k)}. 
\end{align}

Here, $B$ denotes the size of the mini-batch, and $b$ denotes the index of the data point in the subset of $B$ data points that are chosen for the particular iteration $k$. Again, it is more convenient to represent the whole term $\frac{1}{B} \sum_{b=0}^{B-1}  \epsilon^2(b,\bm{W}^1(k), \ldots, \bm{W}^L(k))$ as only a function of the iteration index $k$. Thus, we have,  $\epsilon^2(k)= \frac{1}{B} \sum_{b=0}^{B-1}   \epsilon^2(b,\bm{W}^1(k), \ldots, \bm{W}^L(k))$. \\

\noindent \textbf{DNN Operations:}

As we are operating on $B$ data points instead of a single one during the training using mini-batch SGD, we introduce some matrix notations here as an extension to the vector notations for the case of mini-batch $B=1$.
\begin{enumerate}
\item $N_0, N_1, \dots,N_L :$ (Stays same as before)

\item $W_{i,j}^l (k):$ (Stays same as before)

\item $\bm{X}^l(k) \in \mathcal{R}^{N_{l-1} \times B}  : $ The input of layer $l$ at the $k$-th iteration. Note that, this is now a matrix with $B$ columns instead of a single column vector. Note that, for the first layer, $\bm{X}^1(k)$ consists of the set of $B$ data points used for the $k$-th iteration of training, all arranged in consecutive columns.

\item $ \bm{S}^l(k) \in \mathcal{R}^{N_{l}\times B }  : $ The summed output of the neurons of layer $l$ before a nonlinear function $f(\cdot)$ is applied on it, at the $k$-th iteration. Note that, during the feedforward stage, $\bm{S}^l(k)=\bm{W}^l(k)\bm{X}^l(k)$. 

\item $\hat{\bm{Y}}^l(k) \in \mathcal{R}^{N_{l}\times B }  : $ The output of layer $l$ at the $k$-th iteration. Note that, $\hat{\bm{Y}}^l(k) = f(\bm{S}^l(k))$, applied element-wise and the input to the next layer, \textit{i.e.}, $\bm{X}^{l+1}(k)= \hat{\bm{Y}}^l(k) $.

\end{enumerate}

The derivation of the backpropagation and update rule is similar to the case of mini-batch $B=1$. The backpropagated error is also now a matrix $ \in \mathcal{R}^{N_l} \times B $ denoted by $\bm{\Delta}^l(k) $ whose $(i,b)$-th entry is defined as:
\begin{equation}
\Delta_{i,b}^l(k)= - \frac{\partial \epsilon^2(k) }{\partial s_{i,b}^l(k) } \forall \ i=0,1,\ldots,N_l-1, \ b=0,1,\ldots,B-1.
\end{equation}

Using an analysis similar to the previous case, the recursion for backpropagation can be derived as follows:
\begin{align}
\Delta^{l}_{i,b}(k) = \left( \sum_{j=0}^{N_{l+1}-1}\Delta^{l+1}_{j,b}(k) W_{j,i}^{l+1}(k)  \right) f'(S_{i,b}^l(k) ).
\end{align}

Or, in matrix notations, this reduces to the following steps:
\begin{itemize}
\item Matrix Multiplication: $ [\bm{C}^l(k)]^T = [\bm{\Delta}^l(k)]^T \bm{W}^l(k) .$
\item Hadamard product: $[\bm{\Delta}^{l-1}(k)]^T= \bm{e}_{(0)}[\bm{C}^l(k)]^T \bm{D}^l_0(k) + \bm{e}_{(1)}[\bm{C}^l(k)]^T \bm{D}^l_1(k) + \dots + \bm{e}_{(B-1)}[\bm{C}^l(k)]^T \bm{D}^l_{B-1}(k)
$ where $\bm{D}^l_b(k)$ is a diagonal matrix that only depends on the $b$-th column of $\bm{X}^l(k)$, \textit{i.e.},  whose $i$-th diagonal element depends on only the element at location $(i,b)$ of the matrix $\bm{X}^l(k)$. More specifically, $\bm{D}^l_b(k)$ is a diagonal matrix whose $i$-th diagonal element is a function $g(\cdot)$ of the element at location $(i,b)$ of the matrix $\bm{X}^l(k)$, such that $g(f(u))=f'(u)$ for the chosen nonlinear activation function $f(\cdot)$ in the feedforward stage. Also, $\bm{e}_{(b)}$ is a unit vector of dimension $1 \times B $ whose $b$-th entry is $1$ and all others are $0$. To clarify, observe that,
$$ \bm{I}_{B \times B} = \begin{bmatrix}
1 & 0 & \dots & 0 \\
0 & 1 & \dots & 0 \\
\vdots & \vdots & \ddots & \vdots \\
0 & 0 & \dots & 1 
\end{bmatrix} = \begin{bmatrix}
-\bm{e}_{(0)}-\\
-\bm{e}_{(1)}-\\
\vdots \\
-\bm{e}_{(B-1)}-
\end{bmatrix}.$$

Note that the aforementioned step can also be rewritten as a Hadamard product between two matrices as follows:

$$ \bm{\Delta}^{l-1}(k) = \bm{C}^l(k) \circ f'(\bm{S}^{l-1}(k)) =  \bm{C}^l(k) \circ g(\bm{X}^{l}(k)), $$
the Hadamard product ``$\circ$'' between two matrices of the same dimension is defined as a matrix of the same dimension as the operands, with each element given by the element-wise product of the corresponding elements of the operands.
\end{itemize}

In matrix notations, the update rule is:
\begin{equation}
\bm{W}^{l}(k+1)= \bm{W}^{l}(k) + \eta \bm{\Delta}^l(k) [ \bm{X}^l(k)]^T.
\end{equation}

\textbf{Important Note:} Similar to the case of $B=1$, we omit the iteration index $k$ and layer index $l$ in the main text of the paper because we are mainly interested in only the matrix operations that occur at each iteration across each layer and these operations are similar across all iterations and layers.

\subsection{Regularization.} \label{subsec:app_reg}

The loss function discussed in \Cref{eq:loss1} is often used in practice with an additional regularization term. Thus, the empirical loss function takes the form:

\begin{equation}
\label{eq:loss2}
E(\bm{W}^1, \bm{W}^2, \ldots, \bm{W}^L) =  \frac{1}{|\chi|} \sum_{b=0}^{|\chi|-1} \epsilon^2(b,\bm{W}^1,\bm{W}^2, \ldots, \bm{W}^L )  + R(\bm{W}^1, \cdots, \bm{W}^l).
\end{equation}
where $R(\cdot)$ is a regularization function of the matrices $\bm{W}^1, \cdots,\bm{W}^L$. Similar to above, the goal of training DNN is finding the weight matrices $\bm{W}^1$ to $\bm{W}^L$ that minimize the loss function in \Cref{eq:loss2}.



One of the commonly used regularization techniques in deep learning is adding penalty norms on weights, such as L2 norm and L1 norm. For L2 regularization (weight decay), loss function in (\ref{eq:loss2}) can be written as:
\begin{equation}
E(\bm{W}^1,\bm{W}^2,\ldots, \bm{W}^L) = \frac{1}{|\chi|}\sum_{b=0}^{|\chi|-1} \epsilon^2(b,\bm{W}^1,\bm{W}^2,\ldots, \bm{W}^L) + \frac{\lambda}{2} \sum_{l=1}^{L} ||\bm{W}^l||_F^2.
\end{equation}
It is desirable to use different $\lambda$ for the different layers, but for the simplicity we assume that we use the same $\lambda$ across all layers. As we already described, SGD is usually preferred for these kinds of optimization problems as compared to the batch Gradient Descent. Thus, the SGD update rule takes the form:
\begin{equation}
W_{i,j}^l(k+1)= W_{i,j}^l(k) - \eta \left( \frac{\partial \epsilon^2(b,\bm{W}^1,\bm{W}^2,\ldots, \bm{W}^L)}{\partial W_{i,j}^l(k)} +  \frac{\partial \frac{\lambda}{2} \sum_{l=1}^{L} ||\bm{W}^l||_F^2 }{\partial W_{i,j}^l(k)}   \right), 
\end{equation}
where $b$ is the index of the data point used in the $k$-th iteration. With the weight decay regularization term, the update equation in (\ref{eq:update_Wij}) can thus be rewritten as:
\begin{equation} 
W_{i,j}^l(k+1)= W_{i,j}^l(k) + \eta \left(\delta_i^l(k) x^{l}_j(k) -\lambda W_{i,j}^l(k) \right) = (1-\eta\lambda)W_{i,j}^l(k) + \eta \delta_i^l(k) x^{l}_j(k).
\end{equation}
Note that by adding the weight decay term, we have to shrink $W_{i,j}$'s by a constant factor $(1-\eta\lambda)$ at each iteration. A similar expression can also be derived for the case of mini-batch $B>1$.


\subsection{Training Autoencoders.} \label{subsec:app_autoenc}
\begin{figure}[ht]
\centering
\subfloat[]{
\includegraphics[height=5.5cm]{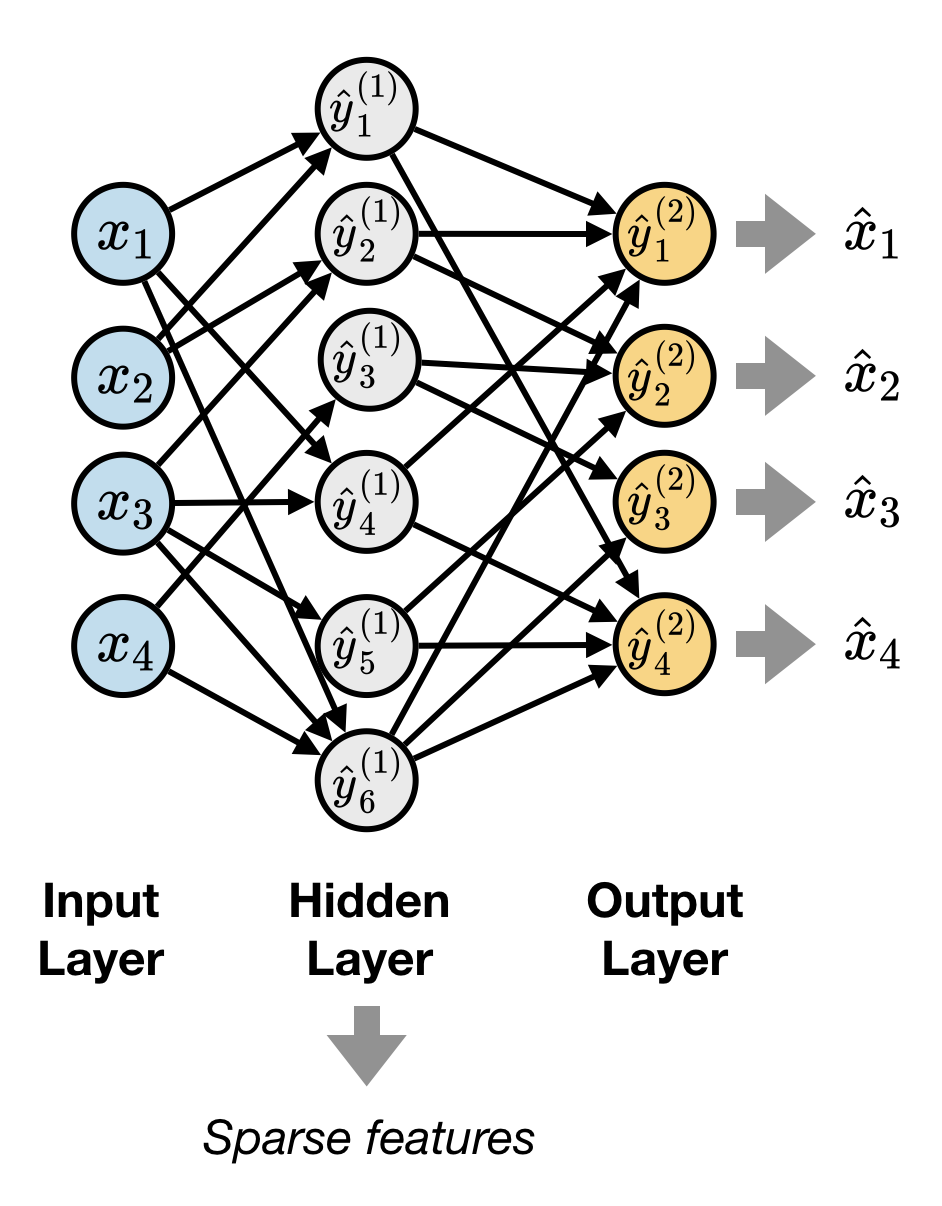} }
\hspace{1cm}
\subfloat[]{ \includegraphics[height=6cm]{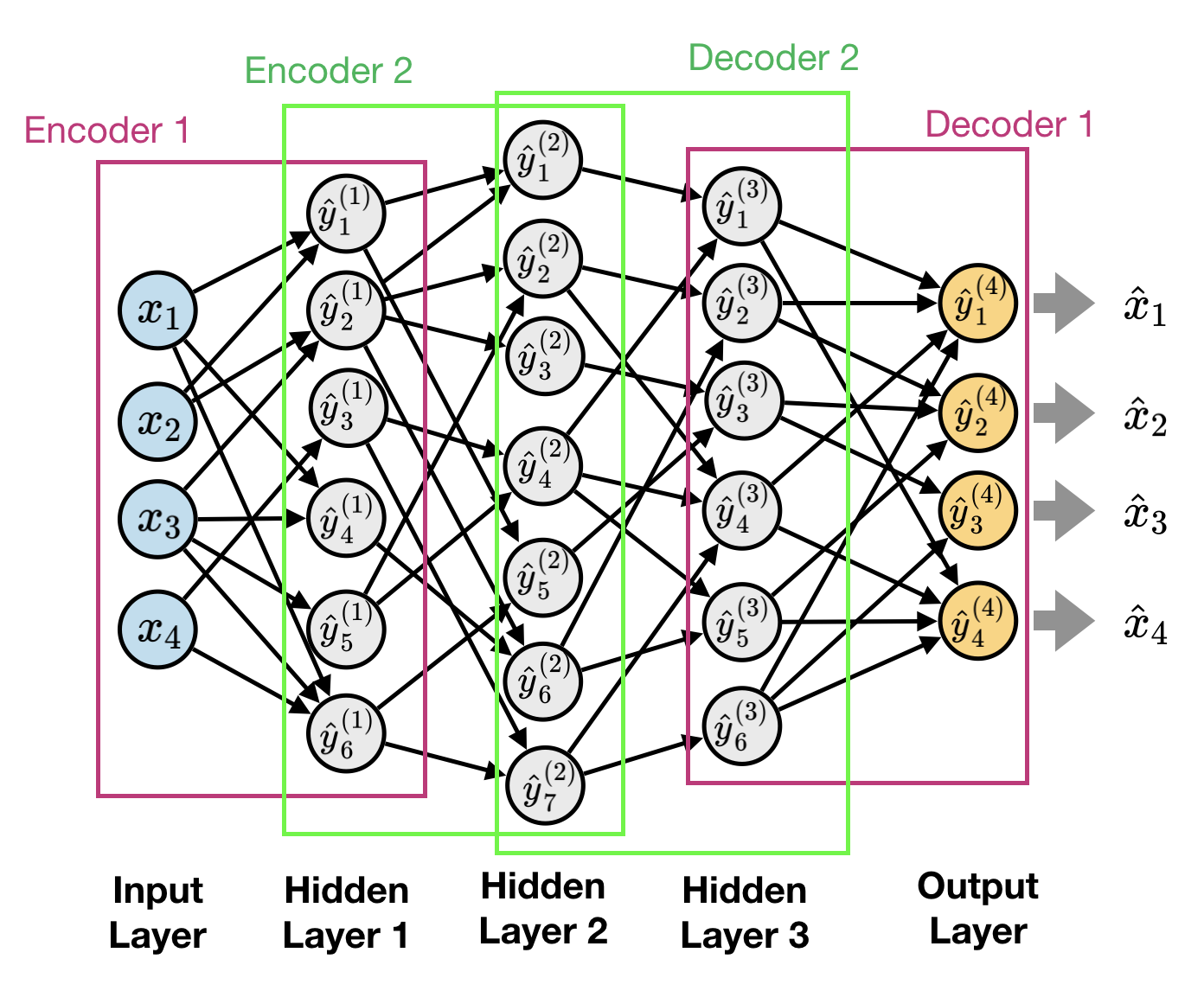} }
\caption{(a) A sparse autoencoder with one hidden layer. (b) An autoencoder with three hideen layers as two single-layer autoencoders stacked together. }
\label{fig:autoenc}
\end{figure}
Autoencoders are a type of neural networks that are used to learn generative models of data in an unsupervised manner. Autoencoders usually consist of one hidden layer that has smaller number of neurons than input or output layers. The goal of autoencoders is to reconstruct the input data at at the output layer, and the outputs of hidden layer can be considered as compressed representation of data. Autoencoders can have more than one hidden layer, but they are usually treated as single-hidden-layer autoencoders stacked together, and they are trained one by one in a greedy fashion (See Fig~\ref{fig:autoenc}). There are many variants of autoencoders, such as denoising autoencoders~\cite{vincent2008denoising}, variational autoencoders~\cite{kingma2013auto}, or sparse autoencoders~\cite{ranzato08sparse2,ranzato07sparse}. In this paper, we will focus on sparse autoencoders. 

In sparse autoencoders, the number of hidden layer neurons can be bigger than the number of neurons in input or output layers. Obviously, an identity function can perfectly reconstruct the input at the output layer in this case. However, what we are interested in is finding a sparse  representation of the data, similar to sparse coding problem. Hence, we add additional term in the loss function to enforce sparsity of the outputs of the hidden layer. Since the autoencoder network has only $2$ layers, the loss function $E$ is defined as below:
\begin{equation}
E(\bm{W}^1,\bm{W}^2) = \frac{1}{|\chi|}\sum_{b=0}^{|\chi|-1} \epsilon^2(b,\bm{W}^1,\bm{W}^2) + \lambda (|| \bm{W}^{1} ||_F^2 + ||\bm{W}^{2}||_F^2) + \beta \sum_{i=0}^{N_1-1} \textnormal{KL}(\rho || \hat{\rho}_i),
\end{equation}
where KL denotes Kullback-Leibler divergence, $\rho$ is desired sparsity and $\hat{\rho}_i$ is the sample sparsity of the second layer's activation averaged over the entire data set, \emph{i.e.,} $\hat{\rho}_i = \frac{1}{|\chi|} \sum_{b=0}^{|\chi|-1} Y^1_{i,b} $ where $\bm{Y}^1 = f(\bm{W}^1\bm{X}^1) $ applied element-wise and $\bm{X}^1 \in \mathcal{R}^{N_0 \times |\chi|}$ is a matrix whose columns represent all the data points.

As full batch GD is expensive, one tends to use mini-batch SGD where the update depends on the gradient calculated on a mini-batch of $B$ data points instead of the whole data set. In mini-batch SGD, for calculating the gradient $\frac{\partial E(\bm{W}^1(k),\bm{W}^2(k)) }{\partial W^l_{i,j}(k)}$, it easier to redefine the loss function $E(\bm{W}^1(k),\bm{W}^2(k))$ over one mini-batch of data, as follows:
\begin{equation}
E(\bm{W}^1(k),\bm{W}^2(k)) = \frac{1}{B}\sum_{b=0}^{B-1} \epsilon^2(b,\bm{W}^1(k),\bm{W}^2(k)) + \lambda (|| \bm{W}^{1}(k) ||_F^2 + ||\bm{W}^{2}(k)||_F^2) + \beta \sum_{i=0}^{N_1-1} \textnormal{KL}(\rho || \hat{\rho}_i),
\end{equation}
where $\hat{\rho}_i$ is the sample sparsity of the second layer's activation averaged over the mini-batch, \emph{i.e.,} $\hat{\rho}_i = \frac{1}{B} \sum_{b=0}^{B-1} Y^1_{i,b} (k)$ where $\bm{Y}^1(k) = f(\bm{W}^1(k)\bm{X}^1(k)) $ applied element-wise and $\bm{X}^1(k) \in \mathcal{R}^{N_0 \times B}$ is a matrix whose columns represent the $B$ data points chosen at the $k$-th iteration.

Since the additional term in the loss function only affects $\bm{W}^1$, the update of $\bm{W}^2$ remains the same and rewrite the update of $\bm{W}^1$ as:
\begin{equation}\label{eq:autoenc_update}
\bm{W}^1(k+1) = (1-\eta \lambda) \bm{W}^1(k) + \eta \bm{\Delta}_{\textnormal{auto}}^1(k) \bm{X}^1(k)^T.
\end{equation}
We now derive $\bm{\Delta}^1_{\textnormal{auto}}(k)$  for the backpropagation of autoencoders in the following lemma.

\begin{lem}\label{lem:backprop_autoenc}
\begin{equation}\label{eq:Delta_auto}
\bm{\Delta}_{\textnormal{auto}}^1(k) = \left((\bm{W}^2(k))^T \bm{\Delta}^{2}(k) - \bm{Q}_{\bm{\hat{\rho}}}(k) \right) \circ f'(\bm{S}^1(k)),
\end{equation}
where 
$$\bm{Q}_{\bm{\hat{\rho}}} = \left[ \vphantom{\begin{array}{c}1\\1\\1\\1\end{array}}
                  \smash{\underbrace{
                      \begin{array}{cccc}
                             Q(\hat{\rho}_0) & Q(\hat{\rho}_0) & \cdots & Q(\hat{\rho}_0)\\
\vdots & \vdots & \ddots & \vdots \\
Q(\hat{\rho}_{N_1-1}) & Q(\hat{\rho}_{N_1-1}) & \cdots & Q(\hat{\rho}_{N_1-1})\\
                      \end{array}
                      }_{B}}
              \right], \; \textnormal{ for } \; Q(\hat{\rho}_i) = -\frac{\rho}{\hat{\rho}_i} + \frac{1-\rho}{1-\hat{\rho}_i}.$$
\\

\end{lem}
\begin{proof}
We will omit $k$ here for the sake of simplicity. Let us first derive the partial derivative of the second term in the loss function first.
\begin{align*}
\frac{\partial}{\partial W_{i,j}^2} \sum_{\tau=0}^{N_1-1} \textnormal{KL}(\rho || \hat{\rho}_{\tau}) &= \sum_{\tau=0}^{N_1-1} \frac{\partial}{\partial W_{i,j}^2} \hat{\rho}_{\tau} \frac{\partial}{\partial \hat{\rho}_{\tau}} \textnormal{KL}(\rho || \rho_{\tau}) \\ 
&= \frac{\partial}{\partial W_{i,j}^2}\hat{\rho}_i \frac{\partial}{\partial \hat{\rho}_i} \textnormal{KL}(\rho || \hat{\rho}_i) \\ 
&= \frac{\partial}{\partial W_{i,j}} \frac{1}{B} \sum_{\nu=0}^{B-1} \bm{Y}^1_{i,\nu} \cdot \left( -\frac{\rho}{\hat{\rho}_i} + \frac{1-\rho}{1-\hat{\rho}_i}\right) \\ 
&= \frac{1}{B} \sum_{\nu=0}^{B-1} \frac{\partial}{\partial W_{i,j}} f(S_{i,\nu}^1) \cdot \left( -\frac{\rho}{\hat{\rho}_i} + \frac{1-\rho}{1-\hat{\rho}_i}\right) \\
&= \frac{1}{B} \sum_{\nu=0}^{B-1} X_{j,\nu} f'(S_{i,\nu}^1) \cdot \left( -\frac{\rho}{\hat{\rho}_i} + \frac{1-\rho}{1-\hat{\rho}_i}\right).
\end{align*}

Now we obtain the derivative of the first two terms in the loss function. 
\begin{align*}
\frac{\partial}{\partial W_{i,j}^1} \epsilon^2 +  \frac{\partial}{\partial W_{i,j}^1} \beta \sum_{\tau=0}^{N_2-1} \textnormal{KL}(\rho || \hat{\rho}_{\tau}(k)) 
&= -\frac{1}{B} \sum_{\nu=0}^{B-1} \Delta_{i, \nu}^1 X_{j, \nu} + \frac{\beta}{B} \sum_{\nu=0}^{B-1} f'(S_{i,\nu}^1) Q(\hat{\rho}_i) X_{j,\nu} \\
&= -\frac{1}{B} \left( \sum_{\nu=0}^{B-1} 
 (\bm{\Delta}_{:,\nu}^2)^T \bm{W}^1_{:,i} \cdot f'(S_{i,\nu}) X_{j,\nu} - \beta f'(S_{i,\nu}^1) Q(\hat{\rho}_i) X_{j,\nu} \right) \\
 &= -\frac{1}{B} \left( \sum_{\nu=0}^{B-1} 
 \left( (\bm{\Delta}_{:,\nu}^2)^T \bm{W}^1_{:,i} - \beta Q(\hat{\rho}_i) \right) f'(S_{i,\nu}^1)  X_{j,\nu} \right). 
\end{align*}
Hence,
\begin{equation}
\frac{\partial}{\partial \bm{W}^1} E = \left(\left((\bm{W}^2)^T \bm{\Delta}^2 - \bm{Q}_{\bm{\hat{\rho}}} \right) \circ f'(\bm{S}^1)\right) \bm{X}^T + \lambda \bm{W}^1.
\end{equation}
\end{proof}

{\color{black}{
\section{Proof of \Cref{thm:main} and  a discussion on error models and decoding techniques}
\label{appendix:decoding}

In this appendix, we provide a proof of \Cref{thm:main} along with discussion on the error models and decoding techniques. For completeness, we restate some of the descriptions already presented in \Cref{sec:background}.

\subsection{Notations and Definitions.}

Recall the channel coding scenario. Let $\bm{q}$ be a $Q \times 1$ vector consisting of $Q$ real-valued symbols. The received output vector is as follows:
\begin{equation}
\bm{z}=\bm{G}^T\bm{q} + \bm{e}.
\end{equation}
Here $\bm{G}$ is the generator matrix of a $(P,Q)$ real number MDS Code and $\bm{e}$ is the $P \times 1$ error vector that corrupts the codeword $\bm{G}^T\bm{q}$. The locations of the codeword that are affected by errors is a subset $\mathcal{A} \subseteq \{0,1,\ldots,P-1\}$. The elements of $\bm{e}$ indexed in $\mathcal{A}$ denote the corresponding values of additive error while the rest are $0$.

We use the notation $\mathbb{Q}$, $\mathbb{Z}$, $\mathbb{E}$ and $\widehat{\mathbb{E}}$ to denote random vectors corresponding to the symbol vector, output vector, the true error vector and the estimated error vector respectively. Note that $\mathbb{Q}$ and $\mathbb{E}$ are independent. The vectors $\mathbb{Q}$ and $\mathbb{E}$ are independent of each other.
Thus, the received output random vector is as follows:
\begin{equation}
\mathbb{Z}=\bm{G}^T\mathbb{Q} + \mathbb{E}.
\end{equation}
To denote the entire sample space of any random vector $\mathbb{Q}$, we use the notation $\Omega_{\mathbb{Q}}$. We also use the notation $\int_{\Omega_{\mathbb{Q}}} (\text{some function}) d\bm{q}$ to denote the multiple integral with respect to all the elements of the vector $\bm{q}$ as follows: $$\int_{\Omega_{\mathbb{Q}_0}} \ldots \int_{\Omega_{\mathbb{Q}_{Q-1}}} (\text{some function}) dq_0 \ dq_1 \ \ldots dq_{Q-1}.$$
We also let $\Pr{(\cdot)}$ denote the probability of an event and $p(\mathbb{Q}=\bm{q})$ denote the pdf of $\mathbb{Q}$. These notations also apply to other random variables. Lastly, $Null(\cdot)$ denotes the null-space of a matrix, and $||\cdot||_0$ denotes the number of non-zero elements of a vector. 

\textbf{Adversarial Error Model:} The subset $\mathcal{A}$ satisfies $|\mathcal{A} | \leq \lfloor \frac{P-Q}{2} \rfloor$, with no specific assumptions on the locations or values of the errors and they may be chosen advarsarially.

\textbf{Probabilistic Error Model:} The subset $\mathcal{A}$ can be of any cardinality from $0$ to $P$, and these locations may be chosen adversarially. However, given $\mathcal{A}$, the elements of $\mathbb{E}$ indexed in $\mathcal{A}$ are drawn from iid Gaussian distributions and the rest are $0$. Also note that $\mathbb{Q}$ and $\mathbb{E}$ are independent.

\subsection{Formal Proof of \Cref{thm:main}.}

\newtheorem*{thm*}{Theorem}

\begin{thm*}[\Cref{thm:main} Restated]
 Under the Probabilistic Error Model for channel coding, the decoder of a $(P,Q)$ MDS Code can perform the following:
\begin{itemize}
\item [1.] It can detect the occurrence of errors with probability $1$, irrespective of the number of errors that occurred.
\item [2.] If the number of errors that occurred is less than or equal to $P-Q-1$, then all those errors can be corrected with probability $1$, even without knowing in advance that how many errors actually occurred.
\item [3.] If the number of errors that occurred is more than $P-Q-1$, then the decoder is able to determine that the errors are too many to be corrected and declare a ``decoding failure'' with probability $1$.
\end{itemize}
\end{thm*}




Let $\bm{H}$ be the $(P-Q)\times P$ sized parity check matrix of the MDS code, such that $\bm{H}\bm{G}^T=\bm{0}$. We first propose the following decoding algorithm (Algorithm~\ref{algo:decoding} restated) to produce $\widehat{\bm{e}}$, as an estimate of $\bm{e}$, for a given channel output $\bm{z}$. If an $\widehat{\bm{e}}$ is obtained, the decoder can uniquely solve for $\widehat{\bm{q}}$ from the linear set of equations $\bm{G}^T\widehat{\bm{q}}=\bm{z}-\widehat{\bm{e}}$.

\begin{algorithm*}
\begin{algorithmic}[1]
\STATE \textbf{If} $\bm{H}\bm{z}=\bm{0}$, \textbf{then} declare ``no errors detected'' and produce $\widehat{\bm{e}}=\bm{0}$. 
\STATE \hspace{1cm} \textbf{Else} find an $\widehat{\bm{e}}$ as follows: $\widehat{\bm{e}} =
\arg \min ||\bm{e}||_0 \text{ such that } \bm{H}\bm{e}=\bm{H}\bm{z}.$ 
\STATE \hspace{1.5cm} \textbf{If} the obtained $\widehat{\bm{e}}$ is such that $||\widehat{\bm{e}}||_0\leq P-Q-1$, \textbf{then} produce this estimate $\widehat{\bm{e}}$.
\STATE \hspace{2cm} \textbf{Else} declare a ``decoding failure.'' 
\end{algorithmic}
\end{algorithm*}





Now we will show that the three claims of \Cref{thm:main} hold using this proposed decoding algorithm. Our first claim is that Algorithm~\ref{algo:decoding} detects the occurrence of errors almost surely when it checks if $\bm{H}\mathbb{Z}=\bm{0}$.

\noindent {\bf Claim 1:}   Under the Probabilistic Error Model, $$\Pr{(\bm{H}\mathbb{Z} = \bm{0}|\ \mathbb{E}\neq \bm{0} )}=0.$$ 

\begin{proof}[Proof of Claim $1$ in \Cref{thm:main}]
\begin{align}
\Pr{(\bm{H}\mathbb{Z} = \bm{0}|\ \mathbb{E}\neq \bm{0} )} &= \frac{\Pr{(\bm{H}\mathbb{Z} = \bm{0},\ \mathbb{E}\neq \bm{0} )}}{\Pr{(\mathbb{E}\neq \bm{0} )}} \nonumber \\
& = \int_{\Omega_{\mathbb{Q}}}   \frac{\Pr{(\bm{H}\mathbb{Z} = \bm{0},\ \mathbb{E}\neq \bm{0}| \mathbb{Q}=\bm{q} )} }{\Pr{(\mathbb{E}\neq \bm{0} )} } p(\mathbb{Q}=\bm{q}) d \bm{q} \nonumber \\
&= \int_{\Omega_{\mathbb{Q}}} \frac{\Pr{(\bm{H}\mathbb{Z} = \bm{0},\ \mathbb{E}\neq \bm{0}| \mathbb{Q}=\bm{q} )}}{ \Pr{(\mathbb{E}\neq \bm{0} )}} p(\mathbb{Q}=\bm{q}) d \bm{q}.
\label{eq:clm1_wrt_x}
\end{align}

We examine a term from \Cref{eq:clm1_wrt_x} as follows:
\begin{align}
& \Pr{(\bm{H}\mathbb{Z} = \bm{0},\ \mathbb{E}\neq \bm{0}| \mathbb{Q}=\bm{q} )}\nonumber \\
&=\sum_{  \substack{ \text{all possible}\\ \mathcal{A} \subseteq \{0,1,\ldots,P-1\} }  } 
\Pr{\text{(Error pattern is $\mathcal{A}$)}}
\int_{ \Omega_{\mathbb{E}_{\mathcal{A}}} \backslash \{\bm{0} \} }  \Pr{(\bm{H}\mathbb{Z} = \bm{0}| \mathbb{E}_{\mathcal{A}}= \bm{e}_{\mathcal{A}}, \mathbb{Q}=\bm{q}) } p(\mathbb{E}_{\mathcal{A}}=\bm{e}_{\mathcal{A}}) d \bm{e}_{\mathcal{A}}  \nonumber\\
&=  \sum_{  \substack{ \text{all possible}\\ \mathcal{A} \subseteq \{0,1,\ldots,P-1\} }  } 
\Pr{\text{(Error pattern is $\mathcal{A}$)}}
\int_{ \Omega_{\mathbb{E}_{\mathcal{A}}} \backslash \{\bm{0} \} } I[  \bm{e}_{\mathcal{A}} : \bm{H}_{\mathcal{A}}\bm{e}_{\mathcal{A}}=\bm{0}  ]   p(\mathbb{E}_{\mathcal{A}}=\bm{e}_{\mathcal{A}}) d \bm{e}_{\mathcal{A}}.
\label{eq:clm1_wrt_e}\\
&  \hspace{10cm} \text{[since }  \bm{H}\bm{z}= \bm{H}(\bm{G}^T\bm{q}+\bm{e})=\bm{H}\bm{e}= \bm{H}_{\mathcal{A}}\bm{e}_{\mathcal{A}}] \nonumber 
\end{align}
Here $\bm{e}_{\mathcal{A}}$ and $\mathbb{E}_{\mathcal{A}}$ denote the elements of the vectors $\bm{e}$ or $\mathbb{E}$ respectively, indexed by $\mathcal{A}$ and $\bm{H}_{\mathcal{A}}$ denotes the columns of $\bm{H}$ indexed by $\mathcal{A}$. Moreover, $I[ \bm{e}_{\mathcal{A}} : \bm{H}_{\mathcal{A}}\bm{e}_{\mathcal{A}}=\bm{0}  ]$  denotes the indicator function of $\bm{e}_{\mathcal{A}}$, which is $1$ when its condition is satisfied, and is $0$ otherwise.   

For a particular realization of $\mathbb{E}=\bm{e}$, if $\bm{H}\bm{e}=0$, then it means that the true error vector $\bm{e}$ lies in $Null(\bm{H})$. When $0<|\mathcal{A}| < P-Q+1$, then $\bm{e}$ can never lie in $Null(\bm{H})$ because $\bm{H}$ is also the generator of a $(P,P-Q)$ MDS Code and hence any vector lying in its null space has at least $P-Q+1$ non-zeros. When $|\mathcal{A}| \geq P-Q+1$, if $\bm{e}$ lies in $Null(\bm{H})$, then
\begin{equation}
\bm{H}_{\mathcal{A}}\bm{e}_{\mathcal{A}} = \bm{0}.
\label{eq:detection}
\end{equation}
Note that, any $\bm{e}_{\mathcal{A}}$ satisfying \Cref{eq:detection}, lies in a sub-space of dimension $|\mathcal{A}| - (P-Q)$, which becomes a measure $0$ subset for random vector $\mathbb{E}_{\mathcal{A}}$, whose all $|\mathcal{A}|$ entries are drawn from iid Gaussian distributions. Thus, 
\begin{align}
&\int_{ \Omega_{\mathbb{E}_{\mathcal{A}}} \backslash \{\bm{0} \} } I[  \bm{e}_{\mathcal{A}} : \bm{H}_{\mathcal{A}}\bm{e}_{\mathcal{A}}=\bm{0}  ]   p(\mathbb{E}_{\mathcal{A}}=\bm{e}_{\mathcal{A}}) d \bm{e}_{\mathcal{A}} =0, \ \  \forall \ \mathcal{A}\subseteq \{0,1,\ldots,P-1\}.
\label{eq:clm1_measure_0}
\end{align}
This holds because for $|\mathcal{A} |< P-Q+1$, the set $\{ \bm{e}_{\mathcal{A}} \in \Omega_{\mathbb{E}_{\mathcal{A}}}\backslash \{\bm{0} \} : \bm{H}_{\mathcal{A}}\bm{e}_{\mathcal{A}}=\bm{0}\} $,   is an empty set and for $ |\mathcal{A}| \geq P-Q+1$, the set $\{ \bm{e}_{\mathcal{A}} \in \Omega_{\mathbb{E}_{\mathcal{A}}}\backslash \{\bm{0} \} : \bm{H}_{\mathcal{A}}\bm{e}_{\mathcal{A}}=\bm{0}\} $,  lies in a subspace of dimension at most $|\mathcal{A}|-(P-Q)$ which becomes a measure $0$ set for random vector $\mathbb{E}_A$ whose all $|\mathcal{A}|$ entries are drawn from iid Gaussian distribution. Thus,
\begin{align}
&\Pr{(\bm{H}\mathbb{Y} = \bm{0}|\ \mathbb{E}\neq \bm{0} )} \nonumber \\
&\overset{ \Cref{eq:clm1_wrt_x}}{=} \int_{\Omega_{\mathbb{Q}}} \frac{\Pr{(\bm{H}\mathbb{Z} = \bm{0},\ \mathbb{E}\neq \bm{0}| \mathbb{Q}=\bm{q} )}}{ \Pr{(\mathbb{E}\neq \bm{0} )}} p(\mathbb{Q}=\bm{q}) d \bm{q} \nonumber \\
&\overset{\Cref{eq:clm1_wrt_e}}{=} \int_{\Omega_{\mathbb{Q}}} \sum_{ \substack{ \text{all possible}\\ \mathcal{A} \subseteq \{0,1,\ldots,P-1\} } } 
\frac{\Pr{\text{(Error pattern is $\mathcal{A}$)}}}{ \Pr{(\mathbb{E}\neq \bm{0})} }
\int_{ \Omega_{\mathbb{E}_{\mathcal{A}}} \backslash \{\bm{0} \} } I[  \bm{e}_{\mathcal{A}} : \bm{H}_{\mathcal{A}}\bm{e}_{\mathcal{A}}=\bm{0}  ]   p(\mathbb{E}_{\mathcal{A}}=\bm{e}_{\mathcal{A}}) \ d \bm{e}_{\mathcal{A}} \ p(\mathbb{Q}=\bm{q}) \ d \bm{q} \nonumber \\
&\overset{\Cref{eq:clm1_measure_0}}{=} \int_{\Omega_{\mathbb{Q}}} \sum_{ \substack{ \text{all possible}\\ \mathcal{A} \subseteq \{0,1,\ldots,P-1\} } } 
\frac{\Pr{\text{(Error pattern is $\mathcal{A}$)}}}{\Pr{(\mathbb{E}\neq \bm{0})}}\cdot 0 \cdot p(\mathbb{Q}=\bm{q}) \ d \bm{q} = 0.
\end{align}
\end{proof}

The next two claims show the error correction capability of Algorithm~\ref{algo:decoding}. Observe that, for a particular realization of $\mathbb{E}=\bm{e}$, Algorithm~\ref{algo:decoding} can have three possible outcomes. Either the algorithm produces an $\hat{\mathbb{E}}=\bm{e}$, or an $\hat{\mathbb{E}}=\bm{e}' \neq \bm{e}$, or it declares a decoding failure. Thus,
$$ \Pr{(\hat{\mathbb{E}} = \mathbb{E})} + \Pr{(\hat{\mathbb{E}} \neq \mathbb{E})} + \Pr{(\text{Decoding Failure})}=1.$$

\noindent{\bf Claim $2$:} Under the Probabilistic Error Model,
$$\Pr{(\hat{\mathbb{E}} = \mathbb{E}| \ || \mathbb{E}||_0 \leq P-Q-1 )} =1.$$
Given that $||\bm{e}||_0 \leq P-Q-1$, there is at least one vector, which is the true $\bm{e}$, which lies in the search-space of the decoding algorithm and hence the declaration of a decoding failure does not arise. Thus, it is sufficient to show: 
$$\Pr{(\hat{\mathbb{E}} \neq \mathbb{E}| \ || \mathbb{E}||_0 \leq P-Q-1 )} =0.$$

\noindent{\bf Claim $3$:} Under the Probabilistic Error Model, $$\Pr{(\text{Decoding Failure }| \ || \mathbb{E}||_0 > P-Q-1 )} =1.$$ Given $|| \bm{e}||_0 > P-Q-1$, the case of $\hat{\mathbb{E}}=\bm{e}$ cannot arise because the decoder only searches for a solution with number of non-zeros at most $P-Q-1$. Thus, it is sufficient to show: $$\Pr{(\hat{\mathbb{E}} \neq \mathbb{E}| \ || \mathbb{E}||_0 > P-Q-1 )} =0.$$ 





\begin{proof}[Proof of Claims $2$ and $3$ in \Cref{thm:main}]
Essentially, to prove both Claims $2$ and $3$, it is sufficient to show that $$\Pr{(\hat{\mathbb{E}} \neq \mathbb{E})}=0.$$

Observe that,
\begin{equation}
\Pr{(\hat{\mathbb{E}} \neq \mathbb{E} )}   = \int_{\Omega_{\mathbb{Q}}}   \Pr{(\hat{\mathbb{E}} \neq \mathbb{E}  | \mathbb{Q}=\bm{q} )}  p(\mathbb{Q}=\bm{q}) d \bm{q}. 
\label{eq:clm2_wrt_x}
\end{equation}

Next, we examine a term in \Cref{eq:clm2_wrt_x} as follows:
\begin{align}
&\Pr{(\hat{\mathbb{E}} \neq \mathbb{E}  | \mathbb{Q}=\bm{q} )} \nonumber \\
&=  \sum_{  \substack{ \text{all possible}\\ \mathcal{A} \subseteq \{0,1,\ldots,P-1\}}}
\Pr{\text{(Error pattern is $\mathcal{A}$)}}
\int_{ \Omega_{\mathbb{E}_{\mathcal{A}}} \backslash \{\bm{0} \} } \Pr{(\hat{\bm{e}} \neq \bm{e} | \mathbb{E}_{\mathcal{A}}=\bm{e}_{\mathcal{A}}, \mathbb{Q}=\bm{q} )} p(\mathbb{E}_{\mathcal{A}}=\bm{e}_{\mathcal{A}})d \bm{e}_{\mathcal{A}} \nonumber \\
& \leq \sum_{  \substack{ \text{all possible}\\ \mathcal{A} \subseteq \{0,1,\ldots,P-1\} }} 
\Pr{\text{(Error pattern is $\mathcal{A}$)}}
\int_{ \Omega_{\mathbb{E}_{\mathcal{A}}} \backslash \{\bm{0} \} } I[\bm{e}_{\mathcal{A}}: \bm{e}_{\mathcal{A}} \in \mathcal{E}_{\mathcal{A}} ]\ p(\mathbb{E}_{\mathcal{A}}=\bm{e}_{\mathcal{A}}) \ d \bm{e}_{\mathcal{A}}
\label{eq:clm2_wrt_e}
\end{align}
where $\mathcal{E}_{\mathcal{A}}$ is the set of $\bm{e}_{\mathcal{A}}$ for which there exists another $\bm{e}' \neq \bm{e}$ such that $\bm{H}\bm{e}=\bm{H}\bm{e}'$, $||\bm{e}'||_0\leq |\mathcal{A}|$ and $ ||\bm{e}'||_0\leq P-Q-1$. Note that, this set $\mathcal{E}_{\mathcal{A}}$ is a super-set of the set of $\bm{e}_{\mathcal{A}}$s for which $\bm{e}' \neq \bm{e}$ because the existence of another such $\bm{e}'$ does not necessarily mean that the decoding algorithm always picks it over the correct $\bm{e}$. Examining a term in \Cref{eq:clm2_wrt_e},
\begin{align}
&\int_{ \Omega_{\mathbb{E}_{\mathcal{A}}} \backslash \{\bm{0} \} } I[\bm{e}_{\mathcal{A}}: \bm{e}_{\mathcal{A}} \in \mathcal{E}_{\mathcal{A}} ]\ p(\mathbb{E}_{\mathcal{A}}=\bm{e}_{\mathcal{A}}) \ d \bm{e}_{\mathcal{A}} \nonumber \\
&=\int_{ \Omega_{\mathbb{E}_{\mathcal{A}}} \backslash \{\bm{0} \} } \sum_{\substack{ \text{all possible}\\ \mathcal{A}' \subseteq \{0,1,\ldots,P-1\}  }} I[\bm{e}_{\mathcal{A}}: \bm{e}_{\mathcal{A}} \in \mathcal{E}_{\mathcal{A},\mathcal{A}'} ]\ p(\mathbb{E}_{\mathcal{A}}=\bm{e}_{\mathcal{A}}) \ d \bm{e}_{\mathcal{A}} \nonumber \\
& = \sum_{\substack{ \text{all possible}\\ \mathcal{A}' \subseteq \{0,1,\ldots,P-1\}  }} \int_{ \Omega_{\mathbb{E}_{\mathcal{A}}} \backslash \{\bm{0} \} }  I[\bm{e}_{\mathcal{A}}: \bm{e}_{\mathcal{A}} \in \mathcal{E}_{\mathcal{A},\mathcal{A}'} ]\ p(\mathbb{E}_{\mathcal{A}}=\bm{e}_{\mathcal{A}}) \ d \bm{e}_{\mathcal{A}},
\label{eq:clm2_sparsity_split}
\end{align}
where $\mathcal{E}_{\mathcal{A},\mathcal{A}'} \subseteq \mathcal{E}_{\mathcal{A}}$ is the set of all $\bm{e}_{\mathcal{A}}$s for which $ \exists \bm{e}' \text{ with non-zero indices } \mathcal{A}' $ such that   $\bm{H}\bm{e}=\bm{H}\bm{e}'$, $||\bm{e}'||_0\leq |\mathcal{A}|$ and $ ||\bm{e}'||_{0}\leq P-Q-1$. Note that, $$ \mathcal{E}_{\mathcal{A}}= \cup_{\mathcal{A}'} \ \mathcal{E}_{\mathcal{A},\mathcal{A}'} .$$

Observe that when $|\mathcal{A}'|$ is greater than $\min \{|\mathcal{A}|, P-Q-1 \}$ for a given $\bm{e}_{\mathcal{A}}$, there is no possible $\bm{e}'$ that satisfies all the aforementioned conditions, and thus $\mathcal{E}_{\mathcal{A},\mathcal{A}'}$ is an empty set. Let us see what happens when $|\mathcal{A}'| \leq \min \{|\mathcal{A}|, P-Q-1 \}$.

First we will show that $\mathcal{E}_{\mathcal{A},\mathcal{A}'}$ is also an empty set when  $| \mathcal{A} \cup \mathcal{A}'  | = |\mathcal{A}'|$.
For $\bm{H}\bm{e}=\bm{H}\bm{e}'$ to hold,
\begin{equation}
\bm{e}=\bm{e}' + \bm{h},
\label{eq:e_prime}
\end{equation}
where $\bm{h}\in Null(\bm{H}) \backslash \{\bm{0} \}$. As $\bm{H}$ is also the generator matrix of a $(P,P-Q)$ MDS Code (parity check of a $(P,Q)$ MDS Code), for any $\bm{h}\in Null(\bm{H}) \backslash \{\bm{0} \}$, the number of non-zeros is $||\bm{h}||_0 \geq P-Q+1$. On the other hand, $ |\mathcal{A}'| \leq P-Q-1$ which is less than $P-Q+1$. Thus $$| \mathcal{A} \cup \mathcal{A}'  | \geq ||\bm{h}||_0 \geq P-Q+1 > P-Q-1 \geq |\mathcal{A}' |,$$ for any such $\bm{e}'$ (with non-zero elements $\mathcal{A}'$) to exist, for a given $\bm{e}_{\mathcal{A}}$. 

Therefore, if there exists such an $\bm{e}'$ for which \Cref{eq:e_prime} holds, then $|\mathcal{A} \cup \mathcal{A}'| - |\mathcal{A}'|(>0)$ entries of $\bm{h}(=\bm{e}-\bm{e}')$, indexed by the set $(\mathcal{A} \cup \mathcal{A}')\backslash \mathcal{A}'$ should match exactly with the true error vector $\bm{e}$. For particular $\mathcal{A}$ and $\mathcal{A}'$, satisfying
$| \mathcal{A} \cup \mathcal{A}'  | > |\mathcal{A}'|$ and $|\mathcal{A}'| \leq \min \{|\mathcal{A}|, P-Q-1 \}$,  let us denote the set of indices $(\mathcal{A} \cup \mathcal{A}') \backslash \mathcal{A}'$ as $\mathcal{A}^*$. Observe that,
\begin{equation}
\bm{0}=\bm{H}(\bm{e}-\bm{e}') = \bm{H}\bm{h} =
\bm{H}_{\mathcal{A} \cup \mathcal{A}'}\bm{h}_{\mathcal{A} \cup \mathcal{A}'} 
 = 
\begin{bmatrix} \bm{H}_{\mathcal{A}^*} & \bm{H}_{\mathcal{A}'}
\end{bmatrix} \begin{bmatrix} \bm{h}_{\mathcal{A}^*} \\
\bm{h}_{\mathcal{A}'}
\end{bmatrix}
=\begin{bmatrix} \bm{H}_{\mathcal{A}^*} & \bm{H}_{\mathcal{A}'}
\end{bmatrix} \begin{bmatrix} \bm{e}_{\mathcal{A}^*} \\
\bm{h}_{\mathcal{A}'}
\end{bmatrix},
\label{eq:null_H}
\end{equation}
where $\bm{H}_{\mathcal{A}^*}$, $\bm{H}_{\mathcal{A}'}$ and $\bm{H}_{\mathcal{A} \cup \mathcal{A}'}$ denotes the sub-matrices of the matrix $\bm{H}$ consisting of the columns indexed in sets $\mathcal{A}^*$, $\mathcal{A}'$ and $\mathcal{A} \cup \mathcal{A}' $ respectively. Similarly the sub-scripted $\bm{h}$ denotes sub-vectors of $\bm{h}$ consisting of the elements indexed in the sub-script.

Thus, the set $\mathcal{E}_{\mathcal{A},\mathcal{A}'}$ can be redefined as the set of all $\bm{e}_{\mathcal{A}}$s for which there exists an $\bm{h}=\bm{e}'-\bm{e}$ (and hence an $\bm{h}_{\mathcal{A}'}$) such that $\begin{bmatrix} \bm{H}_{\mathcal{A}^*} & \bm{H}_{\mathcal{A}'}\end{bmatrix} \begin{bmatrix} \bm{e}_{\mathcal{A}^*} \\
\bm{h}_{\mathcal{A}'}
\end{bmatrix}=\bm{0}$ where $\mathcal{A}^*=(\mathcal{A}\cup \mathcal{A}') \backslash \mathcal{A}'$. Returning to where we left off in \Cref{eq:clm2_sparsity_split},
\begin{align}
&\sum_{\substack{ \text{all possible}\\ \mathcal{A}' \subseteq \{0,1,\ldots,P-1\}  }} \int_{ \Omega_{\mathbb{E}_{\mathcal{A}}} \backslash \{\bm{0} \} }  I[\bm{e}_{\mathcal{A}}: \bm{e}_{\mathcal{A}} \in \mathcal{E}_{\mathcal{A},\mathcal{A}'} ]\ p(\mathbb{E}_{\mathcal{A}}=\bm{e}_{\mathcal{A}}) \ d \bm{e}_{\mathcal{A}} \nonumber \\
&= \sum_{\substack{ \text{all possible}\\ \mathcal{A}' \subseteq \{0,1,\ldots,P-1\} \\  |\mathcal{A}'|\leq  \min \{|\mathcal{A}|, P-Q-1\},\ |\mathcal{A}\cup\mathcal{A}'|> |\mathcal{A}'|   }} \int_{ \Omega_{\mathbb{E}_{\mathcal{A}}} \backslash \{\bm{0} \} }  I[\bm{e}_{\mathcal{A}}: \bm{e}_{\mathcal{A}} \in \mathcal{E}_{\mathcal{A},\mathcal{A}'} ]\ p(\mathbb{E}_{\mathcal{A}}=\bm{e}_{\mathcal{A}}) \ d \bm{e}_{\mathcal{A}} \nonumber
\\
& = \sum_{\substack{ \text{all possible}\\ \mathcal{A}' \subseteq \{0,1,\ldots,P-1\} \\  |\mathcal{A}'|\leq  \min \{|\mathcal{A}|, P-Q-1\},\ |\mathcal{A}\cup\mathcal{A}'|> |\mathcal{A}'|  }} \int_{ \Omega_{\mathbb{E}_{\mathcal{A}\backslash \mathcal{A}* }} \backslash \{\bm{0} \} } \int_{ \Omega_{\mathbb{E}_{\mathcal{A}^* }} \backslash \{\bm{0} \} }  I[\bm{e}_{\mathcal{A}^*}: \bm{e}_{\mathcal{A}^*} \in \mathcal{E}^*_{\mathcal{A},\mathcal{A}',\bm{e}_{\mathcal{A}\backslash \mathcal{A}^*}} ]\nonumber \\
&\hspace{7cm} p(\mathbb{E}_{\mathcal{A}^*}=\bm{e}_{\mathcal{A}^*}| \mathbb{E}_{\mathcal{A}\backslash \mathcal{A}^*}=\bm{e}_{\mathcal{A}\backslash \mathcal{A}^*}   ) p(\mathbb{E}_{\mathcal{A}\backslash \mathcal{A}^*}=\bm{e}_{\mathcal{A}\backslash \mathcal{A}^*} ) \ d \bm{e}_{\mathcal{A}^*} d \bm{e}_{\mathcal{A}\backslash \mathcal{A}^* } \nonumber \\
& = \sum_{\substack{ \text{all possible}\\ \mathcal{A}' \subseteq \{0,1,\ldots,P-1\} \\  |\mathcal{A}'|\leq  \min \{|\mathcal{A}|, P-Q-1\}, \\
 |\mathcal{A}\cup\mathcal{A}'|> |\mathcal{A}'| }} \int_{ \Omega_{\mathbb{E}_{\mathcal{A}\backslash \mathcal{A}* }} \backslash \{\bm{0} \} } \int_{ \Omega_{\mathbb{E}_{\mathcal{A}^* }} \backslash \{\bm{0} \} }  I[\bm{e}_{\mathcal{A}^*}: \bm{e}_{\mathcal{A}^*} \in \mathcal{E}^*_{\mathcal{A},\mathcal{A}',\bm{e}_{\mathcal{A}\backslash \mathcal{A}^*}} ]p(\mathbb{E}_{\mathcal{A}^*}=\bm{e}_{\mathcal{A}^*}   ) p(\mathbb{E}_{\mathcal{A}\backslash \mathcal{A}^*}=\bm{e}_{\mathcal{A}\backslash \mathcal{A}^*} ) \ d \bm{e}_{\mathcal{A}^*} d \bm{e}_{\mathcal{A}\backslash \mathcal{A}^* } 
\label{eq:clm2_measure_0}
\end{align}
where $\bm{e}_{\mathcal{A}\backslash \mathcal{A}^*}$ denotes the elements of $\bm{e}$ that are independent and outside of those in set $\mathcal{A}^*$, and the set $\mathcal{E}^*_{\mathcal{A},\mathcal{A}',\bm{e}_{\mathcal{A}\backslash \mathcal{A}^* } }$ is defined as the set of all $\bm{e}_{\mathcal{A}^*}$s for a given $\bm{e}_{\mathcal{A}\backslash \mathcal{A}^*} $, such that there exists an $\bm{h}_{\mathcal{A}'}$ satisfying $$\begin{bmatrix} \bm{H}_{\mathcal{A}^*} & \bm{H}_{\mathcal{A}'} \end{bmatrix} \begin{bmatrix} \bm{e}_{\mathcal{A}^*} \\
\bm{h}_{\mathcal{A}'} 
\end{bmatrix}= \bm{0}.$$

Now, any vector $\begin{bmatrix} \bm{e}_{\mathcal{A}^*} \\
\bm{h}_{\mathcal{A}'} 
\end{bmatrix}= \begin{bmatrix} \bm{h}_{\mathcal{A}^*} \\
\bm{h}_{\mathcal{A}'} 
\end{bmatrix}$ satisfying \Cref{eq:null_H}, lies in a sub-space of dimension $|\mathcal{A}\cup\mathcal{A}'|-(P-Q)$. Thus, $\bm{e}_{\mathcal{A}^*}$ being a sub-vector of this vector also lies in a sub-space of dimension at most $|\mathcal{A}\cup\mathcal{A}'|-(P-Q)$. Thus, essentially the entire set $\mathcal{E}^*_{\mathcal{A},\mathcal{A}',\bm{e}_{\mathcal{A}\backslash \mathcal{A}^* } }$  lies within a sub-space of dimension at most $|\mathcal{A}\cup\mathcal{A}'|- (P-Q)$, and is thus a measure $0$ subset for a random vector $\mathbb{E}_{\mathcal{A}^*}$ whose all $|\mathcal{A}^*|= |\mathcal{A}\cup\mathcal{A}' |-|\mathcal{A}'|$ entries are iid Gaussian random variables since $|\mathcal{A}\cup\mathcal{A}' |-|\mathcal{A}'| \geq |\mathcal{A}\cup\mathcal{A}' | - (P-Q-1) > |\mathcal{A}\cup\mathcal{A}' |-(P-Q)$. 
Thus, examining one of the terms in \Cref{eq:clm2_measure_0},
\begin{align}
& \int_{ \Omega_{\mathbb{E}_{\mathcal{A}\backslash \mathcal{A}* }} \backslash \{\bm{0} \} } \int_{ \Omega_{\mathbb{E}_{\mathcal{A}^* }} \backslash \{\bm{0} \} }  I[\bm{e}_{\mathcal{A}^*}: \bm{e}_{\mathcal{A}^*} \in \mathcal{E}^*_{\mathcal{A},\mathcal{A}',\bm{e}_{\mathcal{A}\backslash \mathcal{A}^*}} ]p(\mathbb{E}_{\mathcal{A}^*}=\bm{e}_{\mathcal{A}^*}   ) p(\mathbb{E}_{\mathcal{A}\backslash \mathcal{A}^*}=\bm{e}_{\mathcal{A}\backslash \mathcal{A}^*} ) \ d \bm{e}_{\mathcal{A}^*} d \bm{e}_{\mathcal{A}\backslash \mathcal{A}^* }  \nonumber \\
&= \int_{ \Omega_{\mathbb{E}_{\mathcal{A}\backslash \mathcal{A}* }} \backslash \{\bm{0} \} } \underbrace{\int_{ \Omega_{\mathbb{E}_{\mathcal{A}^* }} \backslash \{\bm{0} \} }  I[\bm{e}_{\mathcal{A}^*}: \bm{e}_{\mathcal{A}^*} \in \mathcal{E}^*_{\mathcal{A},\mathcal{A}',\bm{e}_{\mathcal{A}\backslash \mathcal{A}^*}} ]p(\mathbb{E}_{\mathcal{A}^*}=\bm{e}_{\mathcal{A}^*} ) d \bm{e}_{\mathcal{A}^*} }_{=0} \ p(\mathbb{E}_{\mathcal{A}\backslash \mathcal{A}^*}=\bm{e}_{\mathcal{A}\backslash \mathcal{A}^*} ) \  d \bm{e}_{\mathcal{A}\backslash \mathcal{A}^* } \nonumber \\
&=0.
\end{align}
This leads to,
\begin{align}
\Pr{(\hat{\mathbb{E}} \neq \mathbb{E} )} =0,
\end{align}
using \Cref{eq:clm2_wrt_x}, \Cref{eq:clm2_wrt_e}, \Cref{eq:clm2_sparsity_split} and \Cref{eq:clm2_measure_0} successively.
\end{proof}

\subsection{Comparison with Adversarial Error Model.}
For completion of the discussion on error models and also to facilitate both the proofs of \Cref{thm:error_tolerance_MV,thm:error_tolerance_MM}, we now also formally show that the number of errors that can be corrected using a $(P,Q)$ MDS Code under the adversarial model is $\lfloor \frac{P-Q}{2} \rfloor $. The result is well-known in coding theory \cite{ryan2009channel} for finite fields. This is an extension for real numbers (also see \cite{candes2005decoding}).

\begin{lem}[Real Number Error Correction under Adversarial Model]
\label{lem:errors}
Under the adversarial error model for channel coding, the decoder of a $(P,Q)$ MDS Code can correct at most $\lfloor \frac{P-Q}{2} \rfloor $ errors in the worst case. 
\end{lem}

\textbf{This statement is also equivalent to the following: A real-valued polynomial with $Q$ coefficients can be interpolated correctly from its evaluations at $P$ distinct values, if at most $\lfloor \frac{P-Q}{2} \rfloor $ evaluations are erroneous under the adversarial model. }

\begin{proof}[Proof of \Cref{lem:errors}]

Let $\bm{z}$ denote the vector of length $P$ that denotes the corrupted codeword with additive error $\bm{e}$. Thus,
$$\bm{z} =  \bm{G}^T\bm{q} +  \bm{e},$$
where $\bm{G}$ is the generator matrix of any real number $(P,Q)$ MDS code. The location of errors, \textit{i.e.}, $\mathcal{A}$ represents the non-zero indices of $\bm{e}$.

As the matrix $\bm{G}^T$ is full-rank, there exists a full-rank annihilating matrix or parity check matrix $\bm{H}$ of dimension $(P-Q)\times P $ such that $\bm{H} \bm{G}^T =0$, leading to $\bm{H}\bm{e}= \bm{H} \bm{z} := \widetilde{\bm{z}}.$ Following the steps of \cite{candes2005decoding}, we will show that if $|\mathcal{A}|\leq \lfloor \frac{P-Q}{2}\rfloor$, then the solution $\bm{e}$ to this linear system of equations $\bm{H}\bm{e}=\widetilde{\bm{z}}$ is unique. Therefore, given $\bm{z}$ and $\bm{G}^T$, the vector $\bm{q}$  can also be reconstructed uniquely as $\bm{G}^T$ is full-rank. 


Suppose there exists another $\bm{e}'\neq \bm{e}$ with non-zero indices $\mathcal{A}'$ such that $ |\mathcal{A}'|\leq \lfloor \frac{P-Q}{2}\rfloor$ and $\bm{H}\bm{e}=\bm{H}\bm{e}'$. Then, $\bm{e}-\bm{e}'=\bm{h}$ where $\bm{h}\in Null(\bm{H})\backslash \{ \bm{0} \}$. The number of non-zero elements in $\bm{h}=\bm{e}-\bm{e}'$ satisfies:
\begin{equation}
||\bm{h}||_0 \leq |\mathcal{A} \cup \mathcal{A}'| \leq |\mathcal{A}|+ |\mathcal{A}'| \leq 2\lfloor \frac{P-Q}{2}\rfloor \leq P-Q.
\end{equation}

This is a contradiction. Any $\bm{h}$ lying in $Null(\bm{H})\backslash \{ \bm{0} \}$ has at most $P-Q+1$ non-zero elements. This is because $\bm{H}$ is also the generator of a $(P, P-Q)$ MDS code, and thus any $(P-Q)$ columns are linearly dependent.




\end{proof}
For the special case of polynomial interpolation, the matrix $\bm{G}^T$ becomes a Vandermonde matrix. Consider a polynomial $g(u)=q_0 +q_1 u +  \dots + q_{Q-1}u^{Q-1} $ with $Q$ unique coefficients, evaluated at $P$ distinct values $a_0,a_1,\dots, a_{P-1}  $ respectively. Let these evaluations be represented as $\gamma_0,\gamma_1,\dots,\gamma_{P-1}  $ respectively. Thus, we have,
$$ \begin{bmatrix}   
\gamma_0 \\
\vdots \\
\gamma_{P-1}
\end{bmatrix}= 
\begin{bmatrix}
1 & a_1 & \ldots & a_1^{Q-1}\\
\vdots & \vdots & \ddots & \vdots \\
1 & a_{P-1} & \ldots & a_{P-1}^{Q-1}
\end{bmatrix}  \begin{bmatrix}
q_0\\
\vdots\\
q_{Q-1}
\end{bmatrix}  =  \bm{G}^T_{P \times Q} \bm{q}_{Q\times 1} \ \text{where } \bm{G}^T_{P \times Q}= \begin{bmatrix}
1 & a_1 & \ldots & a_1^{Q-1}\\
\vdots & \vdots & \ddots & \vdots \\
1 & a_{P-1} & \ldots & a_{P-1}^{Q-1}
\end{bmatrix} . $$

\subsection{Decoding Methods (Polynomial Time).}
The realistic decoding technique under either of the error models would be to first compute $\bm{H} \ \bm{z} = \widetilde{\bm{z}}$. If $\widetilde{\bm{z}}=\bm{0}$, one can declare that there are no errors with probability $1$. Otherwise, one can use standard sparse reconstruction algorithms \cite{candes2005decoding,tropp2010computational} to find a solution for the undetermined linear system of equations $\bm{H}\bm{e}= \widetilde{\bm{z}}$ with minimum non-zero elements for $\bm{e}$. Under Error Model $1$, if a solution is found with number of non-zero entries at most $\lfloor \frac{P-Q}{2} \rfloor$, the error is detected and corrected.
Under Error Model $2$, there is no upper bound on the number of errors. If a solution is found with number of non-zero entries at most $P-Q-1$, it is declared as the correct $\bm{e}$. Otherwise, the node declares a decoding failure even though error is detected, and reverts to the last checkpoint.

}}
\section{Proofs of \Cref{thm:error_tolerance_MV,thm:error_tolerance_MM} }
\label{appendix:error_tolerance}
Here we provide details of the decoding technique for both feedforward and backpropagation. As we discussed before, the problem of decoding reduces to the problem of interpolating the coefficients of a polynomial in two variables from its values evaluated at $(a_p,b_p)$ for $p=0,1,\ldots,P-1$, where some of the values might be prone to additive undetected errors (or erasures in the case of stragglers). 

Observe that the total number of unique coefficients that the polynomial $\widetilde{\bm{s}}(u,v)$ has is $ m (2n-1)$ (similarly $\widetilde{\bm{c}}^T(u,v)$ has $n(2m-1)$ unique coefficients). Thus, at first glance it might appear that one would require at the very least  $ m (2n-1) +2t $ (or  $n(2m-1)+2t$ for backpropagation) evaluations to be able to correct any $t$ errors during the matrix-vector product of the feedforward or backpropagation stages respectively (under Error Model $1$).


However, we actually do not require all the unique $ m(2n-1)$ (or  $n(2m-1)$) coefficients. We only need the coefficients of $ u^i v^{n-1}$ (or $ u^{m-1} v^{j}$) which means we effectively require the values of only $m$ (or $n$) coefficients respectively. One technique for solving this interpolation problem is to convert it into a polynomial in a single variable, such that we can reconstruct the required coefficients from its evaluation at $P$ distinct values at  $P$ different nodes, where some evaluations may be erroneous.


\begin{proof}[Proof of \Cref{thm:error_tolerance_MV}] We consider the four cases separately.

$\mathcal{C}_{\mathrm{GP}}(K,N,P)$ with $u=v^{n} : $ Suppose we substitute $u=v^{n}$ and then evaluate the polynomials $\widetilde{\bm{s}}(u,v)$ (or similarly $\widetilde{\bm{c}}^T(u,v)$ at distinct values of $v$ at each of the $P$ nodes. For the feedforward stage, the polynomial $ \widetilde{\bm{s}}(u,v)$ thus becomes only a polynomial of a single variable $v$ as:
\begin{align}
\widetilde{\bm{s}}(v) =  \sum_{i=0}^{m-1}\sum_{j=0}^{n-1} \sum_{j'=0}^{n-1} \bm{W}_{i,j}\bm{x}_{j'} v^{in + n-j'+j-1}.
\end{align}
The degree of this polynomial is $mn + n -2$ which means it has  $mn + n -1$ unique coefficients now. The number of unique coefficients reduce since some of the unique coefficients in $\widetilde{\bm{s}}(u,v)$ now align together with the substitution $u=v^{n}$. However, interestingly, the coefficients of our interest, \textit{i.e.} the coefficients of $u^iv^{n-1}$ in $\widetilde{\bm{s}}(u,v)$  exactly correspond to the coefficients of $v^{in + n-1}$ in $\widetilde{\bm{s}}(v) $.  For interpolation of a polynomial with $mn + n -1$ unique coefficients from its evaluations at $P$ distinct values, one can correct at most $\lfloor \frac{P-mn-n+1}{2} \rfloor $ errors under Error Model $1$ (see \Cref{lem:errors}) and $P-mn-n$ under Error Model $2$ (see \Cref{thm:main}).


For the backpropagation however, none of the coefficients of $\widetilde{\bm{c}}^T(u,v)$ align in $\widetilde{\bm{c}}^T(v) $ with the substitution $u=v^{n}$ and $\widetilde{\bm{c}}^T(v) $ still has $n(2m-1)$ unique coefficients but is now a polynomial of one variable. Thus it can correct at most $\lfloor \frac{P-2mn +n}{2} \rfloor $ errors under Error Model $1$ (see \Cref{lem:errors}) and $P-2mn+n-1$ under Error Model $2$ (see \Cref{thm:main}).\newline

$\mathcal{C}_{\mathrm{GP}}(K,N,P)$ with $v=u^{m} : $
Alternately, using the substitution $ v=u^{m}$, one can increase the number of errors that can be corrected in the backpropagation while reducing that of feedforward. For backpropagation, with the substitution  $v=u^{m}$, the polynomial $\widetilde{\bm{c}}^T(u,v)$ reduces to a variable in one variable $u$ as follows:
\begin{equation}
\widetilde{\bm{c}}^T(u) =  \sum_{i'=0}^{m-1} \sum_{i=0}^{m-1} \sum_{j=0}^{n-1}  \bm{\delta}_{i'}^T\bm{W}_{i,j}  u^{j m+ m-i'+i-1}.
\end{equation}
This is a polynomial of degree $mn + m-2  $, and thus has $ mn + m-1 $ unique coefficients. Thus, during backpropagation, one can correct $\lfloor \frac{P-mn-m+1}{2} \rfloor $ errors under Error Model $1$ and $P-mn-m$ under Error Model $2$. However, for $\widetilde{\bm{s}}(u)$ with the substitution $ v=u^m $, the number of unique coefficients is $m(2n-1)$. Thus, the number of errors that can be corrected in feedforward stage is 
$\lfloor \frac{P-2mn+m}{2} \rfloor $ errors under Error Model $1$ (see \Cref{lem:errors}) and $P-2mn+m-1$ under Error Model $2$ (see \Cref{thm:main}). \newline

$\mathcal{C}_{\mathrm{mds}}(K,N,P):$ The error tolerances are provided in \Cref{lem:mds_error_tolerance}. For the MDS-code-based strategy, all the $P$ nodes are not used in both the feedforward and backpropagation stages. We let $P_f$ and $P_b$ denote the number of nodes in the feedforward and backpropagation stages respectively, of which only $mn$ nodes are common. Thus, $P=P_f+P_b-mn$.

Now, for the feedforward stage, one uses a $(\frac{P_f}{n},m)$ MDS code, and is thus able to correct any $\lfloor \frac{\frac{P_f}{n} -m}{2} \rfloor $ errors under Error Model $1$ and $\frac{P_f}{n} -m-1$ under Error Model $2$. 

And for the backpropagation stage, one uses a $(\frac{P_b}{m},n)$ MDS code. Therefore, one can correct any $\lfloor \frac{\frac{P_b}{m} -n}{2} \rfloor $ errors under Error Model $1$ and $\frac{P_b}{m} -n-1$ under Error Model $2$. \newline

$\mathcal{C}_{\mathrm{rep}}(K,N,P):$ The error tolerances for the replication strategy is provided in \Cref{lem:replication_error_tolerance}. The entire grid of $mn$ nodes is replicated $\frac{P}{mn}$ times, and nodes performing the same computation exchange their outputs for error correction. Under Error Model $1$, one can use a majority decoding and is thus able to correct $\lfloor \frac{\frac{P}{mn}-1}{2}\rfloor$ errors.

Under Error Model $2$, the probability of two outputs having exactly same error is $0$ as the errors are drawn from continuous distributions (recall \Cref{rem:error_model}). Thus, as long as an output occurs at least twice, it is almost surely the correct output. Thus, any $\frac{P}{mn}-2$ errors can be detected and corrected both in the feedforward and backpropagation stages respectively.
\end{proof}

Now we move on to the proof of \Cref{thm:error_tolerance_MM}. 

\begin{proof}[Proof of \Cref{thm:error_tolerance_MM}] The proof is very similar to the previous case. We again consider the four cases separately.\\
\newline

$\mathcal{C}_{\mathrm{GP}}(K,N,P,B)$ with $u=v^n, w=v^{mn}:$ Using this substitution, the polynomial $\widetilde{\bm{S}}(u,v,w)$ reduces to a polynomial of a single variable $v$ as follows:
\begin{align}
\widetilde{\bm{S}}(v)=\widetilde{\bm{S}}(u,v,w)|_{u=v^n, w=v^{mn}}= \sum_{i=0}^{m-1}\sum_{j=0}^{n-1} \sum_{j'=0}^{n-1}\sum_{k=0}^{d_1-1} \bm{W}_{i,j}\bm{X}_{j',k} v^{mnk+ ni+n-1+j-j'}.
\end{align}
The degree of this polynomial is $mnd_1 +n -2 $ which means that it has $mnd_1 +n -1 $ unique coefficients. The problem of interpolation of the coefficients of this polynomial under erroneous evaluations is equivalent to the problem of decoding of a $(P,mnd_1 +n -1)$ MDS code. Thus, one can correct $\lfloor  \frac{P -mnd_1 -n +1}{2}  \rfloor $ errors under Error Model $1$ and $P -mnd_1 -n$ under Error Model $2$. Now let us see what happens for backpropagation.
\begin{align}
\widetilde{\bm{C}}^T(v) =\widetilde{\bm{C}}^T(w,u,v)|_{u=v^n, w=v^{mn}}  = \sum_{k=0}^{d_2-1}\sum_{i'=0}^{m-1}\sum_{i=0}^{m-1}\sum_{j=0}^{n-1} \bm{\Delta}^T_{k,i'}\bm{W}_{i,j} v^{mnk+ n(m-1+i-i')+j}. 
\end{align}
The degree of this polynomial is $mnd_2 +mn -n-1$ which means there are $mnd_2+mn-n$ unique coefficients. The problem of interpolation of the coefficients of this polynomial under erroneous evaluations is equivalent to the problem of decoding of a $(P,mnd_2+mn-n)$ MDS code. Thus, one can correct $\lfloor \frac{P -mnd_2 -mn+n}{2}   \rfloor$ errors under Error Model $1$ and $P -mnd_2 -mn+n-1$ errors under Error Model $2$. \\

$\mathcal{C}_{\mathrm{GP}}(K,N,P,B)$ with $v=u^m, w=v^{mn}:$ Using this substitution, the polynomial $\widetilde{\bm{S}}(u,v,w)$ reduces to a polynomial of a single variable $v$ as follows:
\begin{align}
\widetilde{\bm{S}}(v)=\widetilde{\bm{S}}(u,v,w)|_{v=u^m, w=v^{mn}}= \sum_{i=0}^{m-1}\sum_{j=0}^{n-1} \sum_{j'=0}^{n-1}\sum_{k=0}^{d_1-1} \bm{W}_{i,j}\bm{X}_{j',k} u^{mnk+ m(n-1+j-j')+i  }.
\end{align}
The degree of this polynomial is $mnd_1 +mn -m-1 $ which means that it has $mnd_1 +mn -m $ unique coefficients. Thus, one can correct $\lfloor  \frac{P -mnd_1 -mn +m}{2}  \rfloor $ errors under Error Model $1$ and $P -mnd_1 -mn +m-1$ under Error Model $2$. Now let us see what happens for backpropagation.
\begin{align}
\widetilde{\bm{C}}^T(v) =\widetilde{\bm{C}}^T(w,u,v)|_{v=u^m, w=v^{mn}}  = \sum_{k=0}^{d_2-1}\sum_{i'=0}^{m-1}\sum_{i=0}^{m-1}\sum_{j=0}^{n-1} \bm{\Delta}^T_{k,i'}\bm{W}_{i,j} u^{mnk+mj+m-1+i-i'}. 
\end{align}
The degree of this polynomial is $mnd_2 +n-2$ which means there are $mnd_2+n-1$ unique coefficients. Thus, one can correct $\lfloor \frac{P -mnd_2 -n+1}{2}   \rfloor$ errors under Error Model $1$ and $P -mnd_2 -n$ under Error Model $2$. \\

The expressions for the MDS-based technique and replication follows from the proof of \Cref{thm:error_tolerance_MV}.
\end{proof}

Now, we provide a proof of \Cref{coro:scaling_sense_comparison}.

\begin{proof}[Proof of \Cref{coro:scaling_sense_comparison}]
The proof can be derived by simply substituting $m=n=\sqrt{K}$ in the expressions of $t_f$ (or $t_b$) in \Cref{thm:error_tolerance_MV}. Let us look at the ratio of $t_f$ for 
$\mathcal{C}_{\mathrm{GP}}(K,N,P)$ with $u=v^n$ to $\mathcal{C}_{\mathrm{mds}}(K,N,P)$ under Error Model $1$.

\begin{align}
\text{ Ratio of } t_f = \lim_{P \to \infty} \frac{P-mn-n+1}{\frac{P_f-mn}{n}}\geq \lim_{P \to \infty} \frac{P-mn-n+1}{\frac{P-mn}{n}} = n = \Theta(\sqrt{K}).
\end{align}

Similarly, the ratio of $t_f$ for $\mathcal{C}_{\mathrm{GP}}(K,N,P)$ with $u=v^n$ to $\mathcal{C}_{\mathrm{rep}}(K,N,P)$ can be derived as,

\begin{align}
\text{ Ratio of } t_f = \lim_{P \to \infty} \frac{P-mn-n+1}{\frac{P-mn}{mn}} = \Theta(mn) = \Theta(K).
\end{align}

Similar ratios can also be derived for $t_b$ under Error Model $1$ in the limit of large $P$, as well as for both $t_f$ and $t_b$ under Error Model $2$.

\end{proof}

\section{Complexity Analysis and Proofs of \Cref{thm:complexity_MV} and \Cref{thm:complexity_MM}}
\label{appendix:complexity}

Before proceeding to the proofs, we remind the reader about \Cref{rem:broadcast}. In this work, we assumed that each node can multi-cast simultaneously to at most a constant number (say $2$) of nodes. In order to multi-cast the same $N$ values from one node to $P$ other nodes, the communication complexity is  $\alpha \log_2{P} + \beta N = \Theta(N)$, assuming a node cannot communicate to all the $P$ other nodes simultaneously but uses a spanning-tree type multi-cast mechanism\cite{van1997summa,chan2007collective}. It might also be worth mentioning that when all $P$ nodes have to broadcast their own set of $N$ values to all other nodes except itself, we assume that the communication complexity is $\alpha \log_2{P} + \beta NP = \Theta(NP)$ using an efficient all-to-all communication protocol called All-Gather~\cite{chan2007collective}. Currently, we are examining strategies to reduce this communication cost even further using improved implementation strategies.

For completeness, we include two of the communication protocols here. The reader can refer to \cite{chan2007collective} for more details.

\textbf{Broadcast Communication Protocol:} In a cluster of $P$ nodes, when one node sends a data (say a vector $\bm{x}$) of $N$ values to all other nodes, it is called a broadcast. The communication cost of Broadcast is $\alpha \log{P} +\beta N$.

\textbf{All-Gather Communication Protocol:} In a cluster of $P$ nodes, when every node sends its own, unique data (say vector $\bm{x}_p$ for node index $p=0,1,\ldots,P-1$) of $N$ values to all other nodes, it is called an All-Gather. The communication cost of All-Gather is $\alpha \log{P} +2\beta PN$.

\begin{proof}[Proof of Theorem \ref{thm:complexity_MV}] Let us understand the computational and communication complexity per node for all the steps during DNN training in the decentralized implementation. 
\begin{itemize}
\item Matrix-vector products (steps $O1$ and $O2$):\\
The computational complexity is $\Theta( \frac{N^2}{K})$ for both feedforward (step $O1$) and backpropagation (step $O2$).
\item Multi-casting the partial result to all other nodes for decoding:\\
 Feedforward stage: For node $p=0,1,\ldots,P-1$, a sub-vector $\widetilde{\bm{s}}_p$ of length $\frac{N}{m} $ is to be multi-casted from the $p$-th node to all other nodes. The total communication complexity of this step is $\alpha \log{P} + 2\beta \frac{N}{m} P =  \Theta(\frac{N}{m} P ) $ using All-Gather.\\
 Backpropagation stage: For node $p=0,1,\ldots,P-1$, a sub-vector $\widetilde{\bm{c}}^T_p$ of length $\frac{N}{n} $ is to be multi-casted from the $p$-th node to all other nodes. The communication complexity of this step is $\alpha \log{P} + 2\beta \frac{N}{n} P = \Theta( \frac{N}{n} P )$ using All-Gather.
\item Decoding at every node (Requires solving sparse reconstruction problems \cite{candes2005decoding}):\\ Feedforward stage: $ \mathcal{O}(\frac{N}{m} P^3)$.\\
Backpropagation stage: $\mathcal{O}(\frac{N}{n} P^3)$. \\
Note that, this step is actually the most dominant among the complexities of all the other steps apart from the steps $O1$, $O2$ and $O3$. We are using a pessimistic upper bound for this decoding complexity which can reduce further based on appropriate choice of decoding algorithm.
\item Crosscheck for decoding errors:\\
In this step, every node generates a vector of length $P$ denoting which nodes were found erroneous and then it communicates this vector of $P$ values to every other node for crosscheck. Thus, the complexity is $\alpha \log{P} + 2\beta P^2 =\Theta(P^2 )$ for both feedforward and backpropagation stage, using All-Gather.
\item Post-processing on decoded result: \\ Feedforward stage: Each node generates the vector $\bm{x}$ for next layer by applying nonlinear function $f(\cdot)$ element-wise on vector $\bm{s}$. The complexity is thus $\Theta(N) $.\\
Backpropagation stage: Each node multiplies the resulting vector $\bm{c}^T$ with diagonal matrix $\bm{D}$. The complexity is thus $\Theta(N)$.
\item Encoding of vectors: \\
Feedforward stage: The vector $\bm{x}$ is divided into $n$ equal parts of length $\frac{N}{n}$ and a linear combination of all of them is taken. The computational complexity is thus $\Theta\left(n\frac{N}{n}\right) = \Theta(N)$. \\
Backpropagation stage: The vector $\bm{\delta}^T$ is divided into $m$ equal parts of length $\frac{N}{m}$. The computational complexity is thus $\Theta\left(m\frac{N}{m}\right) = \Theta(N)$.
\item Update stage (step $O3$):\\
It involves a rank-$1$ update to the stored encoded sub-matrix $\widetilde{\bm{W}}_p$ and is thus of computational complexity $\Theta( \frac{N^2}{K})$.
\end{itemize}

Now, we can compute the total computational and communication complexity of all the steps except the individual matrix vector products at each node (step $O1$) for the feedforward stage.
\begin{align}
\text{Total complexity of all additional steps } &= \Theta\left(\frac{N}{m} P  \right) +  \mathcal{O}\left(\frac{N}{m} P^3\right)  + \Theta(P^2 ) + \Theta(N) \nonumber \\ 
& \leq \Theta\left(\frac{N}{m} P^3 \right) =  \Theta \left(\frac{N}{K} n P^3\right) \nonumber \\
& \leq \Theta\left(\frac{N}{K} P^4 \right) = o\left(\frac{N^2}{K}  \right). 
\end{align}
Here the last line follows from the condition of the theorem that $P^4=o(N)$.

Similarly, for the backpropagation stage, the total complexity of all additional steps except the matrix-vector product (step $O2$)  is given by,
\begin{align}
\text{Total complexity of all additional steps } &= \Theta\left(\frac{N}{n} P  \right) +  \mathcal{O}\left(\frac{N}{n} P^3\right)  + \Theta(P^2 ) + \Theta(N) \nonumber \\ 
& \leq \Theta \left(\frac{N}{n} P^3 \right) =  \Theta \left(\frac{N}{K} m P^3\right) \nonumber \\
& \leq \Theta\left(\frac{N}{K} P^4 \right) = o\left(\frac{N^2}{K}  \right).
\end{align}
Here the last line follows from the condition of the theorem that $P^4=o(N)$.

\noindent Thus, it is proved that the individual matrix-vector products (steps $O1$ and $O2$) and the update (step $O3$) are the most computationally intensive steps and the additional steps including encoding/decoding add negligible overhead. 
\end{proof}

\begin{proof}[Proof of Theorem \ref{thm:complexity_MM}] Let us now understand the computational and communication complexity of the additional steps for the case of $B>1$ in the decentralized implementation. Again, we look at the complexities of the various steps as follows:
\begin{itemize}
\item Matrix-matrix products: \\ 
Feedforward stage (step $O1$): The multiplication of sub-matrix $\widetilde{\bm{W}}_p$ of size $\frac{N}{m}\times \frac{N}{n}$ with sub-matrix $\widetilde{\bm{X}}_p$ of size $\frac{N}{n}\times \frac{B}{d_1}$ is performed at each node, which is of computational complexity $\Theta(\frac{N^2B}{Kd_1})$ where $K=mn$. \\
Backpropagation stage (step $O2$): The multiplication of sub-matrix $\widetilde{\bm{\Delta}}^T_p$ of size $\frac{B}{d_2} \times \frac{N}{m}$ with sub-matrix $\widetilde{\bm{W}}_p$ of size $\frac{N}{m}\times \frac{N}{n}$ is performed at each node, which is of computational complexity $\Theta( \frac{N^2B}{Kd_2})$ where $K=mn$.  
\item Broad-casting the partial result to all other nodes for decoding:\\
Feedforward stage: For $p=0,1,\ldots,P-1$, the node with index $p$ multi-casts the resulting sub-matrix $\widetilde{\bm{S}}_p (=\widetilde{\bm{W}}_p\widetilde{\bm{X}}_p) $ of dimension $\frac{N}{m}\times \frac{B}{d_1}$ to all other nodes. Thus, the communication complexity is given by: $ \alpha \log{P} + 2\beta \frac{NB}{md_1} P = \Theta(\frac{NB}{md_1} P )$ using All-Gather.\\
Backpropagation stage: For $p=0,1,\ldots,P-1$, the node with index $p$ multi-casts the resulting sub-matrix $\widetilde{\bm{C}}^T_p (=\widetilde{\bm{\Delta}}^T_p\widetilde{\bm{W}}_p)$ of dimension $\frac{B}{d_2}\times \frac{N}{n}$ to all other nodes. Thus, the communication complexity is given by: $\alpha \log{P} + 2\beta \frac{NB}{nd_2} P = \Theta( \frac{NB}{nd_2} P )$ using All-Gather.
\item Decoding at every node (Requires solving a sparse reconstruction problem \cite{candes2005decoding}):\\
Feedforward stage: $ \mathcal{O}(\frac{NB}{md_1} P^3)$.\\
Backpropagation stage: $\mathcal{O}(\frac{NB}{nd_2} P^3)$. \\
Note that, this step is actually the most dominant among the complexities of all the other steps apart from the matrix-matrix products (steps $O1$ and $O2$) and the update (step $O3$). We are using a pessimistic upper bound for this decoding complexity which can reduce based on appropriate choice of decoding algorithm.
\item Crosscheck for decoding errors:\\
For both feedforward and backpropagation, this step is similar to that for the case of $B=1$. Each node sends a vector of $P$ values to every other node. Thus, the communication complexity is $\alpha  \log{P} + 2\beta P^2= \Theta(P^2)$.
\item Post-processing on decoded result: \\
Feedforward stage: Each node generates the matrix $\bm{X}$ for the next layer by applying nonlinear function $f(\cdot)$ element-wise on the matrix $\bm{S}$ of dimensions $N \times B$. The complexity is thus $\Theta(NB)$.\\
Backpropagation stage: Each node performs the Hadamard product of the matrix $\bm{C}$ with the matrix $f'(\bm{S})$ applied element-wise, where both the matrices are of dimension $N \times B$. The complexity is thus $\Theta(NB)$.
\item Encoding of sub-matrices: \\
Feedforward stage: The matrix $\bm{X}$ is divided into an $n \times d_1$ grid of sub-matrices of dimensions $\frac{N}{n} \times \frac{B}{d_1}$ and a linear combination of all of these $nd_1$ matrices is computed at each node. The computational complexity is thus $\Theta\left(\frac{NB}{nd_1}nd_1\right)= \Theta(NB)$.\\
Backpropagation stage: The matrix $\bm{\Delta}^T$ is divided into a $d_2 \times m$ grid of sub-matrices of dimensions $\frac{B}{d_2} \times \frac{N}{m}$ and a linear combination of all of these $md_2$ matrices is computed at each node. The computational complexity is thus $\Theta\left(\frac{NB}{md_2}md_2\right)= \Theta(NB)$.
\item Update stage (step $O3$):\\ This involves performing $B$ consecutive rank-$1$ updates to the stored, encoded sub-matrix $\widetilde{\bm{W}}_p$. Thus, the complexity is $\Theta( \frac{N^2B}{K})$.
\end{itemize}

Now, we can compute the total complexity of all the steps except the individual matrix-matrix products (step $O1$) at each node for the feedforward stage.
\begin{align}
\text{Total complexity of all additional steps } &= \Theta \left(\frac{NB}{md_1} P  \right) +  \mathcal{O}(\frac{NB}{md_1} P^3)  + \Theta(P^2 ) + \Theta(NB) \nonumber \\ 
& \leq \Theta \left(\frac{NB}{md_1} P^3 \right) =  \Theta \left(\frac{NB}{Kd_1} n P^3\right) \nonumber \\
& \leq \Theta\left(\frac{NB}{Kd_1} P^4 \right) =  \underbrace{o\left(\frac{N^2 B}{Kd_1} \right)}_{\text{Complexity of step $O1$} }.
\end{align}
Here the last line follows from the condition of the theorem that $P^4=o(N)$.
Thus the ratio of the complexity of all the additional steps to that of the matrix-matrix products (step $O1$) tends to $0$ as $K,N$ and $P$ scale.

Similarly, for the backpropagation stage, the total complexity of all additional steps except the matrix-matrix product (step $O2$) is given by,
\begin{align}
\text{Total complexity of all additional steps } &= \Theta \left(\frac{NB}{nd_2} P  \right) +  \mathcal{O}\left(\frac{NB}{nd_2} P^3\right)  + \Theta(P^2) + \Theta(NB) \nonumber \\ 
& \leq \Theta \left(\frac{NB}{nd_2} P^3 \right) =  \Theta \left(\frac{NB}{Kd_2} m P^3\right) \nonumber \\
& \leq \Theta\left(\frac{NB}{Kd_2} P^4 \right) =  \underbrace{o\left(\frac{N^2 B}{Kd_2} \right)}_{\text{Complexity of step $O2$}}. 
\end{align}
Here the last line follows from the condition of the theorem that $P^4=o(N)$. Thus the ratio of the complexity of all the additional steps to that of the matrix-matrix products (step $O2$) tends to $0$ as $K,N$ and $P$ scale.

\noindent Thus, it is proved that the individual matrix-matrix products and the update are the most computationally intensive steps and the additional steps including encoding or decoding add negligible overhead. 
\end{proof}

\end{document}